%% file: MMtetra.tex
\documentclass[a4paper,envcountsame,11pt]{easychair}
\usepackage{amsthm}       %%% for ams theorem environment.
\usepackage{amsmath}
\usepackage{amssymb}      %%% necessary for \thickapproach (under others...)
\usepackage{graphicx}     %%% shoul
\usepackage[pdftex]{xcolor}
\usepackage[long,raise]{datetime}
\usepackage{multirow}
\usepackage{mathrsfs}
\usepackage{rotating}
\usepackage{varioref}
\usepackage[caption=false,font=footnotesize]{subfig}
\usepackage{enumerate}
\usepackage{cite}
\def\romannum{\begingroup
  \def\theenumi{\textup{(\roman{enumi})}}%
  \def\p@enumi{}%
  \def\labelenumi{\theenumi}%
  \enumerate}

\newenvironment{longenum}[0]{
  \begin{enumerate}}%
  {\end{enumerate}}%
\newenvironment{longromannum}[0]{
  \begin{enumerate}[(i)]%\tightlist
  }%
  {\end{enumerate}}
\usepackage{tikz}
\usetikzlibrary{shapes,arrows,fit,decorations.pathmorphing,calc}

\newtheorem{theorem}{Theorem}%[chapter]
\newtheorem{lemma}[theorem]{Lemma}
\newtheorem{corollary}[theorem]{Corollary}
\newtheorem{proposition}[theorem]{Proposition} 
\newtheorem{definition}[theorem]{Definition}

\theoremstyle{example}
\newtheorem{remark}{Remark}%[section]
\newtheorem{example}{Example}

\include{MMmacro}

\begin{document}

\title{On the complexity of the model checking problem}
\titlerunning{On the complexity of the model checking problem}
\author{Florent R. Madelaine 
  \thanks{This author is thankfull to the CNRS for supporting his
    one year research leave at the \textsl{Laboratoire d'Informatique de l'\'Ecole Polytechnique}.}
  \\
  Universit{\'e} d'Auvergne\\
  \url{florent.madelaine@udamail.fr}\\
  \and
  Barnaby D. Martin 
  \thanks{This work was supported by EPSRC under grant EP/G020604/1
    while this author was based at Durham University.}
  \\
  Laboratoire d'Informatique de l'\'Ecole Polytechnique\\
  \url{barnabymartin@gmail.com}\\
}
\authorrunning{F.R. Madelaine and B.D. Martin}
\maketitle

\begin{abstract}
The model checking problem for various fragments of first-order logic
has attracted much attention over the last two decades: in particular,
for the fragment induced by $\exists$ and $\land$ and that 
induced by $\forall, \exists$ and $\land$, which are better known as 
the \textsl{constraint satisfaction problem} and the \emph{quantified constraint 
satisfaction problem}, respectively. 
These two fragments are in fact the only ones for which there is
currently no known complexity classification. All other
syntactic fragments can be easily classified, either directly or
using Schaefer's dichotomy theorems for SAT and QSAT, with the
exception of the \emph{positive equality free fragment} induced by
$\exists,\forall, \land$ and $\lor$. This
outstanding fragment can also be classified and enjoys a
\emph{tetrachotomy}: according to the model, the corresponding model checking
problem is either tractable, \NP-complete, \coNP-complete or
Pspace-complete. Moreover, the complexity drop is always witnessed by a generic
solving algorithm which uses quantifier relativisation (for example,
in the \coNP-complete case, the model has a constant $e$ to which all $\exists$
quantifiers may be relativised).
Furthermore, its complexity is characterised by algebraic means: the
presence or absence of specific \emph{surjective hyper-operations}
among those that preserve the model characterise the complexity.
Our classification methodology relies on this suitably tailored algebraic
approach and it suffices to classify the complexity of a finite number of cases for each model size: each case
corresponds to an element of the \textbf{finite} lattice of \emph{down-closed monoids of surjective hyper operations}.
This is unlike the constraint satisfaction problem where the
corresponding lattices are uncountable and essentially uncharted in general. 
Though finite, the number of elements of the lattice of down-closed monoids of
surjective hyper operations grows rapidly as the size $n$ of the model
increases. We are able to compute suitable parts of this lattice by
hand for $n=2$ and $3$, and in a computer assisted manner for $n=4$.
For arbitrarily large $n$, one can restrict the classification to specific
monoids of surjective hyper-operations which corresponds to certain cores, for 
a suitable notion of core for the positive equality-free fragment of first
order logic.
% : that is of a minimal model agreeing on all sentences of
% this fragment; which can be characterised logically via relativisation
% of quantifiers; and, algebraically via the presence of suitable surjective
% hyper-operations. 
These specific monoids enjoy a nice normal form which means that we
are able to provide generic hardness proofs which mimic cases encountered
when $n\leq4$.
\end{abstract}

%%%%%%%%%%%%%%%%%%%%%%%%%%%%%%%%%%%%%%%%%%%%%%%%%%%%%%%%%%%%%%%%%%%%%%%%%%%

%\renewcommand{\thefootnote}{\fnsymbol{footnote}}
%\renewcommand{\thefootnote}{\arabic{footnote}}

%%%%%%%%%%%%%%%%%%%%%%%%%%%%%%%%%%%%%%%%%%%%%%%%%%%%%%%%%%%%%%%%%%%%%%%%%%%%%%%%%%%%%%%%%%%%%%%%%%%%%%%
\section{Introduction}

The \emph{model checking problem} over a logic $\mathscr{L}$ takes as
input a structure $\mathcal{D}$ and a sentence $\varphi$ of
$\mathscr{L}$, and asks whether $\mathcal{D} \models \varphi$. The
problem can also be parameterised, either by the sentence $\varphi$,
in which case the input is simply $\mathcal{D}$, or by the model
$\mathcal{D}$, in which case the input is simply $\varphi$. Vardi has
studied the complexity of this problem, principly for logics which
subsume $\FO$~\cite{DBLP:conf/stoc/Vardi82}. He describes the complexity
of the unrestricted problem as the \emph{combined complexity}, and the
complexity of the parameterisation by the sentence (respectively,
model) as the \emph{data complexity} (respectively, \emph{expression
  complexity}). For the majority of his logics, the expression and
combined complexities are comparable, and are one exponential higher
than the data complexity.

In this paper, we will be interested in taking syntactic fragments
$\mathscr{L}$ of \FO, induced by the presence or absence of
quantifiers and connectives, and studying the complexities of the parameterisation of the model checking problem by the model $\mathcal{D}$, that is the expression complexities for certain $\mathcal{D}$.
When $\mathscr{L}$ is the \emph{primitive positive}
fragment of \FO, \csplogic, the model checking problem is equivalent
to the much-studied \emph{constraint satisfaction problem} (\CSP). The
parameterisation of this problem by the model $\mathcal{D}$ is
equivalent to what is sometimes described as the \emph{non-uniform}
constraint satisfaction problem,
$\CSP(\mathcal{D})$~\cite{DBLP:journals/jcss/KolaitisV00}. It has been  conjectured
\cite{DBLP:conf/icalp/BulatovKJ00,DBLP:journals/siamcomp/FederV98} that the class of $\CSP$s
exhibits \emph{dichotomy} -- that is, $\CSP(\mathcal{D})$ is always
either in \Ptime\ or is \NP-complete, depending on the model
$\mathcal{D}$. This is tantamount to the condition that the expression
complexity for $\{ \wedge, \exists \}$-\FO\ on $\mathcal{D}$ is always
either in \Ptime\ or is \NP-complete. While in general this conjecture
remains open, it has been proved for substantial classes
and various methods, combinatorial (graph-theoretic), logical and universal-algebraic have
been brought to bear on this classification project, with many
remarkable consequences. Schaefer was a precursor and provided a
dichotomy for Boolean structures using a logico-combinatorial approach~\cite{Schaefer}.
Further dichotomies were obtained: \textsl{e.g.} for
structures of size at most three~\cite{DBLP:journals/jacm/Bulatov06}, for
undirected graphs~\cite{HellNesetril}, smooth 
digraphs~\cite{DBLP:journals/siamcomp/BartoKN09}.
A conjectured delineation for the dichotomy
was given in the algebraic language in~\cite{DBLP:journals/siamcomp/BulatovJK05}.

When $\mathscr{L}$ is \emph{positive Horn}, $\qcsplogic$, the model
checking problem is equivalent to the well-studied \emph{quantified
  constraint satisfaction problem} (\QCSP). No overarching 
polychotomy has been conjectured for the non-uniform $\QCSP(\mathcal{D})$, although the
only known attainable complexities are \Ptime, \NP-complete and
\Pspace-complete. Schaefer announced a dichotomy in the Boolean
case~\cite{Schaefer} between \Ptime{} and \Pspace-complete in the
presence of constants, a dichotomy which was proved to hold even when
constants are not present~\cite{DalmauLSI9743R,Nadia}. Some
partial classification were obtained,
algebraically~\cite{DBLP:conf/cp/Chen04,DBLP:journals/siamcomp/Chen08,DBLP:journals/iandc/BornerBCJK09}
or combinatorially~\cite{DBLP:conf/cie/MartinM06,DBLP:conf/cp/Martin11}. A
conjecture delineating the border between \NP{} and \Pspace-complete
was recently ventured by Chen in the algebraic language for structures
with all constants~\cite{DBLP:journals/corr/abs-1201-6306}.

Owing to the natural duality between $\exists, \lor$ and $\forall,
\land$, we  consider also various dual fragments. For example, the
dual of  $\{\exists,\land \}$-\FO\ is \emph{positive universal
  disjunctive} \FO, $\{ \forall, \lor \}$-\FO. It is straightforward
to see that this class of expression complexities exhibits dichotomy
between \Ptime\ and \coNP-complete if, and only if, the class of \CSP
s exhibits dichotomy between \Ptime\ and \NP-complete. 
Table~\ref{tab:ModelChecking} summarises known results regarding the
complexity of the model checking for syntactic fragments of
first-order logic, up to this duality.

In the case of primitive positive logic, it makes little
difference whether or not equality is allowed, that is the expression
complexities for  $\{ \exists,\land \}$-\FO\ and $\{ \exists, \land, =
\}$-\FO\ are equivalent. This is because equality may be propagated out in all but trivial instances. The same is true for positive Horn logic, but is not true, e.g., for positive
universal disjunctive \FO. Indeed, a classification of the expression
complexities over $\{\forall,\lor \}$-\FO\ is equivalent to the
unproven \CSP\ dichotomy conjecture, though we are able to give a full dichotomy for the expression complexities over $\{ \forall, \lor,=
\}$-\FO. The reason for this is that the equality relation in the
latter simulates a disequality relation in the former. If the
model $\mathcal{D}$ has $k\geq 3$ elements then $\{ \exists, \land, \neq\}$-\FO{}
can simulate $k$-colourability; and, otherwise we have a Boolean model
and Schaefer's dichotomy theorem provides the classification.
A similar phenomenon occurs at a higher level when $\forall$ is also present.

Other fragments can be easily classified, as the model checking
problem is always hard except for pathological and rather trivial models, with the notable exception of 
\emph{positive equality-free first-order logic} $\mylogic$.
For this outstanding fragment, the corresponding model checking
problem can be seen as an extension of $\QCSP$ in which disjunction is
returned to the mix. Note that the absence of equality is here
important, as there is no general method for its being propagated out
by substitution. Indeed, we will see that evaluating the related
fragment $\{ \exists, \forall, \wedge, \vee, = \}$-\FO\ is
\Pspace-complete on any structure $\mathcal{D}$ of size at least two.  
\begin{table}[h]
\centering
% \tbl{Complexity of the model checking
%   according to the model for syntactic fragments of \FO (\Logspace{}
%   stands for logarithmic space, \Ptime{} for polynomial time,
%   \NP{} for non-deterministic polynomial time, \coNP{} for its dual and
%   \Pspace{} for polynomial space).}{%
 \begin{tabular}{|p{.15\textwidth}|p{.15\textwidth}|p{.6\textwidth}|} 
    % \hline
    % \multicolumn{3}{|c|}{Polychotomies for the expression complexity
    %   of the model checking problem}\\
    % \hline
    \hline
    \multicolumn{1}{|c|}{Fragment} & \multicolumn{1}{|c|}{Dual} & \multicolumn{1}{|c|}{Classification?} \\
    \hline
    $\{ \exists, \vee\}$            & 
    $\{ \forall, \wedge\}$          &\multirow{2}{.6\textwidth}{Trivial (in \Logspace{}).} \\
    $\{ \exists, \vee, = \}$        & 
    $\{ \forall, \wedge, \neq \}$   &
    \\
    \hline
    $\{ \exists, \wedge, \vee \}$ & 
    $\{ \forall, \wedge, \vee \}$ & 
    \multirow{2}{.6\textwidth}{Trivial (in \Logspace) if the core of
      $\mathcal{D}$ has one element and \NP-complete otherwise.}
    \\
    $\{ \exists, \wedge, \vee, = \}$    & 
    $\{ \forall, \wedge, \vee, \neq \}$ & 
    \\
    \hline
    $\{ \exists, \wedge, \vee, \neq \}$    & 
    $\{ \forall, \wedge, \vee, = \}$ & 
    Trivial (in \Logspace) if $|D|=1$ and \NP-complete otherwise.
    \\
    \hline
    $\{ \exists, \wedge\}$     & 
    $\{ \forall, \vee\}$       &
    \multirow{2}{.6\textwidth}{\textbf{\CSP{} dichotomy} conjecture: \Ptime{} or
      \NP-complete.}
    \\
    $\{ \exists, \wedge, = \}$ & 
    $\{ \forall, \vee, \neq \}$&
    \\
    \hline
    \multirow{2}{*}{$\{ \exists, \wedge, \neq \}$} & 
      \multirow{2}{*}{$\{ \forall, \vee, = \}$}      &
    Trivial if $|D|=1$; in \Ptime\ if $|D|=2$ and $\mathcal{D}$ is affine or
    \\
    &&    bijunctive; and, \NP-complete 
    otherwise.\\
    \hline
    $\{ \exists, \forall, \wedge\}$     & 
    $\{ \exists, \forall, \vee\}$       &
      \multirow{2}{.6\textwidth}{a \textbf{\QCSP{} trichotomy} should be conjectured: \Ptime, \NP-complete, or \Pspace-complete.}\\
    $\{ \exists, \forall, \wedge, = \}$ & 
    $\{ \exists, \forall, \vee,  \neq\}$    &
    \\
    \hline
    \multirow{2}{*}{$\{ \exists, \forall, \wedge, \neq \}$} & 
    \multirow{2}{*}{$\{ \exists, \forall, \vee,  =\}$}    &
    Trivial if $|D|=1$; in \Ptime\ if $|D|=2$ and $\mathcal{D}$ is affine or
    \\
    &&    bijunctive; and, \Pspace-complete 
    otherwise.\\
    \hline
    \multicolumn{2}{|c|}{ \multirow{2}{*}{$\{ \forall, \exists, \wedge, \vee\}$}} &
    \textbf{Positive equality free tetrachotomy}: \Ptime,
    \NP-complete, \coNP-complete or \Pspace-complete\\
    \hline
    \multicolumn{2}{|c|}{\multirow{2}{*}{$\{ \neg,\exists, \forall,
        \wedge, \vee \}$}} &  
    Trivial when $\mathcal{D}$ contains only trivial relations (empty
    or all tuples, and
    \Pspace-complete otherwise.\\ 
    \hline
    \multicolumn{2}{|l|}{
      $\{ \forall, \exists, \wedge, \vee, =\}$  
      $\{ \forall, \exists, \wedge, \vee, \neq \}$
    }%  
    &
    \multirow{2}{.6\textwidth}{
      Trivial when $|D|=1$, \Pspace-complete otherwise.
    }
    \\
    \multicolumn{2}{|c|}{$\{ \neg,\exists, \forall, \wedge, \vee, = \}$} &  \\
    \hline
  \end{tabular} 
%}
\caption{Complexity of the model checking
  according to the model for syntactic fragments of \FO\ (\Logspace{}
  stands for logarithmic space, \Ptime{} for polynomial time,
  \NP{} for non-deterministic polynomial time, \coNP{} for its dual and
  \Pspace{} for polynomial space).}
\label{tab:ModelChecking}
\end{table}

The case of $\mylogic$ is considerably richer than all other cases (as seen on
Table~\ref{tab:ModelChecking}) -- with the exception of $\csplogic$ and
$\qcsplogic$ which are still open and active fields of research -- and
is the main contribution of this paper. 
We undertook the study of the complexity of the model checking of
$\mylogic$ through the algebraic method that has been so fruitful in
the study of the \CSP\ and
\QCSP~\cite{Schaefer,DBLP:journals/jacm/JeavonsCG97,DBLP:journals/jacm/Bulatov06,DBLP:journals/iandc/BornerBCJK09,DBLP:journals/siamcomp/Chen08}. To
this end, we defined \emph{surjective hyper-endomorphisms} and used them
to define a Galois connection that characterises definability under
$\mylogic$ and prove that it suffices to study the complexity of
problems associated with the closed sets of the associated lattice, the so-called
\emph{down-closed monoids of unary surjective hyper-operations} (DSM
for short)~\cite{DBLP:journals/tocl/MadelaineM12}.
Unlike the case of \CSP{} where the corresponding lattice, the
so-called \emph{clone lattice}, is infinite and essentially uncharted
when the domain size exceeds two, our lattice of DSMs is finite for any fixed
domain. This has meant that we were able to compute the lattice for modest domain sizes, or charter
parts relevant to our classification project, whether by
hand for a domain of up to three
elements~\cite{DBLP:conf/lics/MadelaineM09}, or using a computer for
up to four elements~\cite{DBLP:conf/csl/MartinM10}. 
These papers culminate in a full classification -- a tetrachotomy --
as $\mathcal{D}$ ranges over structures  with up to four elements
domains. Specifically, the problems $\mylogic(\mathcal{D})$ are either
in \Logspace, are \NP-complete, are \coNP-complete or are
\Pspace-complete. It is a pleasing consequence of our algebraic
approach that we can give a quite simple explanation to the delineation
of our subclasses. A drop in complexity arises precisely when we may
relativise w.l.o.g. all quantifiers of one type to a single domain element: for
example, all existential quantifiers may be fixed to a special domain
element $e$, resulting in a natural complexity drop from \Pspace{} to
\coNP. Moreover, for membership of \Logspace, \NP{} and \coNP, it is proved
in~\cite{DBLP:journals/tocl/MadelaineM12} that it is sufficient that $\mathcal{D}$ has certain special
surjective hyper-endomorphisms. For our previous example, we would
have a surjective hyper-operation, i.e. a function $f$ from $D$ to the
power set of $D$, such that $e \in f(d)$ for any element $d$ of $D$, which is a surjective
hyper-endomorphism of $\mathcal{D}$. 
Intuitively, a winning strategy for the existential player for some input sentence $\varphi$
may be transformed through ``application of $f$'' into a winning strategy where any
existential variable is played on the constant $e$. 
The converse, that it is
necessary to have these special surjective hyper-endomorphisms, is
more subtle and was
initially only an indirect consequence of our exploration of the
lattice of DSMs. We settled this converse direction and the
tetrachotomy for any domain size via the
introduction of the novel notion of $U$-$X$-core~\cite{DBLP:conf/lics/MadelaineM11} which is the
analog for $\mylogic$ of the core, so useful in the case of
$\csplogic$ and \CSP.

The well-known notion of the core of $\mathcal{D}$ may be seen as the
minimal induced substructure $\mathcal{D}' \subseteq \mathcal{D}$ such
that $\mathcal{D}'$ and $\mathcal{D}$ agree on all primitive positive
sentences. Equivalently, the domain $D'$ of $\mathcal{D}'$ is minimal
such that any primitive positive sentence is true on $\mathcal{D}$ iff
it is true on $\mathcal{D}$ with all (existential) quantifiers
relativised to $D'$. Cores are minimal structures in their equivalence
classes, given by the equivalence relation of satisfying the same
primitive positive sentences. Cores are very robust, for instance,
being unique up to isomorphism, and sitting as induced substructures
in all other structures in their equivalence class. A similar notion
to core exists for the QCSP, but it is not nearly so robust (they need
no longer be uniquely minimal in size nor sit as an induced substructure in other
structures in their equivalence
class~\cite{DBLP:journals/corr/abs-1204-5981}). For the problems 
$\mylogic(\mathcal{D})$, a notion of core returns, and it is once
again robust. The $U$-$X$-core of $\mathcal{D}$ consists of a minimal
substructure induced by the union $U \cup X$ of two minimal sets $U$
and $X$ of $D$ such that a positive equality-free sentence is true on
$\mathcal{D}$ iff it is true on $\mathcal{D}$ with the universal
quantifiers relativised to $U$ and the existential quantifiers
relativised to $X$. Analysing $U$-$X$-cores gives us the necessary
converse alluded to in the previous paragraph. In the \Pspace-complete
case, some completion of the $U$-$X$-core is either fundamentally very
simple and can be classified as in a two-element domain, known
from~\cite{DBLP:conf/lics/MadelaineM09}, or it is a generalisation of
one of the four-element cases from~\cite{DBLP:conf/csl/MartinM10}. For
the \NP-complete and \coNP-complete cases, some completion of the
$U$-$X$-core is fundamentally very simple and can be classified as an
easy generalisation of a three-element domain. 

We are able therefore to give the delineation of our tetrachotomy by
two equivalent means. Firstly, by the presence or absence of certain
special surjective hyper-endomorphisms, the so-called \emph{A-shops}
and \emph{E-shops} (in our running example above $f$ is an E-shop, its
dual i.e. a surjective hyper-operation $g$ such that there exists a
constant $u$ such that $g(u)=D$ would be an A-shop). Secondly, by the existence or not of trivial
sets for the relativisation of universal and existential
quantifiers (see Table~\ref{table_reformulatingTetrachotomy}).
Thus, $\mylogic(\mathcal{D})$ is in \Logspace{} iff
$\mathcal{D}$ has both an A-shop and an E-shop for surjective
hyper-endomorphism, iff there exist singleton sets $U$ and $X$ such
that a sentence of positive equality-free logic is true on
$\mathcal{D}$ exactly when it is true on $\mathcal{D}$ with the
universal quantifiers and existential quantifiers relativised to $U$
and $X$, respectively. Otherwise, and in a similar vein,
$\mylogic(\mathcal{D})$ is NP-complete (resp., co-NP-complete) if it
has an A-shop (resp., E-shop) for a surjective hyper-endomorphism, iff there exists a singleton
set $U$ (resp., $X$) such that a sentence of positive equality-free
logic is true on $\mathcal{D}$ exactly when it is true on
$\mathcal{D}$ with the universal quantifiers relativised to $U$
(resp., the existential quantifiers relativised to $X$). In all
remaining cases, $\mylogic(\mathcal{D})$ is Pspace-complete, and
$\mathcal{D}$ has neither an A-shop nor an E-shop as a surjective hyper-endomorphism, and there
are no trivial sets $U$ nor $X$ affording the required relativisation
properties. 

\begin{table}[h]
% \tbl{Reformulations of the tetrachotomy ($U$ and $X$ denote the subsets of
%   the domain to which universal and existential variables relativise,
%   respectively; the relativisation into a weaker logical fragment
%   allows up to two constants).}{%
\resizebox{\textwidth}{!}{
  \begin{tabular}{|c|c|c|c|c|r|c|}
    \hline
    \multicolumn{7}{|c|}{Tetrachotomy for \mylogic$(\mathcal{D})$}\\
    \hline
    \hline
    Case & Complexity & $A$-shop & $E$-shop & $U$-$X$-core & Relativises into & Dual \\
    \hline
    I & \Logspace & yes & yes & $|U|=1$, $|X|=1$ &
    $\{\wedge,\vee\}$-\FO & I\\
    \hline
    II &\NP-complete & yes & no  & $|U|=1$, $|X|\geq 2$ &
    $\{\exists, \wedge,\vee\}$-\FO & III\\
    \hline
    III &\coNP-complete & no  & yes & $|U|\geq 2$, $|X|=1$ &
    \ensuremath{\{\forall, \vee, \wedge\} \mbox{-}\FO} & II \\ 
    \hline
    IV & \Pspace-complete & no & no & $|U|\geq 2$,
    $|X|\geq 2$ & \mylogic & IV\\
    \hline
  \end{tabular}
}%
%}%
\caption{Reformulations of the tetrachotomy ($U$ and $X$ denote the subsets of
  the domain to which universal and existential variables relativise,
  respectively; the relativisation into a weaker logical fragment
  allows up to two constants).}
\label{table_reformulatingTetrachotomy}
\end{table}

The paper is organised as follows. 
In Section~\ref{sec:compl-model-check}, we prove preliminary results,
in Section~\ref{sec:compl-class}, we classify the complexity of the
model checking problem for all
fragments (other than those corresponding to \CSP\ and \QCSP) but
\mylogic. In Section~\ref{sec:tetrachotomy}, we classify the complexity of
the model checking for \mylogic. 

In more detail, in \S~\ref{sec:methodology}, we present our methodology to tackle
systematically the complexity of the 
model checking problem, discuss duality and the fragments it will
suffice to classify. In \S~\ref{sec:containmentCore}, we recall the
notion of containment, equivalence and cores and extend it abstractly
to any fragment $\mathscr{L}$. 
%In \S~\ref{sec:hintikka-games}, we recall Hintikka Games. 
In \S~\ref{sec:hyper-morphisms}, we introduce the notion of
hyper-operations and hyper-morphisms which arise naturally in the
context of equality-free fragments. In \S~\ref{sec:cont-core-eqfr}, we 
investigate containment, equivalence and core for \EqFreeFo. In
\S~\ref{sec:containment-mylogic}, we characterise containment for
\mylogic; in \S~\ref{sec:cores}, we introduce the notion of
a $U$-$X$-core; highlight the link with relativisation in
\S~\ref{sec:relativisation}; prove some basic properties of surjective
hyper-endomorphisms of a $U$-$X$-core in \S~\ref{sec:u-x-core}; and,
shows that it is unique up to isomorphism in \S~\ref{sec:uniqueness-u-x}.

In \S~\ref{sec:compl-class}, we start by recalling Schaefer's theorem
for Boolean \CSP, and its analog for Boolean \QCSP. In
\S~\ref{sec:first-class}, we deal with trivial fragments whose model
checking problem is always in \Logspace. In~\ref{sec:second-class}, we
turn our attention to fragments $\mathscr{L}$ whose complexity is trivial
if the structure has a one element $\mathscr{L}$-core and is hard
otherwise. In particular, we recall a Galois connection using hyper-endomorphisms
to classify the fragment \cspDisj\ following the guidelines given
by Börner~\cite{DBLP:conf/dagstuhl/Borner08} regarding Galois connections.
In~\ref{sec:third-class}, we classify the fragments whose complexity
can be deduced from the Boolean case.

In Section~\ref{tetrachotomy}, we deal with \mylogic. In
\S~\ref{sec:GaloisConnectionInvShe}, we recall the Galois
connection using surjective
hyper-endomorphisms. In \S~\ref{sec:boolean-case}, we recall how the
Boolean case can be classified using the lattice associated with this
Galois connection. In general, the upper bound of our tetrachotomy is a direct
consequence of the characterisation of $U$-$X$-core in terms of
relativisation, and we only need to deal with the lower bounds which
we do in full generality in~\S~\ref{sec:hardness}.
In particular, in \S~\ref{sec:char-reduc-dsms}, we characterise in
some detail the DSM of a $U$-$X$-core, showing that it is of a very
restricted form, which allows us to prove hardness in a generic way in
subsequent sections.
Finally, In \S~\ref{sec:meta-problem} we investigate the complexity of the meta-problem: given
a finite structure $\mathcal{D}$, what is the complexity of
evaluating positive equality-free sentences of \FO{} over
$\mathcal{D}$? We establish that the meta-problem is NP-hard, even  for a fixed and finite signature.

The present paper represents the full version of the conference
reports~\cite{DBLP:conf/cie/Martin08,DBLP:conf/csl/MartinM10,DBLP:conf/lics/MadelaineM11}
and a part of~\cite{DBLP:journals/corr/abs-1204-5981};
\cite{DBLP:conf/lics/MadelaineM11} itself supersedes a
series of papers begun with~\cite{DBLP:conf/lics/MadelaineM09}.
%, whose journal version is~\cite{DBLP:journals/tocl/MadelaineM12}.
Unless otherwise stated, all results appear here for the first time (outside of conference publications).

\section{Preliminaries}
\label{sec:compl-model-check}

\subsection{Basic Definitions}
Unless otherwise stated, we shall work with
finite relational structures that share the same finite relational
signature $\sigma$.  Let $\mathcal{D}$ be such a structure. We will
denote its domain by $D$. We denote the size of such a set $D$ by
$|D|$. 
The \emph{complement} $\overline{\mathcal{D}}$ of a structure $\mathcal{D}$ consists of relations that are exactly the set-theoretic complements of those in $\mathcal{D}$. I.e., for an $a$-ary $R$, $R^{\overline{\mathcal{D}}}:=D^a \setminus R^{\mathcal{D}}$. For graphs this leads to a slightly non-standard notion of complement, as it includes self-loops.

A \emph{homomorphism} (resp. \emph{full homomorphism}) from a
structure $\mathcal{D}$ to a structure $\mathcal{E}$ is a function
$h:D\rightarrow E$ that preserves (resp. preserves fully) the
relations of $\mathcal{D}$, i.e. for all $a_i$-ary relations $R_i$, and 
for all $x_1,\ldots,x_{a_i} \in D$, 
$R_i(x_1,\ldots,x_{a_i})\in \mathcal{D}$ implies $R_i\bigl(h(x_1),\ldots,h(x_{a_i})\bigr) \in \mathcal{E}$ (resp. $R_i(x_1,\ldots,x_{a_i})\in \mathcal{D}$ iff $R_i\bigl(h(x_1),\ldots,h(x_{a_i})\bigr) \in \mathcal{E}$). $\mathcal{D}$ and $\mathcal{E}$ are \emph{homomorphically equivalent} if there are homomorphisms both from $\mathcal{D}$ to $\mathcal{E}$ and from $\mathcal{E}$ to $\mathcal{D}$.

%%% model checking problem
Let $\mathscr{L}$ be a fragment of $\FO$. Let $\mathcal{D}$ be a fixed
structure. The decision problem $\mathscr{L}(\mathcal{D})$ has:
\begin{itemize}
\item Input: a sentence $\varphi$ of $\mathscr{L}$.
\item Question: does $\mathcal{D} \models \varphi$?
\end{itemize}

\subsection{Methodology}
\label{sec:methodology}
In this paper, we will be concerned with syntactic fragments $\mathscr{L}$ of \FO\ defined by allowing or disallowing
symbols from $\{\exists, \forall, \land, \lor, \neq, =,\neg\}$. Given any sentence $\varphi$ in $\mathscr{L}$, we
may compute in logarithmic space an equivalent sentence $\varphi'$ in
prenex normal form, with negation pushed inwards at the atomic level. Since we will not be concerned with complexities
beneath \Logspace, we assume hereafter that all inputs are in this form. 

In general \Pspace{} membership of $\FO(\mathcal{D})$ follows by a simple evaluation procedure inward
through the quantifiers. Similarly, the expression complexity of the existential fragment 
$\{\exists, \land, \lor, \neq,= \} \mbox{-}\mathrm{FO}$ is at most \NP{};
and, that of its dual fragment $\{\forall, \lor, \land, =, \neq\}
\mbox{-}\mathrm{FO}$ is at most \coNP{} (in both cases, we may even
allow atomic negation)~\cite{DBLP:conf/stoc/Vardi82}. 
We introduce formally below this principle of duality. 

Let $\mathscr{L}$ be a syntactic fragment of \FO\ defined by allowing or disallowing
symbols from $\{\exists, \forall, \land, \lor, \neq, =\}$.
We denote by $\overline{\mathscr{L}}$ its dual fragment by de Morgan's
law:
%$\lnot$ is self-dual, %%% removed negation to simplify exposition.
$\land$ is dual to $\lor$, $\exists$ to $\forall$ and $=$ to $\neq$.
% We denote by $\overline{\mathcal{D}}$ the complement structure of
% $\mathcal{D}$: it is the structure with the same domain where for
% every $r$-ary relational symbol $R$ in the signature
% $R^{\overline{\mathcal{D}}}:=D^r\setminus R^{\mathcal{D}}$.
\begin{proposition}%[\textbf{principle of duality}]
\label{prop:duality:principle}
  Let $\mathscr{L}$ be a syntactic fragment of \FO\ defined by allowing or disallowing
symbols from $\{\exists, \forall, \land, \lor, \neq, =\}$.
  The problem $\mathscr{L}(\mathcal{D})$ belongs to a complexity class
  $\mathrm{C}$ if, and only if, the problem
  $\overline{\mathscr{L}}(\overline{\mathcal{D}})$ belongs to the dual
  complexity class 
  $\mathrm{co\mbox{-}C}$.
\end{proposition}
\begin{proof}
  For any sentence $\varphi$ in $\mathscr{L}$, we may rewrite its
  negation $\lnot \varphi$ by pushing the negation inwards until all
  atoms appear negatively, denoting the sentence hence obtained by
  $\psi$ (which is logically equivalent to $\lnot \varphi$). Next, we replace every occurrence of a negated relational
  symbol $\lnot R$ by $R$ to obtain a
  sentence of $\overline{\mathscr{L}}$ which we denote by $\overline{\varphi}$. 
  The following chain of equivalences holds
  $$\mathcal{D}\models \varphi \iff 
  \mathcal{D} \models \lnot (\lnot \varphi) 
  \iff 
  \mathcal{D} \models \lnot (\psi)
  \iff
  \mathcal{D} \not\models \psi
  \iff
  \overline{\mathcal{D}} \not \models \overline{\varphi}.$$
  Clearly, $\overline{\varphi}$ can be constructed in logspace from $\varphi$
  and the result follows.
\end{proof}

We will use this principle of duality to only classify one fragment or
its dual, for example we will study $\csplogic$ and ignore its dual
$\{\forall,\lor \} \mbox{-}\mathrm{FO}$. We will also use this
principle to classify the self-dual fragment $\mylogic$.

We assume at least one quantifier and one binary connective (weaker fragments
being trivial). By the duality principle, we may consider only purely
existential fragments, or fragments with both quantifiers. Regarding
connectives, we have three possibilities: purely disjunctive fragments, 
purely conjunctive fragments and fragments with both connectives.
Regarding equality and disequality, we should have the four possible
subsets of $\{=, \neq\}$ but it will become clear that cases with both
follow the same complexity delineation as the case with $\neq$ only.
%%% BOF.
%%% because: 1) hardness proof do not need equality
%%%          2) tractability allows for equality.
Moreover, for fragments with both quantifiers, we may use the duality
principle between $\{\exists,\forall,\land\}$ and
$\{\forall,\exists,\lor\}$ to simplify our task.
This means that we would need to consider $3 \times 3$ positive existential
fragments and $2 \times 3$ positive fragments with both quantifiers.
Actually, we can decrease this last count by one, due to the duality between
$\posFoNeq$ and $\posFoEq$.
Regarding fragments with $\neg$, since we necessarily have both
connectives and both quantifiers, we only have to consider two
fragments: $\FO$ and $\EqFreeFo$. However, we shall see that the complexity of $\FO$ agrees with
that of $\posFoNeq$ (and its dual $\posFoEq$).

This makes a grand total of \emph{15 fragments to classify}, which are listed
below; the fragments marked with a $\star$ correspond to the \CSP\ and
\QCSP\ and are still open. We will settle all other listed fragments.

The 15 relevant fragments can be organised broadly in the following four
classes.

\paragraph*{First Class} This consists of the following trivial
fragments: for such a fragment $\mathscr{L}$, the problem
$\mathscr{L}(\mathcal{D})$ is trivial (in \Logspace) for any structure $\mathcal{D}$.
\begin{itemize}
\item $\{\exists, \lor\} \mbox{-}\mathrm{FO}$        \hfill (see Proposition~\ref{prop:firstClass:complexity})
\item $\{\exists, \lor, = \} \mbox{-}\mathrm{FO}$    \hfill (see Proposition~\ref{prop:firstClass:complexity})
\item $\{\exists, \lor, \neq\} \mbox{-}\mathrm{FO}$  \hfill (see Proposition~\ref{prop:firstClass:complexity})
\end{itemize}

\paragraph*{Second Class} This consists of the following fragments which exhibit a simple dichotomy: for such a fragment $\mathscr{L}$, the problem
$\mathscr{L}(\mathcal{D})$ is trivial (in \Logspace) when the
$\mathscr{L}$-core (defined in the next section) of $\mathcal{D}$ has one element and hard otherwise
(\NP-complete for existential fragments, \Pspace-complete for
fragments that allow both quantifiers). For this class, tractability amounts to
the relativisation of all quantifiers to some constant.
\begin{itemize}
\item \cspDisj, \cspDisjEq               \hfill (see Proposition~\ref{prop:cspDisj:complexity}.)
\item \cspDisjNeq                        \hfill (see Proposition~\ref{prop:cspDisjNeq:complexity}.)
\item \posFoNeq                          \hfill (see Proposition~\ref{proposition:posFoNeq:complexity}.)
\item \EqFreeFo                          \hfill (see Proposition~\ref{proposition:EqFreeFo:complexity}.)
\end{itemize}
\paragraph*{Third Class}
This exhibits more richness complexity-wise, 
tractability can not be explained simply by $\mathscr{L}$-core size
and relativisation of quantifiers. 
\begin{itemize}
\item \csplogicNeq                       \hfill (see Proposition~\ref{prop:csplogicNeq:complexity}.)
\item \qcsplogicNeq                      \hfill (see Proposition~\ref{prop:qcsplogicNeq:complexity}.)
\item[$\star$] \csplogic, \csplogicEq    %%% csp open
\item[$\star$] \qcsplogic, \qcsplogicEq  %%% qcsp open
\end{itemize}
\paragraph*{Fourth Class} The last class consists of a single fragment and is rich
complexity-wise, though we will see that a drop in complexity is always witnessed by
relativisation of quantifiers.
\begin{itemize}
\item \mylogic                           \hfill (see Theorem~\ref{tetrachotomy}.)
\end{itemize}

\subsection{$\mathscr{L}$-Containment and $\mathscr{L}$-Core}
\label{sec:containmentCore}

It is well known that conjunctive query containment is characterised
by the presence of homomorphism between the corresponding canonical
databases (this goes back to Chandra and
Merlin~\cite{DBLP:conf/stoc/ChandraM77}, see
also~\cite[chapter~6]{Gradel:2005:FMT:1206819}). For exactly the same 
reason, a similar result holds for $\csplogic$-containment.
%%% TO DO add introduction sentence
We state and prove this result for pedagogical reasons, before moving
on to other fragments. The results in this section (\S~\ref{sec:containmentCore}) relating to existential fragments are essentially well known.

Let us fix some notation first.
Given a primitive positive sentence $\varphi$ in
$\csplogic$, we denote by $\mathcal{D}_\varphi$ its \emph{canonical
  database}, that is the structure with domain the variables of
$\varphi$ and whose tuples are precisely those that are atoms of
$\varphi$. In the other direction, given a finite structure
$\mathcal{A}$, we write $\phi_{\!\mathcal{A}}$ for the so-called
\emph{canonical conjunctive query}\footnotemark{} of
\footnotetext{Most authors consider the canonical query to be the sentence which is the existential quantification of $\phi_{\!\mathcal{A}}$.}
$\mathcal{A}$, the quantifier-free formula that is the conjunction
of the positive facts of $\mathcal{A}$, where the variables
$v_1,\ldots,v_{|A|}$ correspond to the elements 
$a_1,\ldots,a_{|A|}$ of $\mathcal{A}$.
It is well known that $\mathcal{D}_\varphi$ is homomorphic to a
structure $\mathcal{A}$ if, and only if,
$\mathcal{A}\models\varphi$. Moreover, we now may define a winning strategy for
$\exists$ in the Hintikka $(\mathcal{A},\varphi)$-game
to be precisely the evaluation of the variables given by a homomorphism from $\mathcal{D}_\varphi$ to
$\mathcal{A}$. 
Note also that $\mathcal{A}$ is isomorphic to the canonical 
database of $\exists v_1 \exists v_2 \ldots v_{|A|}
\phi_{\!\mathcal{A}}$.

\begin{theorem}%[\textbf{Containment for \csplogic}]\\ 
  \label{theorem:CSPcontainment}
  Let $\mathcal{A}$ and $\mathcal{B}$ be two structures.
  The following are equivalent.
  \begin{romannum}
  \item For every sentence $\varphi$ in $\csplogic$, if $\mathcal{A}\models \varphi$ then
    $\mathcal{B}\models \varphi$.
  \item There exists a homomorphism from $\mathcal{A}$ to $\mathcal{B}$.
  \item $\mathcal{B}\models \phi_{\!\mathcal{A}}^{\csplogic}$
    where $\phi_{\!\mathcal{A}}^{\csplogic}:= \exists v_1 \exists v_2
    \ldots v_{|A|} \phi_{\!\mathcal{A}}$. 
    % and  $\phi_{\!\mathcal{A}}$ denotes the \emph{canonical
    %   conjunctive query} of $\mathcal{A}$
  \end{romannum}
\end{theorem}
\begin{proof}
  As we observed above, a homomorphism corresponds to a winning
  strategy in the $(\mathcal{A},\varphi)$-game and (ii) and (iii) are equivalent. 

  Clearly, (i) implies (iii) since $\mathcal{A}\models  \exists v_1 \exists
  v_2 \ldots v_{|A|} \phi_{\!\mathcal{A}}$.

  We now prove that (ii) implies (i). Let $h$ be a homomorphism from
  $\mathcal{A}$ to $\mathcal{B}$. If $\mathcal{A}\models \varphi$,
  then there is a homomorphism $g$ from $\mathcal{D}_\varphi$ to
  $\mathcal{A}$. By composition, $h\circ g$ is a homomorphism from
  $\mathcal{D}_\varphi$ to $\mathcal{B}$. In other words, $h\circ g$
  is a winning strategy for $\exists$ in the $(\mathcal{B},\varphi)$-game.
\end{proof}

%%% equivalence and core
\begin{definition}
  Let $\mathcal{A}$ and $\mathcal{B}$ be two structures. 
  We say that
  $\mathcal{A}$ is \emph{$\mathscr{L}$-contained} in $\mathcal{B}$ if, and
  only if, for any $\varphi$ in $\mathscr{L}$, $\mathcal{A}\models
  \varphi$ implies $\mathcal{B}\models \varphi$.
  We say that
  $\mathcal{A}$ and $\mathcal{B}$ are \emph{$\mathscr{L}$-equivalent} if, and
  only if, for any $\varphi$ in $\mathscr{L}$, $\mathcal{A}\models
  \varphi \Leftrightarrow \mathcal{B}\models
  \varphi$. If $\mathcal{B}$ is a minimal structure w.r.t. domain size such
  that $\mathcal{B}$ and $\mathcal{A}$ are $\mathscr{L}$-equivalent,
  then we say that $\mathcal{B}$ is an \emph{$\mathscr{L}$-core} of $\mathcal{A}$.
\end{definition}

The $\csplogic$-core is unique up to isomorphism and is better known
as \emph{the core}. We proceed to characterise notions of containment,
equivalence and core for other fragments of \FO, which we will use to
study the complexity of the associated model checking problems. 
The results of this section are summarised in Table~\vref{table:Containment:Equivalence:Core}. 
%%%ICI
% Table~\ref{table:Containment:Equivalence:Core} surveys the result of
% this section.
%%%%%%%%%%%%%%%%%%%%%%%%%%%%%%%%%%%%%%%%%%%%%%%%%%%%%%%%%%%%%%%%%%%%%%%%

\begin{table}[t]
% \tbl{The various notions of containment, equivalence and core for
%   syntactic fragments of \FO.}{%
\resizebox{\textwidth}{!}{
  \begin{tabular}{|r|c|c|c|}
    % \hline
    % \multicolumn{4}{|c|}{Containment, Equivalence and Core}\\
    % \hline
    \hline
    Fragment $\mathscr{L}$ & $\mathscr{L}$-containment &
    $\mathscr{L}$-equivalence & $\mathscr{L}$-core \\
    \hline
    \csplogic   & 
    \multirow{4}{*}{homomorphism}&
    \multirow{4}{.2\textwidth}{homomorphic equivalence}&
    \multirow{4}{*}{(classical) core}\\
    \csplogicEq &&&\\
    \cspDisj    &&&\\
    \cspDisjNeq &&&\\
    \hline
    \mylogic    & 
    \multicolumn{1}{|p{.2\textwidth}|}{surjective hyper-morphism} & 
    \multicolumn{1}{|p{.2\textwidth}|}{surjective hyper-morphism equi\-valence} & 
    $U$-$X$-core\\
    \hline
    \EqFreeFo   & 
    \multicolumn{1}{|p{.2\textwidth}|}{Full surjective hyper-morphism} & 
    \multicolumn{1}{|p{.2\textwidth}|}{Full surjective hyper-morphism} & 
    quotient by $\sim$\\ 
    \hline
    contains \csplogicNeq\ & 
    \multirow{3}{*}{isomorphism}&
    \multirow{3}{*}{isomorphism}&
    \multirow{3}{*}{each structure}\\
    \multicolumn{1}{|c|}{or} &&&\\
    contains \CocsplogicNeq &&&\\
    \hline
  \end{tabular}
}%
\caption{The various notions of containment, equivalence and core for
  syntactic fragments of \FO.}
%}
\label{table:Containment:Equivalence:Core}
\end{table}

\begin{proposition}
\label{prop:containment:cspDisjEq}
Let $\mathcal{A}$ and $\mathcal{B}$ be two relational structures.
The following are equivalent.
\begin{romannum}
\item There is a homomorphism from $\mathcal{A}$ to $\mathcal{B}$.
\item $\mathcal{A}$ is \csplogic-contained in $\mathcal{B}$.
\item $\mathcal{A}$ is \csplogicEq-contained in $\mathcal{B}$.
\item $\mathcal{A}$ is \cspDisj-contained in $\mathcal{B}$.
\item $\mathcal{A}$ is \cspDisjEq-contained in $\mathcal{B}$.
\end{romannum}
\end{proposition}
\begin{proof}
  The equivalence of (i) and (ii) are stated in
  Theorem~\ref{theorem:CSPcontainment} and are equivalent to 
  $\mathcal{B}\models \exists v_1 \exists v_2 \ldots v_{|A|}
  \phi_{\!\mathcal{A}}$, a sentence of $\csplogic$.
  This takes care of the implications from (v), (iv) and (iii) to (i).
  Trivially (v) implies both (iv) and (iii).

  It suffices to prove (i) implies (v). As in the proof of
  Theorem~\ref{theorem:CSPcontainment}, it can be easily checked that a homomorphism can be applied
  to a winning strategy for $\exists$ in the
  $(\mathcal{A},\varphi)$-game to obtain a winning strategy for $\exists$ in the
  $(\mathcal{B},\varphi)$-game. To see this, write the quantifier-free
  part $\psi$ of $\varphi$ in conjunctive normal form as a disjunction
  of conjunction-of-positive-atoms $\psi_i$. We may
  even propagate equality out by substitution such that each $\psi_i$
  is equality-free (if some $\psi_i$ contained no extensional
  symbol other than equality, the sentence $\varphi$ would trivially
  holds on any structure as we only ever consider structures with at
  least one element). A winning strategy in the
  $(\mathcal{A},\varphi)$-game corresponds to a homomorphism from some
  $\mathcal{D}_{\psi_i}$ to $\mathcal{A}$. By composition with the
  homomorphism from $\mathcal{A}$ to $\mathcal{B}$, we get a
  homomorphism from $\mathcal{D}_{\psi_i}$ to $\mathcal{B}$,
  \textsl{i.e.} a winning strategy in the $(\mathcal{B},\varphi)$-game
  as required.
\end{proof}
\begin{corollary}
  \label{prop:equivalence:cspDisjEq}
  Let $\mathcal{A}$ and $\mathcal{B}$ be two relational structures.
  The following are equivalent.
\begin{romannum}
\item $\mathcal{A}$ and $\mathcal{B}$ are homomorphically equivalent.
\item $\mathcal{A}$ and $\mathcal{B}$ have isomorphic cores.
\item $\mathcal{A}$ is \csplogic-equivalent to $\mathcal{B}$.
\item $\mathcal{A}$ is \csplogicEq-equivalent to $\mathcal{B}$.
\item $\mathcal{A}$ is \cspDisj-equivalent to $\mathcal{B}$.
\item $\mathcal{A}$ is \cspDisjEq-equivalent to $\mathcal{B}$.
\end{romannum}
\end{corollary}
% \begin{proof}
%   The equivalence between (i) and (ii) is well documented
% \end{proof}

We now move on to fragments containing $\{\exists, \land,\neq\}$.
\begin{proposition}%[\textbf{Containment for \csplogicNeq}]\\ 
  \label{prop:CSPNeqContainment}
  Let $\mathcal{A}$ and $\mathcal{B}$ be two structures.
  The following are equivalent.
  \begin{romannum}
  \item For every sentence $\varphi$ in $\csplogicNeq$, if $\mathcal{A}\models \varphi$ then
    $\mathcal{B}\models \varphi$.
  \item There exists an injective homomorphism from $\mathcal{A}$ to $\mathcal{B}$.
  \item $\mathcal{B}\models \phi_{\!\mathcal{A}}^{\csplogicNeq}$
    where $\phi_{\!\mathcal{A}}^{\csplogicNeq}:=\exists v_1 \ldots v_{|A|}
    \phi_{\!\mathcal{A}} \land \bigwedge_{1\leq i<j \leq |A|} v_i \neq
    v_j$.
  \end{romannum}
\end{proposition}
\begin{proof}
  Similar to Theorem~\ref{theorem:CSPcontainment}.
\end{proof}
\begin{corollary}
  Let $\mathscr{L}$ be a fragment of \FO\ such that $\mathscr{L}$ or
  its dual $\overline{\mathscr{L}}$ contains
  $\csplogicNeq$.
  Let $\mathcal{A}$ and $\mathcal{B}$ be two structures.
  The following are equivalent.
  \begin{romannum}
  \item $\mathcal{A}$ and $\mathcal{B}$ are isomorphic.
  \item $\mathcal{A}$ is $\mathscr{L}$-equivalent to $\mathcal{B}$.
  \end{romannum}
\end{corollary}
\begin{proof}
  For the case when $\mathscr{L}$ contains $\csplogicNeq$, the result
  follows from the previous proposition and the fact that we deal with 
  finite structures only.

  For the case when $\overline{\mathscr{L}}$ contains
  $\csplogicNeq$, we apply the duality principle and the previous case, and equivalently $\overline{\mathcal{A}}$ and
  $\overline{\mathcal{B}}$ are isomorphic. This is in turn equivalent
  to $\mathcal{A}$ being isomorphic to $\mathcal{B}$.
\end{proof}

\subsection{Hintikka Games}
\label{sec:hintikka-games}
%%%%%%%%%%%%%%%%%%%%%%%%%%%%%%%%%%%%%%%%%%%%%%%%%%%%%%%%%%%%%%%%%%%%%%%%
Before moving on to the equality-free fragments \mylogic{} and \EqFreeFo,
let us recall first basic definitions and notations regarding Hintikka Games.
Let $\varphi$ be a sentence of \FO\ in prenex
form with all negations pushed to the atomic level.
A \emph{strategy} for $\exists$ in the (Hintikka)
$(\mathcal{A},\phi)$-game is a set of mappings $\{\sigma_x : \mbox{`$\exists x$'} \in \phi \}$ with one mapping $\sigma_x$ for each existentially
quantified variable $x$ of $\varphi$.
The mapping $\sigma_x$ ranges over the domain $A$ of $\mathcal{A}$; and, its domain is the
set of functions from $Y_x$ to $A$, where $Y_x$ denotes the
universally quantified variables of $\phi$ preceding $x$.
 
We say that $\{\sigma_x : \mbox{`$\exists x$'} \in \phi \}$ is \emph{winning} if for any assignment
$\pi$ of the universally quantified variables of $\varphi$ to $A$, when each 
existentially quantified variable $x$ is set according to $\sigma_x$
applied to $\left.\pi\right|_{Y_x}$, then the quantifier-free part $\psi$ of
$\varphi$ is satisfied under this overall assignment $h$.
When $\psi$ is a conjunction of positive atoms, this amounts to $h$ being a
homomorphism from $\mathcal{D}_\psi$ to $\mathcal{A}$.

\subsection{Hyper-morphisms}
\label{sec:hyper-morphisms}
For the equality-free fragments \EqFreeFo\ and \mylogic, the correct concept to transfer
winning strategies involves unary hyper-operations, that is functions to
the power-set. 

%\begin{definition}
\label{def:HeAndShe} 
A \emph{hyper-operation} $f$ from a set $A$ to a set $B$ is a
function from $A$ to the power-set of $B$. 
For a subset $S$ of $A$, we will define its image $f(A)$ under the
hyper-operation $f$ as $\bigcup_{s \in S} f(s)$. 
When we wish to stress that
an element may be sent to $\emptyset$, we speak of a \emph{partial
  hyper-operation}; and otherwise we assume that $f$ is \emph{total}, that is
for any $a$ in $A$, $f(a)\neq \emptyset$.
We say that $f$ is \emph{surjective} whenever $f(A)=B$.
The \emph{inverse} of a (total) hyper-operation $f$ from $A$ to $B$, denoted
by $f^{-1}$, is the partial hyper-operation from $B$ to $A$
defined for any $b$ in $B$ as 
%\begin{equation*}
$f^{-1}(b):=$ $\{a\in A \mid b \in f(a)\}$.
%\end{equation*}
We call an element of $f^{-1}(b)$ an \emph{antecedent} of $b$ under
$f$. 
Let $f$ be a hyper-operation from $A$ to $B$ and $g$ a hyper-operation
from $B$ to $C$. The hyper-operation $g\circ f$ is defined naturally
as $g\circ f (x) := g\bigl(f(x)\bigr)$ (recall that $f(x)$ is a set).

When $f$ is a (total) surjective hyper-operation from $A$ to $A$, we say
that $f$ is a \emph{shop} of $A$. Note that the inverse of a shop is a
shop and that the composition of two shops is a also a shop.
Observing further that shop composition is associative and
that the identity shop (which sends an element $x$ of $A$ to the
singleton $\{x\}$) is the identity with respect to composition, we may
consider the monoid generated by a set of shops. A shop $f$ is a
\emph{sub-shop} of a shop $g$ whenever, for every $x$ in $A$, $f(x) \subseteq g(x)$. In our context, we will
be interested in a particular monoid which will be closed further
under sub-shops, a so-called \emph{down-shop-monoid} (DSM).
\footnote{The ``down'' comes from \emph{down-closure}, here under sub-shops; a nomenclature inherited from~\cite{BornerTotalMultifunctions}.} 
 We denote by $\langle F \rangle_{DSM}$ the DSM generated by a set $F$ of shops.

Let $f$ be a shop of $A$. When for a subset $U$ of $A$ we have $f(U)=A$, we say
that $f$ is \emph{$U$-surjective}.
Observing that the totality of $f$ may be rephrased as %follows 
$f^{-1}(A)=A$, we say more generally that $f$ is \emph{$X$-total} for
a subset $X$ of $A$ whenever $f^{-1}(X)=A$.
Note that for shops $U$-surjectivity and $X$-totality are dual to one another, that is the inverse of
a $U$-surjective shop is an $X$-total shop with $X=U$ and vice versa.
Somewhat abusing terminology, and when it does not cause confusion, we will drop the word surjective and by
$U$- or $U'$-shop we will mean a $U$- or $U'$-surjective
shop. Similarly, we will speak of an $X$- or $X'$-shop in the total
case and of a $U$-$X$-shop in the case of a shop that is
both $U$-surjective and $X$-total. 
Suitable compositions of $U$-shops and $X$-shops preserve these
properties. 
\begin{lemma}
  \label{lem:compositionAndUXshops}
  Let $f$ and $g$ be two shops.
  \begin{romannum}
  \item If $f$ is a $U$-shop then $g\circ f$ is a $U$-shop.
  \item If $g$ is a $X$-shop then $g\circ f$ is a $X$-shop.
  \item If both $f$ is a $U$-shop and $g$ is a $X$-shop then $g\circ
    f$ is a $U$-$X$-shop.
  \item If both $f$ and $g$ are $U$-$X$-shops  then $g\circ f$ is a
    $U$-$X$-shop.
  \item The iterate of a $U$-$X$-shop is a $U$-$X$-shop.
  \end{romannum}
\end{lemma}
\begin{proof}
  We prove (i). Since $f(U)=A$, we have $g(f(U))=g(A)$.
  By surjectivity of $g$, we know that $g(A)=A$. It follows that
  $g(f(U))=A$ and we are done. 
  (ii) is dual to (i), and (iii) follows directly
  from (i) and (ii).
  (iv) is a restriction of (iii) and is only stated here as
  we shall use it often. (v) follows by induction on
  the order of iteration using (iv).
\end{proof}
  
A \emph{hyper-morphism} $f$ from a structure
$\mathcal{A}$ to a structure $\mathcal{B}$ is a hyper-operation from
$A$ to $B$ that satisfies the following property.
\begin{itemize}
\item (\textbf{preserving}) if $R(a_1,\ldots,a_i) \in \mathcal{A}$ then $R(b_1,\ldots,b_i) \in \mathcal{B}$, for all $b_1 \in f(a_1),\ldots,b_i \in
  f(a_i)$. 
\end{itemize}
When $\mathcal{A}$ and $\mathcal{B}$ are the same structure, we
speak of a \emph{hyper-endomorphism}.
We say that $f$ is \emph{full} if moreover
\begin{itemize}
\item (\textbf{fullness}) $R(a_1,\ldots,a_i) \in \mathcal{A}$ iff $R(b_1,\ldots,b_i) \in \mathcal{B}$, for all $b_1 \in f(a_1),\ldots,b_i \in
  f(a_i)$. 
\end{itemize}
%\end{definition}
Note that the inverse of a full surjective hyper-morphism 
is also a full surjective hyper-morphism.

\subsection{Containment and Core for \EqFreeFo}
\label{sec:cont-core-eqfr}

We now turn our attention to \EqFreeFo. The proofs of the necessary characterisations are somewhat laboured, but will prepare us well for the forthcoming discussion on \mylogic.
\begin{lemma}%[\textbf{strategy transfer for \EqFreeFo}]
\label{lemma:strategy:transfert:EqFreeFo}
  Let $\mathcal{A}$ and $\mathcal{B}$ be two structures such that
  there is a full surjective hyper-morphism from $\mathcal{A}$ to
  $\mathcal{B}$. Then, for every sentence $\varphi$ in $\EqFreeFo$, if
  $\mathcal{A}\models \varphi$ then $\mathcal{B}\models \varphi$.
\end{lemma}
\begin{proof}
  Let $h$ be a full surjective hyper-morphism from $\mathcal{A}$ to
  $\mathcal{B}$ and $\varphi$ be a sentence of $\EqFreeFo$ such that
  $\mathcal{A}\models \varphi$. % For any element $b$ of
  % $\mathcal{B}$, let $h^{-1}(b):=\{a \in A \mbox{ s. t. } b\in h(a)\}$.
  We fix an arbitrary linear order over $A$ and write $\min h^{-1}(b)$
  to denote the smallest antecedent of $b$ in $A$ under $h$. 
  
  Let $\{\sigma_x : \mbox{`$\exists x$'} \in \phi \}$ be a winning strategy in the %Hintikka 
  $(\mathcal{A},\phi)$-game. We construct a strategy $\{\sigma'_x : \mbox{`$\exists x$'} \in \phi \}$
  in the %Hintikka 
  $(\mathcal{B},\phi)$-game as follows.
  Let $\pi_B:Y_x\to B$ be an assignment to the universal variables $Y_x$
  preceding an existential variable $x$ in $\varphi$, we select for 
  $\sigma'_x(\pi)$ an arbitrary element of $h(\sigma(\pi_A))$ where
  $\pi_A:Y_x\to A$ is an assignment such that for any universal variable
  $y$ preceding $x$, we have $\pi_A(y):=\min h^{-1}(\pi_B(y))$. This strategy is well
  defined since $h$ is surjective (which means that $\pi_A$ is well
  defined) and total (which means that $h(\sigma(\pi_A))\neq
  \emptyset$). Note moreover that using $\min$ in the definition of
  $\pi_A$ means that a branch in the tree of the game on $\mathcal{B}$ will
  correspond to a branch in  the tree of the game on $\mathcal{A}$.
  It remains to prove that $\{\sigma'_x : \mbox{`$\exists x$'} \in \phi \}$ is winning. We will see
  that it follows from the fact that $h$ is full and preserving.

  We assume that negations have been pushed to the
  atomic level and write the quantifier-free part $\psi$ of $\phi$ in  disjunctive normal
  form as a disjunction of conjunctions-of-atoms $\psi_i$. If $\psi_i$
  has contradictory positive and negative atoms (as in $E(x,y)\land
  \lnot E(x,y)$) then we may discard the sentence $\psi_i$ as false.
  Moreover, for each pair of atoms $R(v_1,v_2,\ldots,v_r)$ and $\lnot
  R(v_1,v_2,\ldots,v_r)$ (induced by the choice of a relational symbol
  $R$ and the choice of $r$ variables 
  $v_1,v_2,\ldots,v_r$ occuring in
  $\psi_i$) such that neither  is present in $\psi_i$, we may replace $\psi_i$
  by the logically equivalent 
  $\bigl(\psi_i\land R(v_1,v_2,\ldots,v_r)\bigr)
  \lor
  \bigl(\psi_i\land \lnot R(v_1,v_2,\ldots,v_r)\bigr)$.
  After this completion process, note that every
  conjunction of atoms $\psi_i$ corresponds naturally to a structure
  $\mathcal{D}_{\psi_i}$ (take only the positive part of $\psi_i$ which is now maximal). 
  
  Assume first that $\psi$ is disjunction-free.
  The winning
  condition of the %Hintikka 
  $(\mathcal{B},\phi)$-game can be recast as a full homomorphism from
  $\mathcal{D}_\psi$ to $\mathcal{B}$. Composing with $h$ the full homomorphism from $\mathcal{D}_\psi$ to
  $\mathcal{A}$ (induced by the sequence of compatible assignments
  $\pi_A$ to the universal variables and the strategy
  $\{\sigma_x : \mbox{`$\exists x$'} \in \phi \}$), we get a full hyper-morphism from $\mathcal{D}_\psi$ to
  $\mathcal{B}$. The map from the domain of $\mathcal{D}_\psi$ to $\mathcal{B}$
  induced by the sequence of assignments $\pi_B$ and the strategy
  $\{\sigma'_x : \mbox{`$\exists x$'} \in \phi \}$ is a range restriction of this full
  hyper-morphism and is therefore a full homomorphism (we identify hyper-morphism to singletons with homomorphisms).
In general when the quantifier-free part of $\phi$ has several
  disjuncts $\psi_i$, most likely after the completion process of the previous paragraph, the winning condition can be recast as a full homomorphism
  from some $\mathcal{D}_{\psi_i}$. The above argument applies and the result follows.  
\end{proof}

We shall see that the converse of
Lemma~\ref{lemma:strategy:transfert:EqFreeFo} holds. Consequently,
it turns out that containment and equivalence coincide for \EqFreeFo,
since the inverse of a full surjective hyper-morphism is a full
surjective hyper-morphism.

For \EqFreeFo, we define an equivalence relation $\sim$ over the
structure elements in the spirit of the Leibnitz-rule for
equality. For propositions $P$ and $Q$, let $P \leftrightarrow Q$ be
an abbreviation for $(P \land Q) \lor (\lnot P \land \lnot Q)$.
For the sake of clarity, we deal with the case of digraphs first and write
$x \sim y$ as an abbreviation for $\forall z (E(x, z) \leftrightarrow
E(y, z)) \land (E(z, x) \leftrightarrow E(z, y))$.  
It is straightforward to verify that $\sim$ induces an equivalence
relation over the vertices (which we denote also by $\sim$). 
In general, for each $r$-ary symbol $R$, let  $\psi_R$ stands for
\begin{multline*}
\bigl(R(x,z_1,\ldots,z_{r-1})\leftrightarrow
R(y,z_1,\ldots,z_{r-1})\bigr)  
\land \bigl(R(z_1,x,z_2,\ldots,z_{r-1})\leftrightarrow
R(z_1,y,z_2,\ldots,z_{r-1})\bigr)\\ 
\land
\ldots
\land
\bigl(R(z_1,z_2,\ldots,z_{r-1},x)\leftrightarrow
R(z_1,z_2,\ldots,z_{r-1},y)\bigr).
\end{multline*}
We write
$x \sim y$ for $ \bigwedge_{R\in \sigma} \forall z_1,z_2,\ldots,
z_{r-1} \psi_R$.

We write $\mathcal{A}/\!\!\sim$ for the quotient structure defined in the
natural way. 
Note that there is a full surjective homomorphism from
$\mathcal{A}$ to $\mathcal{A}/\!\!\sim$. As observed earlier, its inverse
(viewing the homomorphism as an hyper-morphism) is a full surjective
hyper-morphism from $\mathcal{A}/\!\!\sim$ to $\mathcal{A}$. Thus, it
follows from Lemma~\ref{lemma:strategy:transfert:EqFreeFo} that
$\mathcal{A}$ and $\mathcal{A}/\!\!\sim$ are $\EqFreeFo$-equivalent.

Let $\phi^+_{\!\mathcal{A}}$ denotes the (quantifier-free) canonical conjunctive query
of $\mathcal{A}$ (denoted earlier as $\phi_{\!\mathcal{A}}$) and
$\phi^-_{\!\mathcal{A}}$ denotes the similar sentence which
lists the negative atoms of $\mathcal{A}$ instead of the positive atoms.
\begin{proposition}%[\textbf{Containment and equivalence for \EqFreeFo}]\\ 
  \label{prop:EqFreeFoContainment}
  Let $\mathcal{A}$ and $\mathcal{B}$ be two structures.
  The following are equivalent.
  \begin{romannum}
  \item For every sentence $\varphi$ in $\EqFreeFo$, if $\mathcal{A}\models \varphi$ then
    $\mathcal{B}\models \varphi$.
  \item There exists a full surjective hyper-morphism from
    $\mathcal{A}$ to $\mathcal{B}$.
  \item $\mathcal{B}\models \phi_{\!\mathcal{A}}^{\EqFreeFo}$ where
    $$\phi_{\!\mathcal{A}}^{\EqFreeFo}:=
    \exists v_1 \exists v_2 \ldots v_{|A|} 
    \phi^+_{\!\mathcal{A}} \land \phi^-_{\!\mathcal{A}}
    \land \forall w \bigvee_{1\leq i \leq |A|}  w\sim
    v_i.$$   
  \item for every sentence $\varphi$ in $\EqFreeFo$, $\mathcal{A}\models \varphi$ iff
    $\mathcal{B}\models \varphi$.
  \item $\mathcal{A}/\!\!\sim$ and $\mathcal{B}/\!\!\sim$ are isomorphic.
  \end{romannum}

\end{proposition}
\begin{proof}
  The implication (i) to (iii) is clear since by construction $\mathcal{A}$ models the
  canonical sentence $\phi_{\!\mathcal{A}}^{\EqFreeFo}$.
  
  We prove that (iii) implies (ii). Assume that $\mathcal{B}\models
  \phi_{\!\mathcal{A}}^{\EqFreeFo}$. We construct a full and total surjective
  hyper-morphism $h$ as follows.
  Let $b_1, b_2, \ldots, b_{|A|}$ be witnesses in $B$ for $v_1, v_2,
  \ldots, v_{|A|}$. We set $h(a_i)\ni b_i$ for $1\leq i \leq |A|$ (totality).
  For each $b$ in $B$, we set the universal variable $w$ to $b$ and
  pick some $j$ such that $w\sim v_j$ holds and set $h(a_j)\ni b$ (surjectivity).
  By construction, $h$ is preserving and full.
  
  The implication (ii) to (i) is proved as Lemma~\ref{lemma:strategy:transfert:EqFreeFo}.

  The equivalence of (i), (ii), (iii) with (iv) follows from our earlier
  observation that the inverse $f^{-1}$ of a full surjective hyper-morphism
  $f$ from $\mathcal{A}$ to $\mathcal{B}$ is a full surjective
  hyper-morphism from $\mathcal{B}$ to $\mathcal{A}$. 

  To see that (v) implies (ii), compose the quotient map from
  $\mathcal{A}$ to $\mathcal{A}/\!\!\sim$ (which is a full surjective
  homomorphism) with the inverse of the quotient map from
  $\mathcal{B}$ to $\mathcal{B}/\!\!\sim$ (which is a full surjective
  hyper-morphism).

  For the direction (ii) to (v), the natural quotient $f/\!\!\sim$ of a full surjective
  hyper-morphism $f$ from $\mathcal{A}$ to $\mathcal{B}$ is a full
  surjective homomorphism. Since we deal with finite structures, it is
  an isomorphism and we are done.
\end{proof}

Note that no smaller structure can be $\EqFreeFo$-equivalent to
$\mathcal{A}':=\mathcal{A}/\!\!\sim$. Indeed, a full surjective 
hyper-morphism $f$ from a smaller structure $\mathcal{B}$ to $\mathcal{A}'$
would have to satisfy $\{a'_1,a'_2\}\subseteq f(b)$ for some $b$ in $B$ and
some distinct $a'_1,a'_2$ in $A'$. But this would imply that $a'_1\sim a'_2$ which is not
possible. Moreover, any structure that is $\EqFreeFo$-equivalent and of the same size as 
$\mathcal{A}'$ will be
isomorphic (a full surjective hyper-morphism must induce an isomorphism
by triviality of $\sim$ over $\mathcal{A}'$). 
Thus, $\mathcal{A}/\!\!\sim$ is the (up to isomorphism unique) $\EqFreeFo$-core of
$\mathcal{A}$.

%%%%%%%%%%%%%%%%%%%%%%%%%%%%%%%%%

\subsection{Containment for \mylogic}
\label{sec:containment-mylogic}

\begin{lemma}%[\textbf{strategy transfer for \mylogic}]
  \label{lemma:strategy:transfert:mylogic}
  Let $\mathcal{A}$ and $\mathcal{B}$ be two structures such that
  there is a surjective hyper-morphism from $\mathcal{A}$ to
  $\mathcal{B}$. Then, for every sentence $\varphi$ in $\mylogic$, if
  $\mathcal{A}\models \varphi$ then $\mathcal{B}\models \varphi$.
\end{lemma}
\begin{proof}
  The proof is exactly the same as that of
  Lemma~\ref{lemma:strategy:transfert:EqFreeFo}, except that we no
  longer need to preserve atomic negation, and may drop the
  assumption of fullness. 
\end{proof}

We extend the notion of canonical conjunctive query of a structure $\mathcal{A}$. Given a tuple
of (not necessarily distinct) elements $\tuple{r}:=(r_1,\ldots,r_l)
\in A^l$, define the quantifier-free formula
$\phi_{\!\mathcal{A}(\tuple{r})}(v_1,\ldots,v_l)$ to be the conjunction
of the positive facts of $\tuple{r}$, where the variables
$v_1,\ldots,v_l$ correspond to the elements $r_1,\ldots,r_l$. That is,
$R(v_{\lambda_1},\ldots,v_{\lambda_i})$ appears as an atom in
$\phi_{\!\mathcal{A}(\tuple{r})}$ iff
$R(r_{\lambda_1},\ldots,r_{\lambda_i})$ holds in $\mathcal{A}$. 
When $\tuple{r}$ enumerates the elements of the structure
$\mathcal{A}$, this definition coincides with the usual definition of canonical
conjunctive query.
Note also that in this case there is a full homomorphism from the canonical database
$\mathcal{D}_{\!\phi_{\!\mathcal{A}(\tuple{r})}}$ to $\mathcal{A}$ given by
the map $v_{\lambda_i} \mapsto r_i$.

\begin{definition}[\cite{DBLP:journals/tocl/MadelaineM12}]%\textbf{Canonical $\mylogic$ sentence}]
  Let $\mathcal{A}$ be a structure and $m>0$.
  Let $\tuple{r}$ be an enumeration of the elements of $\mathcal{A}$.
  $$\theta_{\!\mathcal{A},m}^{\mylogic}:=
  \exists v_1, \ldots, v_{|A|}
  \phi_{\!\mathcal{A}(\tuple{r})}(v_1,\ldots,v_{|A|})
  \land
  \forall w_1,\ldots,w_m 
  \bigvee_{\tuple{t} \in A^{m}}
  \phi_{\!\mathcal{A}(\tuple{r},\tuple{t})}(\tuple{v},\tuple{w}).
  $$
\end{definition}
Observe that $\mathcal{A}\models \theta_{\!\mathcal{A},m}^{\mylogic}$. Indeed, we
may take as  witness for the variables $\tuple{v}$ the corresponding enumeration $\tuple{r}$
of the elements of
$\mathcal{A}$; and, for any assignment $\tuple{t} \in A^{m}$ to the
universal variables $\tuple{w}$, it is clear that
$\mathcal{A}\models\phi_{\!\mathcal{A}(\tuple{r},\tuple{t})}(\tuple{r},\tuple{t})$
holds.
\begin{lemma}\label{lemma:canonical:sentence:mylogic}
  Let $\mathcal{A}$ and $\mathcal{B}$ be two structures.
  If $\mathcal{B}\models \theta_{\!\mathcal{A},|\mathcal{B}|}^{\mylogic}$ then
  there is a surjective hyper-morphism from $\mathcal{A}$ to $\mathcal{B}$. 
\end{lemma}
\begin{proof}
  Let $\tuple{b'}:=b'_1,\ldots,b'_{|A|}$ be witnesses for $v_1,\ldots,v_{|A|}$.
  Assume that an enumeration $\tuple{b}:=b_1,b_2,\ldots,b_{|B|}$ of the elements of
  $\mathcal{B}$  is chosen for the
  universal variables $w_1,\ldots w_{|\mathcal{B}|}$.
  Let $\tuple{t} \in A^{m}$ be the witness s.t. 
  $\mathcal{B}\models 
  \phi_{\!\mathcal{A}(\tuple{r})}(\tuple{b'})
  \land
  \phi_{\!\mathcal{A}(\tuple{r},\tuple{t})}(\tuple{b'},\tuple{b})
  $.
  
  Let $f$ be the map from the domain of $\mathcal{A}$ to the power set
  of that of $\mathcal{B}$ which is the union of the following two
  partial hyper-operations $h$ and $g$ (\textsl{i.e.} $f(a_i):=h(a_i)\cup g(a_i)$
  for any element $a_i$ of $\mathcal{A}$), which guarantee totality
  and surjectivity, respectively.
  \begin{itemize}
  \item  $h(a_i):=b'_i$ \hfill (totality.)
  \item  $g(t_i)\ni b_i$ \hfill (surjectivity.)
  \end{itemize}
  It remains to show that $f$ is preserving. This follows from
  $\mathcal{B}\models
  \phi_{\!\mathcal{A}(\tuple{r},\tuple{t})}(\tuple{b'},\tuple{b})$. 

  Let $R$ be a $r$-ary relational symbol such that
  $R(a_{i_1},\ldots,a_{i_r})$ holds in $\mathcal{A}$. Let $b''_{i_1}\in f(a_{i_1}), \ldots, b''_{i_r}\in
  f(a_r)$. We will show that $R(b''_{i_1},\ldots,b''_{i_r})$ holds in $\mathcal{B}$.
  Assume for clarity of the exposition and w.l.o.g. that from $i_1$ to $i_k$ the image is set according to
  $h$ and from $i_{k+1}$ to $i_r$ according to $g$:
  i.e. for $1\leq j \leq k$, $h(a_{i_j})=b'_{i_j}=b''_{i_j}$  and 
  for $k+1\leq j \leq r$, there is some $l_j$ such that
  $t_{l_j}=a_{i_j}$ and $g(t_{l_j})\ni b''_{i_j}=b_{l_j}$. 
  By definition of  $\mathcal{A}(\tuple{r},\tuple{t})$ the atom
  $R(v_{i_1},\ldots,v_{i_k},w_{l_{k+1}},\ldots,w_{r})$ appears  in  $\phi_{\!\mathcal{A}(\tuple{r},\tuple{t})}(\tuple{v},\tuple{w})$. 
  It follows from   $\mathcal{B}\models
  \phi_{\!\mathcal{A}(\tuple{r},\tuple{t})}(\tuple{b'},\tuple{b})$
  that $R(b''_{i_1},\ldots,b''_{i_r})$ holds in $\mathcal{B}$.
\end{proof}
\begin{theorem}%[\textbf{Containment for \mylogic}]\\
  \label{theo:containment:mylogic}
  Let $\mathcal{A}$ and $\mathcal{B}$ be two structures.
  The following are equivalent.
  \begin{romannum}
  \item For every sentence $\varphi$ in $\mylogic$, if $\mathcal{A}\models \varphi$ then
    $\mathcal{B}\models \varphi$.
  \item There exists a surjective hyper-morphism from $\mathcal{A}$ to $\mathcal{B}$.
  \item $\mathcal{B}\models \theta_{\mathcal{A},|B|}^{\mylogic}$.
  \end{romannum}
  % where $\Theta_{\mathcal{A,B}}^{\mylogic}$ is a canonical sentence of $\mylogic$ that is
  % defined in terms of $\mathcal{A}$ and $|\mathcal{B}|$ and that
  % is modelled by $\mathcal{A}$ by construction.
\end{theorem}
\begin{proof}
  By construction $\mathcal{A}\models \theta_{\mathcal{A},|B|}^{\mylogic}$, so
  (i) implies (iii). By Lemma~\ref{lemma:strategy:transfert:mylogic}, (ii) implies (i).
  By Lemma~\ref{lemma:canonical:sentence:mylogic}, (iii) implies (i).
\end{proof}

%%% U-X core and relativisation goes here. We remove the assumption
%%% that the U-X-shop acts as the identity as it is not needed.
\subsection{A core for \mylogic}
\label{sec:cores}
% Recall that a \emph{core} of a finite structure $\mathcal{D}$ is an induced substructure $\widetilde{\mathcal{D}}$ of $\mathcal{D}$ such that there is a homomorphism from $\mathcal{D}$ to $\widetilde{\mathcal{D}}$ (and
% consequently, they are homomorphically equivalent) and every
% endomorphism of $\widetilde{\mathcal{D}}$ is an automorphism. The core
% of a structure is unique up to isomorphism. 
%The notion of core is ubiquitous in the study of CSPs. 
The property of a (classical) core can be rephrased in the logical
context as the minimal $X=\widetilde{A} \subseteq A$ such that a
primitive positive sentence $\phi$ is true on $\mathcal{A}$ iff it is
true on $\mathcal{A}$ with the (existential) quantifiers relativised
to $X=\widetilde{A}$.  
Let us say in this case that $\mathcal{A}$ has
\emph{$X$-relativisation} with respect to \csplogic.

Thus, the notion of a core can be recast in the context of
\csplogic{} in a number of equivalent ways, as a
minimal induced substructure $\widetilde{\mathcal{A}}$ of $\mathcal{A}$,
\begin{romannum}
\item that satisfies the same \csplogic\ 
  sentences; %(\textbf{combinatorial =>hom problem}) 
\item that is induced by minimal $X\subseteq A$ such that $\mathcal{A}$ has
  $X$-relativisation \mbox{w.r.t.} \csplogic; or,  %(\textbf{logic})
\item that is induced by minimal $X\subseteq A$ such that
  $\mathcal{A}$ has an endomorphism with image $X$. % (\textbf{algebraic}) 
\end{romannum}

We are looking for a useful characterisation of the
analogous concept of core for \mylogic.
As we now have both quantifiers, two
sets $U$ and $X$, one for each quantifier, will emerge naturally,
hence we will call this core a \emph{$U$-$X$-core}. 
As we shall see shortly, there are two equivalent ways of  defining a
$U$-$X$-core -- one is logical, the other algebraic -- as a minimal
substructure $\widetilde{\mathcal{A}}$ of $\mathcal{A}$, induced by
minimal $U, X\subseteq A$ such that: 
\begin{romannum}
  \setcounter{enumi}{1}
\item   $\mathcal{A}$ has \emph{$\forall U$-$\exists X$-relativisation}
  \mbox{w.r.t.} \mylogic; or,  %(\textbf{logic}) 
\item $\mathcal{A}$ has a $U$-surjective $X$-total hyper-endomorphism. % (\textbf{algebraic}) 
\end{romannum}
Recall that a surjective hyper-endomorphism $f$ of $\mathcal{A}$ is \emph{$U$-surjective} if
$f(U)=A$ and \emph{$X$-total} if $f^{-1}(X)=A$.

We will show that the sets $U$ and $X$ are unique up to isomorphism
and that within a minimal induced substructure $\widetilde{\mathcal{A}}$, the
sets $U$ and $X$ are uniquely determined. This will reconcile our
definition of a  $U$-$X$-core with the following natural definition,
in which $U$ and $X$ are not explicit:  
\begin{romannum}
\item 
  as a minimal induced substructure $\widetilde{\mathcal{A}}$ of
  $\mathcal{A}$ that satisfies the same 
  sentences of \mylogic. %(\textbf{combinatorial}) 
\end{romannum}

\noindent In our definition of $\mylogic$-core, we ask for a minimal structure,
i.e. not necessarily an induced substructure. We shall see that it is
equivalent to the above.

\subsection{Relativisation}
\label{sec:relativisation}
Given a formula $\varphi$, we denote by $\varphi_{[\forall u /\forall u
  \in U,\exists x/\exists x \in X]}$ the formula obtained from
$\varphi$ by relativising simultaneously every universal
quantifier to $U$ and every existential quantifier to $X$. 
When we only relativise universal quantifiers to $U$, we write $\varphi_{[\forall u /\forall u
  \in U]}$, and when we only relativise existential quantifiers to $X$,
we write $\varphi_{[\exists x/\exists x \in X]}$. 

\begin{definition}\label{def:UXrelativisation}
  Let $\mathcal{A}$ be a finite structure over a set $A$, and
  $U,X$ be two subsets of $A$. 
  We say that $\mathcal{A}$ has \emph{$\forall U$-$\exists X$-relativisation} if, 
  for all sentences $\varphi$ in \mylogic\
  the following are equivalent
  \begin{enumerate}
  \item[$(i)$] $\mathcal{A}\models \varphi$
  \item[$(ii)$] $\mathcal{A}\models \varphi_{[\forall u /\forall u \in U]}$
  \item[$(iii)$] $\mathcal{A}\models \varphi_{[\exists x/\exists x \in X]}$
  \item[$(iv)$] $\mathcal{A}\models \varphi_{[\forall u /\forall u \in U,\exists x/\exists x \in X]}$
  \end{enumerate}
\end{definition}

\begin{lemma}\label{lemma:UXshop:To:UXrelativisation}
  Let $\mathcal{A}$ be a finite structure over a set $A$, and
  $U,X$ be two subsets of $A$.
  If $\mathcal{A}$ has a $U$-surjective $X$-total hyper-endomorphism 
  then $\mathcal{A}$ has $\forall U$-$\exists X$-relativisation.
\end{lemma}
\begin{proof} Note that in Definition~\ref{def:UXrelativisation},
  we have $(iii)\Rightarrow (i)\Rightarrow (ii)$ and $(iii)\Rightarrow
  (iv)\Rightarrow (ii)$ trivially. It suffices to prove that
  $(ii)\Rightarrow (i)$ and  $(i)\Rightarrow (iii)$ to complete the
  proof. To do so, we will consider the well known Hintikka game corresponding to
  Case $(i)$, called the \emph{unrelativised game} hereafter; and, the
  relativised Hintikka games corresponding to the relativised formulae
  from Cases $(ii)$, $(iii)$ and $(iv)$ (the relativised game
  considered being clear from context).

  Let $h$ be a $U$-surjective $X$-total surjective hyper-endomorphism
  of $\mathcal{D}$.
  The proof follows the line of that of
  Lemma~\ref{lemma:strategy:transfert:mylogic}.
  
  ($(ii)\Rightarrow (i)$). Assume that we have a winning strategy in the
  universally relativised game. We produce a winning strategy in the unrelativised
  game using $h$. When taking the antecedent of a universal
  variable, we make sure to pick an antecedent in $U$ which we can do by
  $U$-surjectivity of $h$. To be more precise, the linear order
  over $\mathcal{A}$ used in the proof of
  Lemma~\ref{lemma:strategy:transfert:mylogic} starts with the
  elements of $U$.
  
  ($(i)\Rightarrow (iii)$). Assume that we have a winning strategy in the
  unrelativised game. We produce a winning strategy in the
  existentially relativised game using $h$.
  When taking the image of an existential variable, we no longer pick
  an arbitrary element but one in $X$, which we can do by $X$-totality
  of $h$. 
\end{proof}

\begin{proposition}\label{prop:dualityOfRelativisation}
  The following are equivalent.
  \begin{romannum}
  \item $\mathcal{A}$ has $\forall U$-$\exists X$-relativisation.
  \item $\overline{\mathcal{A}}$ has $\forall X$-$\exists U$-relativisation.
  \end{romannum}
\end{proposition}
\begin{proof}
  It suffices to prove one implication. We prove (ii) implies (i).
  Let $\varphi$ be a sentence of \mylogic.
  We use the duality principle and prove that $\mathcal{A}\models
  \varphi \iff \mathcal{A}\models \varphi_{[\forall u /\forall u \in
    U]}$. The other cases are similar and are omitted. 
  
  We follow the same notation as in
  Proposition~\ref{prop:duality:principle}: $\psi$ is 
  the sentence logically equivalent to $\lnot \varphi$ with negation
  pushed at the atomic level, and $\overline{\varphi}$ is the sentence
  obtained from $\psi$ by replacing every occurrence of a negative
  atom $\lnot R$ by $R$.
  Recall the following chain of equivalence.
  $$\mathcal{A}\models \varphi \iff 
  \mathcal{A} \models \lnot (\lnot \varphi) 
  \iff 
  \mathcal{A} \models \lnot (\psi)
  \iff
  \mathcal{A} \not\models \psi
  \iff
  \overline{\mathcal{A}} \not \models \overline{\varphi}.$$

  By assumption 
  $\overline{\mathcal{A}} \not \models \overline{\varphi} \iff
  \overline{\mathcal{A}}\not\models \overline{\varphi}_{[\exists u
      /\exists u \in U]}$. Using the above chain of
  equivalence backward and propagating the relativisation we obtain
  the following chain of equivalence.
  \begin{multline*}
  \overline{\mathcal{A}} \not \models \overline{\varphi}_{[\exists u
      /\exists u \in U]}.
  \iff
  \mathcal{A} \not\models \psi_{[\exists u /\exists u \in U]}
  \iff
  \mathcal{A} \models \lnot (\psi_{[\exists u /\exists u \in U]})\\
  \iff
  \mathcal{A} \models \lnot (\lnot \varphi_{[\forall u /\forall u \in
    U]}) 
  \iff
  \mathcal{A}\models \varphi_{[\forall u /\forall u \in U]}.
\end{multline*}
\end{proof}

\begin{lemma}\label{lemma:UXrelativisation:To:UXshop}
  Let $\mathcal{A}$ be a finite structure over a set $A$, and
  $U,X$ be two subsets of $A$.
  If $\mathcal{A}$ has $\forall U$-$\exists X$-relativisation then $\mathcal{A}$ has
  a $U$-surjective $X$-total hyper-endomorphism.
\end{lemma}
\begin{proof}
  Using the fact that the identity (defined as $i(x):=\{x\}$ for every
  $x$ in $\mathcal{A}$) is a surjective hyper-endomorphism of
  $\mathcal{A}$ and applying Theorem~\ref{theo:containment:mylogic},
  we derive that $\mathcal{A}\models\theta_{\mathcal{A},|A|}^{\mylogic}$.
  By assumption, we may equivalently relativise only the existential quantifiers to
  $X$ (Definition~\ref{def:UXrelativisation} $(i)\Rightarrow(iii)$) and
  $\mathcal{A} \models \theta^{\mylogic}_{\mathcal{A},|A|[\exists x/\exists x \in
    X]}$. Proceeding as in the proof of
  Lemma~\ref{lemma:strategy:transfert:mylogic} but over this
  relativised sentence, we derive the
  existence of an $X$-total surjective hyper-operation $g$.

  Using Proposition~\ref{prop:dualityOfRelativisation} and working
  over $\overline{\mathcal{A}}$, we derive similarly that
  $\overline{\mathcal{A}}$ has a $U$-total surjective hyper-operation.
  Let $f$ be the inverse of this hyper-operation. Observe that it is a
  $U$-surjective hyper-operation.

  By Lemma~\ref{lem:compositionAndUXshops}, the composition of these
  operations $g\circ f$ is a $X$-total $U$-surjective
  hyper-endomorphism as required. 
\end{proof}

Together, the two previous lemmata establish an algebraic
characterisation of relativisation. 
\begin{theorem}
  \label{relativization}
  Let $\mathcal{A}$ be a finite structure over a set $A$, and
  $U,X$ be two subsets of $A$. The following are equivalent.
  \begin{romannum}
  \item The structure $\mathcal{A}$ has $\forall U$-$\exists X$-relativisation.
  \item The structure $\mathcal{A}$ has a $X$-total $U$-surjective hyper-endomorphism.
  \end{romannum}
\end{theorem}

\begin{corollary}%[\textbf{retraction}]
  \label{cor:UXretract}
  Let $\mathcal{A}$ be a finite structure that has a $U$-surjective
  $X$-total hyper-endomorphism.
  Let $\widetilde{\mathcal{A}}$ be the substructure of $\mathcal{A}$
  induced by $U\cup X$. The following holds.
  \begin{romannum}
  \item $\mathcal{A}$ and $\widetilde{\mathcal{A}}$ are \mylogic-equivalent.
  \item $\widetilde{\mathcal{A}}$ has $\forall U$-$\exists X$-relativisation.
  \end{romannum}
\end{corollary}
\begin{proof}
  Let $f$ be the $U$-surjective $X$-total hyper-endomorphism of
  $\mathcal{A}$. Its range restriction $g$ to $\widetilde{A}=U\cup X$ is a
  surjective hyper-morphism from $\mathcal{A}$ to
  $\widetilde{\mathcal{A}}$. The inverse $g^{-1}$ of $g$ is a
  surjective hyper-morphism from $\widetilde{\mathcal{A}}$ to
  $\mathcal{A}$, by $X$-totality of $f$.
  Appealing to Lemma~\ref{lemma:strategy:transfert:mylogic} twice,
  once with $g$ and once with $g^{-1}$, we obtain (i).

  The restriction of $g$ to $\widetilde{A}$ is a $U$-surjective
  $X$-total hyper-endomorphism of $\widetilde{\mathcal{A}}$, and (ii)
  follows from Lemma~\ref{lemma:UXshop:To:UXrelativisation}.
\end{proof}

%%% 
\subsection{The $U$-$X$ Core}
\label{sec:u-x-core}

Given a structure $\mathcal{D}$, we consider all minimal subsets $X$ of $D$ such that there is an
$X$-total surjective hyper-endomorphism $g$ of $\mathcal{D}$, and all minimal subsets $U$
such that there is a 
$U$-surjective hyper-endomorphism $f$ of $\mathcal{D}$.
Such sets always exist as totality and surjectivity of surjective hyper-endomorphisms mean that
in the worst case we may choose $U=X=D$. 
Since $g\circ f$ is a $X$-total $U$-surjective hyper-endomorphisms of
$\mathcal{D}$ by Lemma~\ref{lem:compositionAndUXshops}, we may furthermore require that among all minimal sets satisfying
the above, we choose a set $U$ and a set $X$ with $U\cap X$ maximal.
Let $\widetilde{\mathcal{D}}$ be the substructure of $\mathcal{D}$
induced by $U\cup X$. We call $\widetilde{\mathcal{D}}$ a
\emph{$U$-$X$-core} of $\mathcal{D}$.

\begin{remark}
  Assume that there is an $X_1$-shop $h_1$ and an $X_2$-shop $h_2$
  that preserves $\mathcal{D}$ such that $|X_1|>|X_2|$. We consider images of $h_1\circ h_2$.
  For each element in $X_2$, pick a single element $x'_1$ of $X_1$ in
  $h_1(X_2)\cap X_1$ such that $x'_1 \in h_1(x_2)$. Let $X'_1$ denote the set of picked elements.
  Since $|X_1|>|X_2|$ then $h_1 \circ h_2$ is an $X'_1$-shop that preserves $\mathcal{D}$
  with $|X'_1|\leq |X_2|$.
  Diagrammatically, this can be written as,
  \begin{equation*}
    D \xrightarrow{h_2} X_2 \xrightarrow{h_1} X'_1 \subseteq h_1(X_2)\cap X_1 \subseteq X_1 \subseteq h_1 \circ h_2 (D).
  \end{equation*}
 
  This means that we may look for an  $X$-shop where
  the set $X$ is minimal with respect to inclusion, or equivalently, for a set
  with minimal size $|X|$.  So, in order to find an $X$-shop with
  a minimal set $|X|$, we may proceed greedily, removing elements from
  $D$ while we have an $X$-shop until we obtain a set $X$ such that
  there is no $X'$-shop for $X'\subsetneq X$.
  The dual argument applies to $U$-shops, and consequently to
  $U$-$X$-shops.

   This further explains why minimising $U$ and $X$, and then maximising their intersection, necessarily leads to a minimal $\widetilde{D}:=U \cup X$ also. Because, would we find $U' \cup X'$ of smaller size, we might look within $U'$ and $X'$ for potentially smaller sets of cardinality $|U|$ and $|X|$, thus contradicting minimality. 
\end{remark}

Note that the sets $U$ and $X$ are not necessarily unique.
However, as we shall see later \emph{the} $U$-$X$-core is unique up to
isomorphism (see Theorem~\ref{thm:unicity}). Moreover, within $\widetilde{\mathcal{D}}$, the sets $U$
and $X$ are uniquely determined. We delay until later the
proof of this second result (see Theorem~\ref{theo:uniquenessUXWhenFullDomain}).

\subsection{Uniqueness of the $U$-$X$-core}
\label{sec:uniqueness-u-x}
Throughout this section, let $\mathcal{D}$ be a finite structure and $\mathcal{M}$ its
associated DSM; i.e. $\mathcal{M}$ is the set of surjective
hyper-endomorphisms of $\mathcal{D}$. Let $U$ and $X$
be subsets of $D$ such that the substructure $\widetilde{D}$ of
$\mathcal{D}$ induced by $\widetilde{D}=U\cup X$ is a $U$-$X$-core of
$\mathcal{D}$. We will progress through various lemmata and
eventually derive the existence of a canonical $U$-$X$-shop in
$\mathcal{M}$ which will be used to prove that the $U$-$X$-core is
unique up to isomorphism. Uniqueness of the $U$-$X$-core has no real bearing on our
classification program but the canonical shop will allow us to
characterise all other shops in $\mathcal{M}$, which will be instrumental in
the hardness proofs for $\mylogic(\mathcal{D})$.

\begin{lemma}
  \label{lem:AtMostOneUInImage}
  Let $f$ be a shop in $\mathcal{M}$.
  For any element $z$ in $D$, $f(z)$ contains at most one element of
  the set $U$, that is $|f(z)\cap U|\leq 1$.
\end{lemma}
\begin{proof}
  Assume for contradiction that there is some $z$ and some distinct
  elements $u_1$ and $u_2$ of $U$ such that $f(z)\supseteq
  \{u_1,u_2\}$. Let $z_3,z_4,\ldots $ be any choice of antecedents
  under $f$ of the remaining elements $u_3,u_4,\ldots$ of $U$ (recall
  that $f$ is surjective). 
  By assumption the monoid $\mathcal{M}$ contains a $U$-shop $g$.
  Hence, $g\circ f$ would be a $U'$-shop with
  $U'=\{z,z_3,z_4,\ldots\}$ since $f(U')\subseteq U$ and $g(U)=D$.
  We get a contradiction as $|U'|<|U|$.
\end{proof}
\begin{lemma}
  \label{lem:UshopPermutesU}
  Let $f$ be a $U$-shop in $\mathcal{M}$.
  There exists a permutation $\alpha$ of $U$ such that: for any $u$ in $U$,
  \begin{romannum}
  \item $f(u)\cap U = \{\alpha(u)\}$; and,
  \item $f^{-1}(u)\cap U = \{\alpha^{-1}(u)\}$.
  \end{romannum}
\end{lemma}
\begin{proof}
  It follows from Lemma~\ref{lem:AtMostOneUInImage} that 
  for any $u$ in $U$, $|f(u)\cap U|\leq 1$.
  Since $f$ is a $U$-shop, every element in $D$ has an antecedent in $U$
  under $f$ and thus in particular for any $u$ in $U$, $|f^{-1}(u)\cap
  U|\geq 1$. Note that if some element of $U$ had no image in $U$ then
  as $U$ is finite, we would have an element of $U$ with two distinct
  images in $U$. Hence, for any $u$ in $U$, $|f(u)\cap U|= 1$ and
  the result follows.
\end{proof}

The dual statements concerning $X$-shops hold.
\begin{lemma}
  \label{lem:AtMostOneXInPreImage}
  Let $f$ be a shop in $\mathcal{M}$.
  for any element $z$ in $D$, $f^{-1}(z)$ contains at most one element of
  the set $X$, that is $|f^{-1}(z)\cap X|\leq 1$.
\end{lemma}
\begin{proof}
  By duality from Lemma~\ref{lem:AtMostOneUInImage}.
\end{proof}
\begin{lemma}
  \label{lem:XshopPermutesX}
  Let $f$ be an $X$-shop in $\mathcal{M}$.
  There exists a permutation $\beta$ of $X$ such that: for any $x$ in $X$,
  \begin{romannum}
  \item $f(x)\cap X = \{\beta(x)\}$; and,
  \item $f^{-1}(x)\cap X = \{\beta^{-1}(x)\}$.
  \end{romannum}
\end{lemma}
\begin{proof}
  By duality from Lemma~\ref{lem:UshopPermutesU}.
\end{proof}
\begin{lemma}
  \label{fact:UXShopMapsNoProperXToProperU}
  \label{cor:UXShopMapsNoXToProperU}
  Let $f$ be a shop in $\mathcal{M}$.
  If $f$ is a $U$-$X$-shop then $f(X)\cap (U\setminus X)=\emptyset$.
\end{lemma}
\begin{proof}
 Assume for contradiction that for some $x_1 \in X$ and some $u_1 \in U\setminus X$, we have $u_1 \in f(x_1)$. 
  Since $f$ is an $X$-shop, every element is an antecedent under $f$ of some
  element in $X$, in particular every element $x_2,x_3,\ldots \in X$ (different from
  $x_1$) has a unique image $x'_2,x'_3,\ldots \in X$ (see
  Lemma~\ref{lem:XshopPermutesX}). 
  Some element of $X$, say $x_i$ does not occur in these images.
  Necessarily, $x_1$ reaches $x_i$.
  Note that $x_i$ can not also belong to $U$ as
  otherwise, $x_i$ and $u_1$,  two distinct elements of $U$, would
  be reached by $x_1$, contradicting Lemma~\ref{lem:AtMostOneUInImage}.
  Thus, we must have that $x_i$ belongs to $X\setminus U$. 
  Let $U':=U$ and $X':=X \setminus \{x_i\} \cup \{u_1\}$.
  Note that $f^2:=f\circ f$, the second iterate of $f$, is a
  $U'$-$X'$-shop with $|U'|= |U|$, $|X'|= |X|$ and $|U'\cap
  X'|<|U\cap X|$. This contradicts our hypothesis on $U$ and $X$.
\end{proof}

\begin{proposition}%[\textbf{canonical $U$-$X$-shop}]
  \label{prop:ShopWithIdentity}
  Let $\mathcal{M}$ be a DSM over a set $D$ and $U$ and $X$ be minimal subsets of $D$
  such that: there is a $U$-shop in $\mathcal{M}$; there is an
  $X$-shop in $\mathcal{M}$ and $U\cup X$ is minimal.
  Then, there is a $U$-$X$-shop $h$ in $\mathcal{M}$ 
  %that may act as the identity on $U\cup X$ and 
  that has the following properties:
   \begin{romannum}
  \item for any $y$ in $U\cap X$, $h(y)\cap (U\cup X)=\{y\}$;
  \item for any $x$ in $X\setminus U$, $h(x)\cap (U\cup X)=\{x\}$;
  \item for any $u$ in $U\setminus X$, $h(u)\cap (U\cup X)=\{u\}\cup X_u$, where
    $X_u\subseteq X\setminus U$; and,
  \item $h(U\setminus X)\cap X=\bigcup_{u \in U\setminus X} X_u = X\setminus U.$
  \end{romannum}
\end{proposition}
\begin{proof}
  By assumption, $\mathcal{M}$ contains a $U$-$X$-shop $f$.
  Let $\alpha$ and $\beta$ be permutations of $U$ and $X$,
  respectively, as in Lemmata~\ref{lem:UshopPermutesU} and~\ref{lem:XshopPermutesX}. 
  Let $r$ be the least common multiple of the order of the
  permutations $\beta$ and $\alpha$. We set $h$ to be the $r$th iterate
  of $f$ and we now know that $h(z)\ni z$ for any element $z$ and that
  $h$ is a $U$-$X$-shop by~\ref{lem:compositionAndUXshops}.
  Let $y$ in $U\cap X$, we know that $h(y)\ni y$.
  We can not have another element from $U\cup X$ in $h(y)$ by
  Lemmata~\ref{lem:AtMostOneUInImage} and~\ref{lem:XshopPermutesX}. 
  This proves (i).
  %%%
  Let $x$ in $X\setminus U$, we know that $h(x)\ni x$.
  We can not have an element from $X$ distinct from $x$ in $h(x)$ by
  Lemma~\ref{lem:XshopPermutesX} and we can not have an element from
  $U\setminus X$ in $h(x)$ by
  Lemma~\ref{fact:UXShopMapsNoProperXToProperU}. 
  This proves (ii).
  %%%
  Let $u$ in $U\setminus X$, we know that $h(u)\ni u$.
  We can not have an element from $U$ distinct from $u$ in $h(u)$ by
  Lemma~\ref{lem:UshopPermutesU}. We may have however some elements from
  $X\setminus U$ in $h(u)$. Thus, there is a set $\emptyset \subseteq
  X_u \subseteq X\setminus U$ such that $h(u)\cap (U\cup X)=\{u\} \cup
  X_u$. This proves (iii).
  %%%
  By construction $h$ is a $U$-shop and every element must have an
  antecedent in $U$ under $h$. Since by the first three points,
  elements from $X\setminus U$ can only be reached from elements in $U
  \setminus X$, the last point (iv) follows.
\end{proof}

\begin{theorem}
\label{thm:unicity}
  The $U$-$X$-core is unique up to isomorphism.  
\end{theorem}
\begin{proof}
Let $h_1$ be a $U_1$-$X_1$-shop with minimal $|U_1|$, $|X_1|$ and $|U_1\cup
X_1|$ and let $h_2$ be a $U_2$-$X_2$-shop  with minimal $|U_2|$, $|X_2|$
and $|U_2\cup X_2|$. Hence, $h_1 \circ h_2$ is a $h_1(X_2)\cap
X_1$-shop with $|h_1(X_2)|\leq |X_1|$. By minimality of $X_1$,
$|h_1(X_2)|=|X_1|$, and the restriction of $h_1$ to domain $X_2$ and
codomain $X_1$ induces a surjective homomorphism from the substructure
induced by $X_2$ to the substructure induced by $X_1$. Similarly $h_2$
induces a surjective homomorphism in the other direction. As we work
with finite structures, $h_1$ induces an isomorphism $i$ from the substructure
induced by $X_1$ to the substructure induced by $X_2$.
By duality, we also get that $h_1$ induces an isomorphism $i'$ from the substructure
induced by $U_1$ to the substructure induced by $U_2$. 
By construction, $i$ and $i'$ agree on $U_1\cap X_1$ (necessarily
to $U_2\cap X_2$) and the result follows.
\end{proof}

%%%%%%%%%%%%%%%%%%%%%%%%%%%%%%%%%%%%%%%%%%%%%%%%%%%%%%%%%%%%%%%%%%%%%%%%
\section{Complexity classification}
\label{sec:compl-class}
%%% Schaeffer for Boolean CSP and Boolean QCSP.
% A structure or relation is \emph{Boolean} if its domain is of size
% $2$. 
We recall below some well known results concerning the complexity of Boolean
\CSP\ and \QCSP\ which we will need later. For definitions and further details
regarding the proof, the reader may consult the nice survey by Chen~\cite{DBLP:journals/csur/Chen09}.

\begin{theorem}%[\textbf{Schaefer's dichotomy for generalised SAT}~\cite{Schaefer}]
  %[Schaefer~\cite{Schaefer}]
  [\cite{Schaefer}]
  \label{theorem:SchaeferCSP}
Let $\mathcal{D}$ be a Boolean structure. Then $\CSP(\mathcal{D})$
(equivalently, $\csplogic(\mathcal{D})$) is in \Ptime\ if, and only
if, all relations of $\mathcal{D}$ are simultaneously $0$-valid,
$1$-valid, Horn, dual-Horn, bijunctive or affine, and otherwise it is
\NP-complete. 
\end{theorem}

A similar result holds when universal quantifiers are added to the mix.
(it was sketched in the presence of constants~\cite{Schaefer}, then proved in the absence of constants in~\cite{Nadia} and~\cite{DalmauLSI9743R}).
\begin{theorem}%[\textbf{Dichotomy for generalised Quantified SAT}[~\cite{Nadia}]
\label{theorem:SchaeferQCSP}
Let $\mathcal{D}$ be a Boolean structure. Then $\QCSP(\mathcal{D})$
(equivalently, $\qcsplogic(\mathcal{D})$) is in \Ptime\ if, and only
if, all relations of $\mathcal{D}$ are simultaneously Horn, dual-Horn, bijunctive or affine, and otherwise it is
\Pspace-complete. 
\end{theorem}
\begin{example}
  The canonical example of a relation that does not fall in any of the
  tractable cases is $\mbox{\textsc{NAE}} := \{0, 1\}^3 \setminus \{(0, 0, 0), (1, 1, 1)\}$.
  Let $\mathcal{B}_{\mbox{\textsc{nae}}}$ be the Boolean structure with this relation.
  It follows from the above theorems that
  $\CSP(\mathcal{B}_{\mbox{\textsc{nae}}})$ is \NP-complete and that
  $\QCSP(\mathcal{B}_{\mbox{\textsc{nae}}})$ is \Pspace-complete.
\end{example}
\begin{example}
  For larger domains, though the classification remains open,
  the canonical hard problem is induced by the relation $\neq$.
  Let $\mathcal{K}_n$ denote the clique of size $n$ (we view an
  undirected graph as a structure with a single binary predicate $E$
  that is symmetric). For $n\geq 3$,
  $\CSP(\mathcal{K}_n)$ is a reformulation of the $n$-colourability
  problem and is \NP-complete. 
  It is also known that for $n\geq 3$
  $\QCSP(\mathcal{K}_n)$ is \Pspace-complete~\cite{DBLP:journals/iandc/BornerBCJK09}.
\end{example}

\subsection{First Class}
\label{sec:first-class}
\begin{proposition}%[\textbf{trivial cases}]
  \label{prop:oneElementIsTrivial}\label{prop:firstClass:complexity}
  \begin{romannum}
  \item When $\mathcal{D}$ has a single element, the model checking problem for \FO{} is in \Logspace.
  \item The model checking problem $\{\exists, \lor, \neq,= \}
    \mbox{-}\mathrm{FO}$ is in \Logspace.
  \end{romannum}
\end{proposition}
\begin{proof}
  \begin{longromannum}
  \item In the case where $|D|=1$, every relation is either empty or contains all tuples (one tuple), and the quantifiers $\exists$ and $\forall$ are semantically equivalent. Hence, the problem translates to the Boolean Sentence Value Problem (under the substitution of $0$ and $1$ for the empty and non-empty relations, respectively), known to be in \Logspace{}~\cite{DBLP:journals/jacm/Lynch77}.
  \item We may assume by the previous point that $|D|>1$.
    We only need to check if one of the atoms that occurs as a disjunct in the input
    sentence holds in $\mathcal{D}$. Since $|D|>1$, a sentence with an atom like
    $x=y$ or $x\neq y$ is always true in $\mathcal{D}$. For sentences
    of $\{\exists, \lor\} \mbox{-}\mathrm{FO}$, the atoms may have
    implicit equality as in $R(x,x,y)$ for a ternary predicate $R$: in any case,
    each atom may be checked in constant time since $\mathcal{D}$ is a fixed
    structure, resulting in overall logspace complexity.
    %%% In fact this is logspace even if $D$ is not fixed.
    \qed
  \end{longromannum}
\end{proof}

\subsection{Second Class}
\label{sec:second-class}
%%% Second class.
%%% Largest fragments.
We now move on to the largest fragments we will need to consider,
which turn out to exhibit trivial dichotomies.

\begin{proposition}%[\textbf{Complexity of $\posFoNeq(\mathcal{D})$}]
  \label{proposition:posFoNeq:complexity}
  In full generality, the class of problems
  $\posFoNeq(\mathcal{D})$ exhibits dichotomy: if $|D|=1$ then the
  problem is in \Logspace, otherwise it is \Pspace-complete. 
  Consequently, the fragment extended with $=$ follows the same dichotomy.
\end{proposition}
\begin{proof}
  When $|D| \geq 2$, \Pspace-hardness may be proved using no extensional
  relation of $D$ other than $\neq$. The formula
  $\varphi_{\mathcal{K}_{|D|}}(x,y):= (x\neq y)$ simulates the edge
  relation of the clique $\mathcal{K}_{|\mathcal{D}|}$ and the problem
  $\qcsplogic(\mathcal{K}_n)$ better known as $\QCSP(\mathcal{K}_n)$ is
  \Pspace-complete for $n\geq 3$~\cite{DBLP:journals/iandc/BornerBCJK09}.
  For n = 2, we use a reduction from the problem
  $\QCSP(\mathcal{B}_{\mbox{\textsc{nae}}})$ to prove that $\qcspDisj(\mathcal{K}_2)$ is
  \Pspace-complete. 
  Let $\varphi$ be an input for
  $\QCSP(\mathcal{B}_{\mbox{\textsc{nae}}})$. Let $\varphi'$ be built
  from $\varphi$ by substituting all 
  instances of $\mathrm{NAE}(x,y,z)$ by $E(x,y) \lor E(y,z) \lor E(x,z)$. It is easy to see
  that $\mathcal{B}_{\mbox{\textsc{nae}}}\models \varphi$ iff $\mathcal{K}_2\models \varphi'$, and the
  result follows.

  Note that we have not used $=$ in our hardness proof; and, in the case $|D|=1$, we
  may allow $=$ without affecting tractability (triviality). Thus, the
  fragment extended with $=$ follows the same delineation.
\end{proof}

\begin{proposition}%[\textbf{Complexity of $\EqFreeFo(\mathcal{D})$}]
  \label{proposition:EqFreeFo:complexity}
  In full generality, the class of problems  $\EqFreeFo(\mathcal{D})$ exhibits
  dichotomy: if all relations of $\mathcal{D}$ are \emph{trivial}
  (either empty or contain all tuples) then the problem is in
  \Logspace, otherwise it is \Pspace-complete. 
\end{proposition}
\begin{proof}% \hspace{-2mm}\footnote{
% In~\cite{DBLP:conf/stoc/Vardi82}, it is claimed that the \Pspace-hard cases, which entail the \Pspace-hard cases of Proposition~\ref{proposition:posFoNeq:complexity}, are proved in~\cite{DBLP:conf/stoc/ChandraM77}. We are unable to find such, and are in some doubt as to what would be an appropriate reference. Certainly Proposition~\ref{proposition:posFoNeq:complexity} qualifies as folklore, having been casually mentioned in~\cite{KolaitisTalk}.
% } 
 
  If all relations are trivial then $\sim$ has a single equivalence
  class. Thus, $\mathcal{D}/\!\!\sim$ has a single element. By
  $\EqFreeFo$-equivalence of $\mathcal{D}$ and $\mathcal{D}/\!\!\sim$ 
  (see Proposition~\ref{prop:EqFreeFoContainment}), it suffices to check
  whether an input $\varphi$ in \EqFreeFo\ holds in
  $\mathcal{D}/\!\!\sim$. Since the latter has a single element, the
  problem is in \Logspace.
  % Alternatively, one should realise that all branches of the game-tree
  % are $\sim$-equivalent. ... 

  Otherwise, the equivalence relation $\sim$ has
  at least 2 equivalence classes since $\mathcal{D}$ is non
  trivial. We may now follow the same proof as in
  Proposition~\ref{proposition:posFoNeq:complexity}, using the negation of $\sim$ in
  lieu of $\neq$, and \Pspace-hardness follows.  
\end{proof}

We proceed with the last two fragments of the second class. 
\begin{proposition}%[\textbf{Complexity of $\cspDisj(\mathcal{D})$ and $\cspDisjEq(\mathcal{D})$}]
  \label{prop:cspDisj:complexity}
  In full generality, the class of problems
  $\cspDisj(\mathcal{D})$ exhibits dichotomy: if the core of
  $\mathcal{D}$ has one element then the problem is in \Logspace,
  otherwise it is \NP-complete.
  As a corollary, the class of problems $\cspDisjEq(\mathcal{D})$
  exhibits the same dichotomy.
\end{proposition}

In our preliminary work~\cite{DBLP:conf/cie/Martin08,DBLP:journals/corr/abs-cs-0609022},
the proof of the above is combinatorial and appeals to Hell and
\Nesetril's dichotomy theorem for undirected
graphs~\cite{HellNesetril}. An alternative proof of this result also appeared in \cite{DBLP:conf/cie/BodirskyHR09} (and to a lesser extent in \cite{DBLP:conf/walcom/HermannR09}). We will give here an algebraic 
proof which uses a variant of the Galois connection $\Inv-\End$ due to
Krasner~\cite{KrasnerJMPA1938} for $\cspDisjEq$. The variant for
$\cspDisj$ involves hyper-endomorphisms rather than endomorphisms
because of the absence of equality. A hyper-endomorphism of
$\mathcal{B}$ is a function from $B$ to the power-set of $B$ that
is total and preserving (see Definition~\ref{def:HeAndShe}).
Our purpose is to provide both a self-contained proof and a gentle
introduction to the techniques we shall use for the fragment \mylogic.

For a set $F$ of hyper-endomorphisms on the finite domain $B$, let $\Inv(F)$
be the set of relations on $B$ of which each $f$ in $F$ is an
hyper-endomorphism (when these relations are viewed as a structure over
$B$). We say that $S \in \Inv(F)$ is invariant or \emph{preserved} by
(the hyper-endomorphisms in) $F$. 
Let $\hE(\mathcal{B})$ be the set of hyper-endomorphisms of $\mathcal{B}$. 
Let $\langle \mathcal{B} \rangle_{\cspDisj}$ be the set of
relations that may be defined on $\mathcal{B}$ in \cspDisj.

\begin{lemma}
\label{lemma:EqFreeWeakKrasnerconnection-by-types}
Let $\tuple{r}:=(r_1,\ldots,r_k)$ be a $k$-tuple of elements of $\mathcal{B}$. There exists
a formula $\theta_{\tuple{r}}^{\cspDisj}(u_1,\ldots,u_k) \in \cspDisj$ such that
the following are equivalent.
\begin{romannum}
\item \label{EqFreeWeakKrasnerconnection-by-types:canonicalSentence}
  $(\mathcal{B}, r'_1,\ldots, r'_k) \models
  \theta_{\tuple{r}}^{\cspDisj}(u_1,\ldots,u_k)$.
\item \label{EqFreeWeakKrasnerconnection-by-types:hyperEndo}
  There is a hyper-endomorphism from $(\mathcal{B}, r_1,\ldots, r_k)$ to $(\mathcal{B},
  r'_1,\ldots, r'_k)$.
\end{romannum}
\end{lemma}
\begin{proof}
  
  let $\mathbf{s}:=(b_1,\ldots,b_{|B|})$ an enumeration of the elements of
  $\mathcal{B}$ and $\phi_{\mathcal{B}(\mathbf{r},\mathbf{s})}(v_1,\ldots,v_{|B|})$ be the
  associated conjunction of positive facts. Set 
  $$\theta_{\tuple{r}}^{\cspDisj}(u_1,\ldots,u_k):= \exists v_1,\ldots,v_{|B|} \ \phi_{\mathcal{B}(\mathbf{r},\mathbf{s})}(v_1,\ldots,v_{|B|}).$$

  The forward direction is clear as the witness $s'_1,\ldots,s'_{|B|}$ for
  $v_1,\ldots,v_{|B|}$ provides a hyper-endomorphism $f$ defined as
  $f(b_i) \ni s'_i$ and $f(r_i) \ni r'_i$.

  For the backwards direction, one may build an endomorphism from $(\mathcal{B}, r_1,\ldots, r_k)$ to $(\mathcal{B},
  r'_1,\ldots, r'_k)$ from the given hyper-endomorphism. The result follows
  from the implication from (i) to (iv) of Proposition~\ref{prop:containment:cspDisjEq}. 
\end{proof}

\begin{theorem}%[\textbf{Galois Connection \Inv-\hE}]\\
  \label{theorem:galois-connectionEqFreeWeakKrasner}
  For a finite structure $\mathcal{B}$ we have
  $\langle \mathcal{B} \rangle_{\cspDisj} = \Inv(\hE(\mathcal{B}))$.
\end{theorem}
\begin{proof}
  Let $\varphi(\tuple{v})$ be a formula of $\cspDisj$ with free variables
  $\tuple{v}$. We denote also by $\varphi(\tuple{v})$ the relation
  induced over $\mathcal{B}$.
  \begin{longenum}
  \item $\varphi(\tuple{v}) \in \langle \mathcal{B}
    \rangle_{\cspDisj} \ \Rightarrow \ \varphi(\tuple{v}) \in
    \Inv(\hE(\mathcal{B}))$.  
    This is proved by induction on the complexity of
    $\varphi(\tuple{v})$. 

    (Base Cases.)
      When $\varphi(\tuple{v}):=R(\tuple{v})$, the variables
      $\tuple{v}$ may appear multiply in $R$ and in any order. Thus
      $R$ is an instance of an extensional relation under substitution
      and permutation of positions. The result follows directly from
      the definition of hyper-endomorphisms.

      (Inductive Step.) There are three subcases. We progress through
      them in a workmanlike fashion. Take $f \in \hE(\mathcal{B})$.
      \begin{enumerate}
      \item
        $\varphi(\tuple{v}):=\psi(\tuple{v}) \land
        \psi'(\tuple{v})$.  Let
        $\tuple{v}:=(v_1,\ldots,v_l)$. Suppose $\mathcal{B} \models
        \varphi(x_1,\ldots,x_l)$; then both $\mathcal{B} \models
        \psi(x_1,\ldots,x_l)$ and $\mathcal{B} \models
        \psi'(x_1,\ldots,x_l)$. By Inductive Hypothesis (IH), for any
        $y_1\in f(x_1),\ldots,$ $y_l\in f(x_l)$, both $\mathcal{B} \models
        \psi(y_1,\ldots,y_l)$ and $\mathcal{B} \models
        \psi'(y_1,\ldots,y_l)$, whence $\mathcal{B} \models
        \varphi(y_1,\ldots,y_l)$.

      \item
        $\varphi(\tuple{v}):=\psi(\tuple{v}) \lor
        \psi'(\tuple{v})$. Let $\tuple{v}:=(v_1,\ldots,v_l)$. Suppose
        $\mathcal{B} \models \varphi(x_1,\ldots,x_l)$; then one of
        $\mathcal{B} \models \psi(x_1,\ldots,x_l)$ or $\mathcal{B}
        \models \psi'(x_1,\ldots,x_l)$; w.l.o.g. the former. By IH,
        for any $y_1\in f(x_1),\ldots,y_l\in f(x_l)$, $\mathcal{B} \models
        \psi(y_1,\ldots,y_l)$, whence $\mathcal{B} \models
        \varphi(y_1,\ldots,y_l)$.

      \item  
        $\varphi(\tuple{v}):=\exists w \ \psi(\tuple{v},w)$. Let
        $\tuple{v}:=(v_1,\ldots,v_l)$. Suppose $\mathcal{B} \models
        \exists w \ \psi(x_1,\ldots,x_l,w)$; then for some $x'$,
        $\mathcal{B} \models \psi(x_1,\ldots,x_l,x')$. By IH, for any
        $y_1\in f(x_1),\ldots,y_l\in f(x_l),y'\in f(x')$, $\mathcal{B} \models
        \psi(y_1,\ldots,y_l,y')$, whereupon $\mathcal{B} \models
        \exists w \ \psi(y_1,\ldots,y_l,w)$.
      \end{enumerate}
  \item $S \in \Inv(\hE(\mathcal{B})) \ \Rightarrow \ S \in \langle
    \mathcal{B} \rangle_{\cspDisj}.$ 
    Consider the $k$-ary relation $S
    \in \Inv(\hE(\mathcal{B}))$. Let $\tuple{r}_1,\ldots,\tuple{r}_m$
    be the tuples of $S$. Set $\theta_S^{\cspDisj}(u_1,\ldots,u_k)$ to
    be the following formula of \cspDisj:
    \[
    \theta_{\tuple{r}_1}^{\cspDisj}(u_1,\ldots,u_k) \vee \ldots \vee
    \theta_{\tuple{r}_m}^{\cspDisj}(u_1,\ldots,u_k). \]     
    For
    $\tuple{r}_i:=(r_{i1},\ldots,r_{ik})$, note that
    $(\mathcal{B},r_{i1},\ldots,r_{ik}) \models
    \theta_{\tuple{r}_i}^{\cspDisj}(u_1,\ldots,u_k)$. That
    $\theta_S(u_1,\ldots,u_k)=S$ now follows from Part (ii) of
    Lemma~\ref{lemma:EqFreeWeakKrasnerconnection-by-types}, since $S
    \in \Inv(\hE(\mathcal{B}))$.
  \end{longenum}
\end{proof}
%%%%%%%%%%%%%%%%%%%%%%%%%%%%%%%%%%%%%%%%%%%%%%%
%%% Use this proof if weak Krasner is proved first.
%%%%%%%%%%%%%%%%%%%%%%%%%%%%%%%%%%%%%%%%%%%%%%%
% \begin{proof}
%   Let $\varphi(\tuple{v})$ be a formula of $\cspDisj$ with free variables
%   $\tuple{v}$. We denote also by $\varphi(\tuple{v})$ the relation
%   induced over $\mathcal{B}$.
  
%   \begin{enumerate}
%   \item $\varphi(\tuple{v}) \in \langle \mathcal{B}
%     \rangle_{\cspDisjEq} \ \Rightarrow \ \varphi(\tuple{v}) \in
%     \Inv(\End(\mathcal{B}))$.  

%     This is proved by induction on the
%     complexity of $\varphi(\tuple{v})$ as in the proof of
%     Theorem~\ref{theorem:galois-connectionKrasner} except that: we no
%     longer need to preserve equality; and, we now deal with
%     hyper-endomorphisms and need to replace occurrences of $y=f(x)$
%     by $y\in f(x)$ in the proof. 

%   \item $S \in \Inv(\hE(\mathcal{B})) \ \Rightarrow \ S \in \langle
%     \mathcal{B} \rangle_{\cspDisj}.$ 

%     This is proved as in the proof of
%     Theorem~\ref{theorem:galois-connectionKrasner} except that we use
%     the formula, 
%     \[ \theta_S^{\cspDisj}(u_1,\ldots,u_k) \ := \
%     \theta_{\tuple{r}_1}^{\cspDisj}(u_1,\ldots,u_k) \vee \ldots \vee
%     \theta_{\tuple{r}_m}^{\cspDisj}(u_1,\ldots,u_k). \] 
%     The result follows from from Part (ii) of
%     Lemma~\ref{lemma:EqFreeWeakKrasnerconnection-by-types}.
%   \end{enumerate}
% \end{proof}
\noindent Note that in \cite{DBLP:conf/walcom/HermannR09} it is erroneously claimed that endomorphisms (not hyper-endomorphisms) are the correct algebraic object for the fragment $\cspDisj$ -- this is not correct and only holds for the richer fragement $\cspDisjEq$.
\begin{corollary}
\label{cor:EqFreeKrasner-reduction}
Let $\mathcal{B}$ and $\mathcal{B}'$ be finite structures over the same domain $B$.
 If $\hE(\mathcal{B}) \subseteq \hE(\mathcal{B}')$ then $\cspDisj(\mathcal{B}') \leq_{\Logspace} \cspDisj(\mathcal{B})$.
\end{corollary}
\begin{proof}
If $\hE(\mathcal{B}) \subseteq \hE(\mathcal{B}')$, then
$\Inv(\hE(\mathcal{B}')) \subseteq \Inv(\hE(\mathcal{B}))$. From
Theorem~\ref{theorem:galois-connectionEqFreeWeakKrasner}, it follows that
$\langle \mathcal{B}' \rangle_{\cspDisj}$ $\subseteq \langle
  \mathcal{B} \rangle_{\cspDisj}$. Recalling that $\mathcal{B}'$
  contains only a finite number of extensional relations, we may
  therefore effect a Logspace reduction from $\cspDisj(\mathcal{B}')$
  to $\cspDisj(\mathcal{B})$ by straightforward substitution of
  predicates. 
\end{proof}

% \begin{proposition}\label{prop:cspDisj:complexity} 
%   In full generality, the class of problems
%   $\cspDisj(\mathcal{D})$ exhibits dichotomy: if the core of
%   $\mathcal{D}$ has one element then the problem is in \Logspace, otherwise it is \NP-complete.
% \end{proposition}
\begin{proof}[of Proposition~\ref{prop:cspDisj:complexity}]
  By Proposition~\ref{prop:equivalence:cspDisjEq}, we may assume
  w.l.o.g. that $\mathcal{D}$ is a core. 
  This means that every hyper-endomorphism of $\mathcal{D}$ is in
  fact an automorphism -- we identify hyper-endomorphisms whose range
  are singletons with automorphisms -- and
  thus $\hE(\mathcal{D})$ is a subset of $S_n$ where $n=|D|$.
  If $D$ has one element, then the problem is trival. If $D$ has two
  elements, then
  $\hE(\mathcal{D})\subseteq\hE(\mathcal{B}_{\mbox{\textsc{nae}}})=S_2$. By
  Lemma~\ref{cor:EqFreeKrasner-reduction}, it follows that
  $\cspDisj(\mathcal{B}_{\mbox{\textsc{nae}}}) \leq_{\Logspace}
  \cspDisj(\mathcal{D})$.
  Since the former is a generalisation of the \NP-complete
  $\CSP(\mathcal{B}_{\mbox{\textsc{nae}}})$, the latter is
  \NP-complete.
  If $\mathcal{D}$ has $n\geq2$ elements, we proceed similarly
  with $\mathcal{K}_n$.
\end{proof}

\begin{proposition}%[\textbf{Complexity of \cspDisjNeq}]
  \label{prop:cspDisjNeq:complexity}
  In full generality, the class of problems
  $\cspDisjNeq(\mathcal{D})$ exhibits dichotomy: if $|D|=1$ then the
  problem is in \Logspace, otherwise it is \NP-complete. 
  Consequently, the fragment extended with $=$ follows the same dichotomy.
\end{proposition}
\begin{proof}
  The proof is similar to that of Proposition~\ref{proposition:posFoNeq:complexity}. 
  Let $|D|=n$.
  The inequality symbol $\neq$ allows to simulate $\mathcal{K}_n$. 
  When $n=2$, using disjunction we may simulate $\mathcal{B}_{\mbox{\textsc{nae}}}$.
  \NP-completeness follows by reduction from $\CSP(\mathcal{K}_n)$
  when $n\geq3$ and from $\CSP(\mathcal{B}_{\mbox{\textsc{nae}}})$
  when $n=2$.
  Note that equality is not used in our hardness proof and may
  trivially be allowed when $|D|=1$. Thus, the
  classification is the same whether one allows $=$ or not.
\end{proof}

All fragments of the second class follow a natural dichotomy.
\begin{corollary}\label{Secondclass:complexity:oneElementCore}
  For any syntactic fragment $\mathscr{L}$ of \FO\ in the second
  class, the model checking problem $\mathscr{L}(\mathcal{D})$ is
  trivial (in \Logspace) when the $\mathscr{L}$-core of $\mathcal{D}$
  has one element and hard otherwise (\NP-complete for existential
  fragments, \Pspace-complete for fragments containing both quantifiers).
\end{corollary}

%%%%%%%%%%%%%%%%%%%%%%%%%%%%%%%%%%%%%%%%%%%%%%%%%
\subsection{Third Class}
\label{sec:third-class}
%%% third class.
%%% "Boolean CSP" "Boolean QCSP".

\begin{proposition}%[\textbf{Complexity of
                   %$\csplogicNeq(\mathcal{D})$}]
\label{prop:csplogicNeq:complexity}
  In full generality, the problem $\csplogicNeq(\mathcal{D})$ is in \Logspace{} if
  $|D|=1$, in \Ptime{} if $|D|=2$ and $\mathcal{D}$ is bijunctive or
  affine, and \NP-complete otherwise.
  The fragment extended with $=$ follows the same dichotomy.
\end{proposition}
\begin{proof}
  We classify first the fragment extended with $=$.
  When $|D|\geq 3$, we may use $\neq$ to simulate
  $\CSP(\mathcal{K}_{|D|})$ which is \NP-complete.
  When $|D|=1$ the problem is trivially in \Logspace.
  We are left with the Boolean case. Let $\mathcal{D}_{\neq}$ denote the extension of
  $\mathcal{D}$ with $\neq$. Note that $\csplogicNeq(\mathcal{D})$
  coincides with $\csplogic(\mathcal{D}_{\neq})$ which is
  the Boolean $\CSP(\mathcal{D}_{\neq})$. We apply Schaefer's theorem. 
  The relation $\neq$ is neither Horn, nor dual-Horn, nor $0$-valid
  nor $1$-valid as it is not closed under any of the following Boolean
  operations: $\land$, $\lor$, $c_0$ or $c_1$ (the constant functions
  $0$ and $1$). The relation $\neq$ is both bijunctive and affine as
  it is closed under both the Boolean majority and minority operation
  (see Chen's survey for the
  definitions~\cite{DBLP:journals/csur/Chen09}).
  Consequently, $\csplogicNeq(\mathcal{D})$ is in \Ptime{} if
  $\mathcal{D}$ is bijunctive or affine and \NP-complete otherwise.

  Note that we have not used $=$ in the hardness proof when $|D|\geq 3$.
  When $|D|=2$, we appeal to Schaefer's theorem
  (Theorem~\ref{theorem:SchaeferCSP}), the proof of which relies on
  the Galois connection $\Pol-\Inv$ which assumes presence of
  $=$. However, the hardness proofs in Schaefer's theorem rely on logical reductions from
  $\csplogic(\mathcal{B}_{\mbox{\textsc{nae}}})$, which use
  definability of $\mathcal{B}_{\mbox{\textsc{nae}}}$ in $\csplogic$. Hence, our
  claim follows for the fragment $\csplogicNeq$.
\end{proof}

\begin{proposition}%[\textbf{Complexity of
                   %$\qcsplogicNeq(\mathcal{D})$}]
  \label{prop:qcsplogicNeq:complexity}
  In full generality, the problem $\qcsplogicNeq(\mathcal{D})$ is in \Logspace{} if
  $|D|=1$, in \Ptime{} if $|D|=2$ and $\mathcal{D}$ is bijunctive or
  affine, and \Pspace-complete otherwise.
  The fragment extended with $=$ follows the same dichotomy.
\end{proposition}
\begin{proof}
  This is similar to Proposition~\ref{prop:csplogicNeq:complexity}.
  When $|D|\geq 3$, we may use $\neq$ to simulate
  $\QCSP(\mathcal{K}_{|D|})$ which is \Pspace-complete. In the Boolean case,
  we apply Theorem~\ref{theorem:SchaeferQCSP} to
  $\qcsplogic(\mathcal{D}_{\neq})$ and the result follows.

  Again equality is not used to prove hardness and the result follows
  for the fragment without $=$.
\end{proof}

%%% TO DO coin a little sentence to convey the method e.g.: 
The case of \csplogic\ and \csplogicEq\ almost coincide as equality may be
propagated out by substitution, and every sentence of the latter is logically
equivalent to a sentence of the former, with the exception of
sentences using only $=$ as an extensional predicate like $\exists x
\, x=x$ which are tautologies as we only ever consider structures with at
least one element.
In the case of \qcsplogicEq,
some equalities like
$\exists x \exists y \, x=y$ and $\forall x \exists y \, x=y$
may also be propagated out by substitution. However, equalities like 
$\exists x \forall y \, x=y$ and $\forall x \forall y \, x=y$ can not,
but they hold only in structures with a single element. This technical
issue does not really affect the complexity classification, and it
would suffice to consider \csplogic\ and \qcsplogic. 
The  complexity classification for these
four fragments remain open and correspond to the dichotomy conjecture
for \CSP\ and the classification program of the \QCSP.
In practice, we like to pretend that equality is present as it provides a better behaved algebraic
framework, without affecting complexity. 

This leaves the fragment $\mylogic$ from our fourth class, which we
deal with in the remainder of this paper.
%%%%%%%%%%%%%%%%%%%%%%%%%%%%%%%%%%%%%%%%%%%%%%%%%%%%%%%%%%%%%%%%%%%%%%%%%%%%%%%%%%%%%%%%%%%%

%%%%%%%%%%%%%%%%%%%%%%%%%%%%%%%%%%%%%%%%%%%%%%%%%%%%%%%%%%%%%%%%%%%%%%%%%%%%%%%%%%%%%%%%%%%%

\section{Tetrachotomy of $\mylogic(\mathcal{D})$}
\label{sec:tetrachotomy}
The following -- left as a conjecture at the end
of~\cite{DBLP:journals/tocl/MadelaineM12,DBLP:conf/csl/MartinM10} --
is the main contribution of this paper. 
Recall first that a shop $f$ over a set $D$ is an \emph{$A$-shop} if there is an
element $u$ in $D$ such that $f(u)=D$; and, that $f$ is an
\emph{$E$-shop} if there is an element $x$ of $D$ such
that $f^{-1}(x)=D$.
\begin{theorem}%[\textbf{Tetrachotomy}]
\label{tetrachotomy}
Let $\mathcal{D}$ be any structure. 
\begin{longromannum}
\item[\textup{I}.] If $\mathcal{D}$ is preserved by both an A-shop and an E-shop, then $\mylogic(\mathcal{D})$ is in L.
\item[\textup{II}.] If $\mathcal{D}$ is preserved by an A-shop but
  is not preserved by any E-shop, then $\mylogic(\mathcal{D})$ is NP-complete.
\item[\textup{III}.] If $\mathcal{D}$ is preserved by an E-shop but
  is not preserved by any A-shop, then $\mylogic(\mathcal{D})$ is co-NP-complete.
\item[\textup{IV}.] If $\mathcal{D}$ is preserved neither by an
  A-shop nor by an E-shop, then $\mylogic(\mathcal{D})$ is Pspace-complete.
\end{longromannum}
\end{theorem}
\begin{proof}
The upper bounds (membership in \Logspace, \NP{} and \coNP) for Cases I, II and
III were known from~\cite{DBLP:journals/tocl/MadelaineM12}, but we
reprove them here as a corollary of Theorem~\ref{relativization} to keep this paper
self-contained. 

Note that an A-shop is simply a $U$-$X$-shop with $U=\{u\}$, for
some $u$ in $D$, and $X\subseteq D$. 
We may therefore replace every universal quantifier by the constant
$u$ and relativise every existential quantifier to $X$  by
Theorem~\ref{relativization}. This means that $\mylogic(\mathcal{D})$
is in \NP\ when it has an A-shop as a surjective hyper-endomorphism.
  
Note that an E-shop is simply a $U$-$X$-shop with
$X=\{x\}$ for some $x$ in $D$, and $U\subseteq D$. So Case III is dual
to Case II and we finally turn to Case I. 
 
With both an A-shop and an E-shop, we have a
$U$-$X$-shop with $U=\{u\}$ and $X =\{x\}$ where $u$ and $x$ are in
$D$.  We may therefore replace
every universal quantifier by the constant $u$ and every existential
quantifier by the constant $x$, by
Theorem~\ref{relativization}. We have reduced
$\mylogic(\mathcal{D})$ to the Boolean sentence value problem, known
to be in \Logspace~\cite{DBLP:journals/jacm/Lynch77}.

Theorem~\ref{LowerBoundtetrachotomy} deals with the lower bounds. NP-hardness for Case II and co-NP-hardness for Case III are proved in Subsection~\ref{sec:NP-and-coNP-hardness}. Pspace-hardness for Case III is proved in Subsection~\ref{sec:Pspace-hardness}.
\end{proof}

\subsection{Methodology : the Galois Connection $\Inv-\shE$}
\label{sec:GaloisConnectionInvShe}
The results of this subsection (\S~\ref{sec:GaloisConnectionInvShe}) appeared
in~\cite{DBLP:journals/tocl/MadelaineM12} and are proved here
to keep the present paper self-contained. 

Let $\shE(\mathcal{B})$ be the set of surjective hyper-endomorphisms of $\mathcal{B}$.
Let $\langle \mathcal{B} \rangle_{\mylogic}$ be the sets of relations
that may be defined on $\mathcal{B}$ in \mylogic. 

\begin{lemma}
\label{lemma:mylogic:connection-by-types}
Let $\tuple{r}:=(r_1,\ldots,r_k)$ be a $k$-tuple of elements of $\mathcal{B}$. There exists
a formula $\theta_{\tuple{r}}^{\mylogic}(u_1,\ldots,u_k) \in \mylogic$ such that
the following are equivalent.
\begin{romannum}
\item 
  $(\mathcal{B}, r'_1,\ldots, r'_k) \models
  \theta_{\tuple{r}}^{\mylogic}(u_1,\ldots,u_k)$.
\item 
  There is a surjective hyper-endomorphism from $(\mathcal{B}, r_1,\ldots, r_k)$ to $(\mathcal{B},
  r'_1,\ldots, r'_k)$.
\end{romannum}
\end{lemma}
\begin{proof}
  Let $\tuple{r} \in B^k$, $\tuple{s} := (b_1,\ldots,b_{|B|})$ be an
  enumeration of $B$ and $\tuple{t} \in B^{|B|}$. Recall that
  $\phi_{\mathcal{B}(\tuple{r},\tuple{s})}(u_1,\ldots,u_k,v_1,\ldots,v_{|B|})$
is a conjunction of the positive facts of $(\tuple{r},\tuple{s})$,
where the variables $(\tuple{u},\tuple{v})$ correspond to the elements
$(\tuple{r},\tuple{s})$. \\Similarly,
$\phi_{\mathcal{B}(\tuple{r},\tuple{s},\tuple{t})}(u_1,\ldots,u_k,v_1,\ldots,v_{|B|},w_1,\ldots,w_{|B|})$
is the conjunction of the positive facts of
$(\tuple{r},\tuple{s},\tuple{t})$, where the variables
$(\tuple{u},\tuple{v},\tuple{w})$ correspond to the elements
$(\tuple{r},\tuple{s},\tuple{t})$. 
Set $\theta_\tuple{r}^{\mylogic}(u_1,\ldots,u_k):=$
\begin{multline*}
  \exists v_1,\ldots,v_{|B|} \ \phi_{\mathcal{B}(\tuple{r},\tuple{s})}(u_1,\ldots,u_k,v_1,\ldots,v_{|B|}) \wedge \forall w_1 \ldots w_{|B|} \\
  \bigvee_{\tuple{t} \in B^{|B|}} \phi_{\mathcal{B}(\tuple{r},\tuple{s},\tuple{t})}(u_1,\ldots,u_k,v_1,\ldots,v_{|B|},w_1,\ldots,w_{|B|}).
\end{multline*}

[Backwards.] Suppose $f$ is a surjective hyper-endomorphism from
$(\mathcal{B}, r_1,\ldots, r_k)$ to $(\mathcal{B}', r'_1,\ldots,
r'_k)$, where $\mathcal{B}':=\mathcal{B}$ (we will wish to
differentiate the two occurrences of $\mathcal{B}$). We aim to prove
that $\mathcal{B}'  \models \theta_\tuple{r}^{\mylogic}(r'_1,\ldots,
r'_k)$. Choose arbitrary $s'_1 \in f(b_1),\ldots,s'_{|B|} \in
f(b_{|B|})$ as witnesses for $v_1,\ldots,v_{|B|}$. Let
$\tuple{t}':=(t'_1,\ldots,t'_{|B|}) \in B'^{|B|}$ be any valuation of
$w_1,\ldots,w_{|B|}$ and take arbitrary $t_1,\ldots,t_{|B|}$
s.t. $t'_1 \in f(t_1)$, \ldots, $t'_{|B|} \in f(t_{|B|})$ (here we use
surjectivity). Let $\tuple{t}:=(t_1,\ldots,t_{|B|})$. It follows from
the definition of a surjective hyper-endomorphism that
\[ \mathcal{B}' \ \models \ \phi_{\mathcal{B}(\tuple{r},\tuple{s})}(r'_1,\ldots,r'_k,s'_1,\ldots,s'_{|B|}) \wedge \phi_{\mathcal{B}(\tuple{r},\tuple{s},\tuple{t})}(r'_1,\ldots,r'_k,s'_1,\ldots,s'_{|B|},t'_1,\ldots,t'_{|B|}). \]

[Forwards.] Assume that $\mathcal{B}' \models \theta_\tuple{r}^{\mylogic}(r'_1,\ldots, r'_k)$, where $\mathcal{B}':=\mathcal{B}$. Let $b'_1, \ldots,b'_{|B|}$ be an enumeration of $B':=B$.\footnote{One may imagine $b_1, \ldots,b_{|B|}$ and $b'_1, \ldots,b'_{|B|}$ to be the same enumeration, but this is not essential. In any case, we will wish to keep the dashes on the latter set to remind us they are in $\mathcal{B}'$ and not $\mathcal{B}$.} Choose some witness elements $s'_1,\ldots,s'_{|B|}$ for $v_1,\ldots,v_{|B|}$ and a witness tuple $\tuple{t}:=$ $(t_1,\ldots,t_{|B|}) \in B^{|B|}$ s.t.
\[ (\dagger) \ \mathcal{B}' \models \ \phi_{\mathcal{B}(\tuple{r},\tuple{s})}(r'_1,\ldots,r'_k,s'_1,\ldots,s'_{|B|}) \wedge \phi_{\mathcal{B}(\tuple{r},\tuple{s},\tuple{t})}(r'_1,\ldots,r'_k,s'_1,\ldots,s'_{|B|},b'_1,\ldots,b'_{|B|}). \]
%Recall that $b_1, \ldots,b_{|B|}$ is the enumeration of $B$ that ties in with the $v_1,\ldots,v_{|B|}$ in both $\phi_{\mathcal{B}(\tuple{r},\tuple{s})}$ and $\phi_{\mathcal{B}(\tuple{s},\tuple{t})}$.
Consider the following partial hyper-functions from $B$ to $B'$.
\begin{itemize}
\item[1.] $f_{\tuple{r}}$ given by $f_{\tuple{r}}(r_i):= \{ r'_i \}$, for $1 \leq i \leq k$.
\item[2.] $f_{\tuple{s}}$ given by $f_{\tuple{s}}(b_i)=\{ s'_i \}$, for $1 \leq i \leq |B|$. \hfill (totality)
\item[3.] $f_{\tuple{t}}$ given by $b'_i \in f_{\tuple{t}}(b_j)$ iff $t_i=b_j$, for $1 \leq i,j \leq |B|$. \hfill (surjectivity) %\footnote{The inverse of $f_{\tuple{t}}$ is a function from $B'$ to $B$.}
\end{itemize}
Let $f:=f_{\tuple{r}} \cup f_{\tuple{s}} \cup f_{\tuple{t}}$; $f$ is a
hyper-operation whose surjectivity is guaranteed by $f_{\tuple{t}}$
(note that totality is guaranteed by $f_{\tuple{s}}$). That $f$ is a
surjective hyper-endomorphism follows from the right-hand conjunct of
$(\dagger)$. 
\end{proof}
\begin{theorem}%[\textbf{Galois Connection $\Inv-\shE$}]
\label{thm:galois-connectionMyLogic}
For a finite structure $\mathcal{B}$ we have $\langle \mathcal{B} \rangle_{\mylogic} = \Inv(\shE(\mathcal{B}))$.
\end{theorem}
\begin{proof}
  \begin{longenum}
  \item 
    $\varphi(\tuple{v}) \in \langle \mathcal{B} \rangle_{\mylogic} \
    \Rightarrow \ \varphi(\tuple{v}) \in \Inv(\shE(\mathcal{B}))$. 
    This is proved by induction on the complexity of $\varphi(\tuple{v})$.
    We only have to deal with the case of universal quantification in
    the inductive step, the other cases having been dealt with in the
    proof of the $\Inv-\hE$ Galois Connection.

    (\textbf{Inductive Step} continued from proof of Theorem~\ref{theorem:galois-connectionEqFreeWeakKrasner}.) 
    \begin{enumerate}
    \item[(d)] $\varphi(\tuple{v}):=\forall w \
      \psi(\tuple{v},w)$. Let $\tuple{v}:=(v_1,\ldots,v_l)$. Suppose
      $\mathcal{B} \models \forall w \ \psi(x_1,\ldots,x_l,w)$; then
      for each $x'$, $\mathcal{B} \models \psi(x_1,\ldots,x_l,x')$. By
      IH, for any $y_1\in f(x_1),\ldots,y_l\in f(x_l)$, we have for
      all $y'$ (remember $f$ is surjective), $\mathcal{B} \models
      \psi(y_1,\ldots,y_l,y')$, whereupon $\mathcal{B} \models \forall
      w \ \psi(y_1,\ldots,y_l,w)$. 
    \end{enumerate}

  \item $S \in \Inv(\shE(\mathcal{B})) \ \Rightarrow \ S \in \langle
    \mathcal{B} \rangle_{\mylogic}.$ 
    Consider the $k$-ary relation $S \in \Inv(\shE(\mathcal{B}))$. Let
    $\tuple{r}_1,\ldots,\tuple{r}_m$ be the tuples of $S$. Let
    $\theta_S^{\mylogic}(u_1,\ldots,u_k)$ be the following formula of $\mylogic$: 
    \begin{multline*}
      \theta_{\tuple{r}_1}^{\mylogic}(u_1,\ldots,u_k) \vee \ldots \vee
      \theta_{\tuple{r}_m}^{\mylogic}(u_1,\ldots,u_k).
    \end{multline*}
    For $\tuple{r}_i:=(r_{i1},\ldots,r_{ik})$, note that
    $(\mathcal{B},r_{i1},\ldots,r_{ik}) \models
    \theta_{\tuple{r}_i}^{\mylogic}(u_1,\ldots,u_k)$ 
    (viewing the identity endomorphism as a surjective hyper endomorphism). 
    That
    $\theta_S^{\mylogic}(u_1,\ldots,u_k)=S$ now follows from Part (ii) of
    Lemma~\ref{lemma:mylogic:connection-by-types}, since $S \in
    \Inv(\shE(\mathcal{B}))$. \qed
  \end{longenum}
\end{proof}
% Let $\leq_{\Logspace}$ indicate the existence of a logspace many-to-one reduction. 
%The following lemma is our counterpart to Corollary 4.11 of~\cite{jeavons98algebraic} (for CSP) and Theorem
% 3.1 of \cite{OxfordQuantifiedConstraints} (for QCSP).
\begin{corollary}
  \label{cor:she-reduction}
  \label{thm:she-reduction}
  Let $\mathcal{B}$ and $\mathcal{B}'$ be finite structures over the same domain $B$.
  If $\shE(\mathcal{B}) \subseteq \shE(\mathcal{B}')$ then $\mylogic(\mathcal{B}') \leq_{\Logspace} \mylogic(\mathcal{B})$.
\end{corollary}
\begin{proof}
  If $\shE(\mathcal{B}) \subseteq \shE(\mathcal{B}')$, then
  $\Inv(\shE(\mathcal{B}')) \subseteq \Inv(\shE(\mathcal{B}))$. From
  Theorem~\ref{thm:galois-connectionMyLogic}, it follows that
  $\langle \mathcal{B}' \rangle_{\mylogic}$ $\subseteq \langle
  \mathcal{B} \rangle_{\mylogic}$. Recalling that $\mathcal{B}'$
  contains only a finite number of extensional relations, we may
  therefore effect a Logspace reduction from $\mylogic(\mathcal{B}')$
  to $\mylogic(\mathcal{B})$ by straightforward substitution of
  predicates. 
\end{proof}

Consequently, the complexity of $\mylogic(\mathcal{B})$ is
characterised by $\shE(\mathcal{B})$. 

%%%%%%%%%%%%%%%%%%%%%%%%%%%%%%%%%%%%%%%%%%%%%%%%%%%%%%%%%%%%%%%%%%%%%%%%%%%

\subsection{The Boolean case}
\label{sec:boolean-case}
We recall the case $|B|=2$ (from \cite{DBLP:journals/tocl/MadelaineM12}), with the normalised domain $B:=\{0,1\}$
as a warm-up.
%For typographic reasons we will mark-up, e.g., the surjective hyper-operation $0 \mapsto \{0,1\}$ and $1 \mapsto \{1\}$ as $\shee{01}{1}$.
It may easily be verified that there are five DSMs in this case,
depicted as a lattice in Figure~\ref{fig:boolean-she-lattice}. The two
elements of this lattice that represent the two subgroups of $S_2$ are
drawn in the middle and bottom. We write $\shee{01}{1}$ for the shop
that sends $0$ to $\{0,1\}$ and $1$ to $\{1\}$.

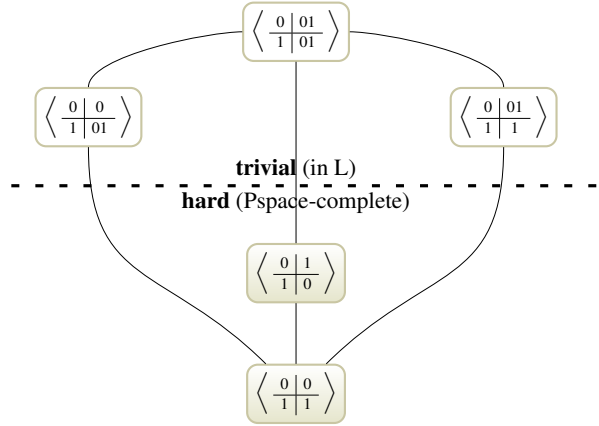
\begin{figure}[h]
\label{fig:boolean-she-lattice}
  \centering
    \begin{tikzpicture}[node distance = 2cm, auto, scale=0.8,every
      node/.style={transform shape}]
      \tikzstyle{logspace} = [rectangle, rounded corners,
      text centered, thick, rounded corners,
      draw=yellow!50!black!50] 
      \tikzstyle{pspacehard} = [rectangle, rounded corners,
      text centered, thick, rounded corners,
      draw=yellow!50!black!50, top color=white, bottom
      color=yellow!50!black!20] 
      \tikzstyle{line} = [draw, very thin] 
      \node[logspace] (top) {$\left\langle \shee{01}{01}
        \right\rangle$};
      \node[pspacehard, below of =top, node distance=4cm] (swap) {$\left\langle
          \shee{1}{0} \right\rangle$};
      \node[logspace, below left of =top,xshift=-2cm] (A1E0) {$\left\langle
          \shee{0}{01} \right\rangle$};
      \node[logspace, below right of =top,xshift=2cm] (A0E1) {$\left\langle
          \shee{01}{1} \right\rangle$}; 
      \node[pspacehard, below of =swap] (bot) {$\left\langle
          \shee{0}{1} \right\rangle$}; 
      %\path [line] (SillyCase) edge [out=45,in=180,looseness=1.3] (CspDisj);
      \path [line] (bot) -- (swap);
      \path [line] (bot) edge [out=135,in=-90,looseness=1.3](A1E0);
      \path [line] (bot) edge [out=45,in=-90,looseness=1.3] (A0E1);
      \path [line] (A1E0) edge [out=90,in=180,looseness=.5] (top);
      \path [line] (A0E1) edge [out=90,in=0,looseness=.5](top);
      \path [line] (swap) -- (top);
      \node[inner sep=18pt,thin,rectangle,
      fit=(top) (A1E0)(A0E1)] 
      (logspaceCases){};
     \node at (logspaceCases.south) [above, inner sep=1mm]
     {\textcolor{black}{\textbf{trivial} (in \Logspace)}};
     \path[line, very thick, loosely dashed] (logspaceCases.south
     east) -- (logspaceCases.south west);
     \node[inner sep=18pt,thin,rectangle,
     fit=(bot) (swap)] 
     (PspaceCases){};
     % \node at (PspaceCases.south) [above, inner sep=1mm]
     % {\textcolor{black}{\textbf{hard} (\Pspace-complete)}};
     \node at (logspaceCases.south) [below, inner sep=1mm]
     {\textcolor{black}{\textbf{hard} (\Pspace-complete)}};
    \end{tikzpicture}
% \[
% \xymatrix{
% & \stackrel{\mbox{\color{black}{}}}{\Logspace} & \\
% \stackrel{\mbox{\color{black}{$\left\langle \shee{0}{01} \right\rangle$}}}{\Logspace} \ar[ur] & \stackrel{\left\langle \shee{1}{0} \right\rangle}{\Pspace -c} \ar[u] & \stackrel{\mbox{\color{black}{$\left\langle \shee{01}{1} \right\rangle$}}}{\Logspace} \ar[ul] \\
% & \stackrel{\left\langle \shee{0}{1} \right\rangle}{\Pspace -c} \ar[ul] \ar[u] \ar[ur] & \\
% }
% \]
\caption{The boolean lattice of DSMs with their associated complexity.}
\end{figure}

\begin{theorem}[\cite{DBLP:journals/tocl/MadelaineM12}]\label{theorem:dichotomy:mylogic:boolean}
Let $\mathcal{B}$ be a boolean structure.
\begin{longromannum}
\item[\textup{I}.] If either $\shee{01}{1}$ or $\shee{0}{01}$ is a surjective hyper-endomorphism of
  $\mathcal{B}$, then $\mylogic(\mathcal{B})$ is in \Logspace. 
\item[\textup{II}.] Otherwise, $\mylogic(\mathcal{B})$ is \Pspace-complete.
\end{longromannum}
\end{theorem}
\begin{proof}
  $\shE(\mathcal{B})$ must be one of the five DSMs depicted in
  Figure~\ref{fig:boolean-she-lattice}. If $\shE(\mathcal{B})$ contains
  $\shee{01}{1}$ then we may relativise every existential quantifier to
  $1$ and every universal quantifier to $0$ by
  Theorem~\ref{lemma:UXrelativisation:To:UXshop} and evaluate in
  \Logspace\ the quantifier-free part. The case of $\shee{0}{01}$ is
  similar with the role of $0$ and $1$ swapped.
  
  We prove that if $\shE(\mathcal{B})=\langle \shee{1}{0} \rangle$ then
  $\mylogic(\mathcal{B})$ is \Pspace-complete.
  The structure $\mathcal{K}_2$ has DSM $\shE(\mathcal{B})=\langle \shee{1}{0} \rangle$.
  It suffices therefore to prove that $\mylogic(\mathcal{K}_2)$ is
  \Pspace-hard, which we did by reduction from $\QCSP(\mathcal{B}_{\mbox{\textsc{nae}}})$ in the
  proof of Proposition~\ref{proposition:posFoNeq:complexity}.
  
  If follows from Corollary~\ref{cor:she-reduction} that when
  $\shE(\mathcal{B})=\langle \shee{0}{1} \rangle$,
  $\mylogic(\mathcal{B})$ is also \Pspace-hard since $\langle
  \shee{0}{1} \rangle \subseteq \langle
  \shee{1}{0} \rangle$.
\end{proof}

%%%%%%%%%%%%%%%%%%%%%%%%%%%%%%%%%%%%%%%%%%%%%%%%%%%%%%%%%%%%%%%%%
%\input{hardness}
\subsection{Proving Hardness}
\label{sec:hardness}
Our aim is to derive the following lower bounds.
\begin{theorem}%[\textbf{lower bounds}]
\label{LowerBoundtetrachotomy}
%Let $\mathcal{D}$ be a structure. 
\begin{itemize}
\item[\textup{II}.] If $\mathcal{D}$ is preserved by an A-shop but
  is not preserved by any E-shop, then $\mylogic(\mathcal{D})$ is NP-hard.
\item[\textup{III}.] If $\mathcal{D}$ is preserved by an E-shop but
  is not preserved by any A-shop, then $\mylogic(\mathcal{D})$ is co-NP-hard.
\item[\textup{IV}.] If $\mathcal{D}$ is preserved neither by an
  A-shop nor by an E-shop, then $\mylogic(\mathcal{D})$ is Pspace-hard.
\end{itemize}
\end{theorem}
It follows from Proposition~\ref{prop:ShopWithIdentity} and
Corollary~\ref{cor:UXretract} that the complexity of a 
structure $\mathcal{D}$ is the same as the complexity of its
$U$-$X$-core. Hence in this Section, \emph{we assume w.l.o.g. that
  $U\cup X = D$}. We will say in this case that the DSM $\mathcal{M}$
is \emph{reduced}. This is the critical ingredient, hitherto missing, that is needed to obtain the full classification. In order to prove
Theorem~\ref{LowerBoundtetrachotomy}, we need to establish the following:
\begin{itemize}
\item[II.] If $U$ is of size one and $X$ of size at least two then $\mylogic(\mathcal{D})$ is NP-hard;
\item[III.] If $X$ is of size one and $U$ of size at least two then $\mylogic(\mathcal{D})$ is co-NP-hard; and,
\item[IV.] If both $U$ and $X$ have at least two elements then $\mylogic(\mathcal{D})$ is Pspace-hard. 
\end{itemize}
In the following. we will describe a DSM $\mathcal{M}$ as being (NP-, co-NP-, Pspace-)hard in the case that $\mylogic(\mathcal{D})$ is hard for some $\mathcal{D} \in \Inv(\mathcal{M})$. In order to facilitate the hardness proof, we would like to show
hardness of a monoid $\widehat{\mathcal{M}}$ with a very simple structure of
which $\mathcal{M}$ is in fact a sub-DSM ($\widehat{\mathcal{M}}$ is the \emph{completion} of $\mathcal{M}$). As in general $\widehat{\mathcal{M}}$
preserves fewer relations than $\mathcal{M}$, the hardness of
$\mathcal{M}$ would follow. We would like the structure of
$\widehat{\mathcal{M}}$ to be sufficiently simple for us to build canonically
some gadgets for our hardness proof.
Thus, we wish to better understand the form that elements of $\mathcal{M}$
may take. In order to do so, we first define the \emph{canonical shop} of $\mathcal{M}$ to be the $U$-$X$ shop $h$ in $\mathcal{M}$, guaranteed by Proposition~\ref{prop:ShopWithIdentity}, with the property that $|h(z)|$ is maximal for each $z \in U \setminus X$. Note that this maximal $h$ is unique, as given $h_1$ and $h_2$ of the form in Proposition~\ref{prop:ShopWithIdentity}, $h_1 \circ h_2$ is also of the required form, and further satisfies $|h_1 \circ h_2(z)|\geq |h_1(z)|, |h_2(z)|$, for all $z \in U \setminus X$.

\subsubsection{Characterising reduced DSMs}
\label{sec:char-reduc-dsms}
%%%%%%%%%%%%%%%%%%%%%%%%%%%%%%%%%%%%%%%%%%%%%%%%%%%%%%%%%%%%%%%%%%%%%%%%
Any $U$-$X$-shop in $\mathcal{M}$ will be shown to be in the following
special form, reminiscent of the form of the canonical shop.
\begin{definition}
  We say that a shop $f$ is in the \emph{3-permuted form} if there are
  a permutation $\zeta$ of $X\cap U$, a permutation $\chi$ of
  $X\setminus U$ and a permutation $\upsilon$ of $U\setminus X$ such
  that $f$ satisfies:
  \begin{itemize}
  \item for any $y$ in $U\cap X$, $f(y)=\{\zeta(y)\}$;
  \item for any $x$ in $X\setminus U$, $f(x)=\{\chi(x)\}$; and,
  \item for any $u$ in $U\setminus X$, $f(u)=\{\upsilon(u)\}\cup X_u$,
    where $X_u\subseteq X\setminus U$.
  \end{itemize}
\end{definition}
\begin{lemma}  \label{lem:3permutedForm}
  If a shop $f$ satisfies $f(X)\cap (U\setminus X)=\emptyset$ then $f$
  is in the \emph{3-permuted form}.
\end{lemma}
\begin{proof}
  The hypothesis forces an element of $X$ to reach an element of $X$ and
  Lemma~\ref{lem:AtMostOneXInPreImage} forces two elements of $X$ to
  have different images. Since $X$ is finite, there exists a
  permutation $\beta$ of $X$ such that for every $x$ in $X$,
  $f(x)=\{\beta(x)\}$. 
  Since Lemma~\ref{lem:AtMostOneUInImage} forces in particular an
  element of $U$ to have at most one element of $U$ in its image and
  since $U$ is finite, it follows that there exists a permutation
  $\alpha$ of $U$ such that for every $u$ in $U$, $f(u)\cap U =
  \{\alpha(U)\}$ and $f^{-1}(u)\cap U = \{\alpha^{-1}(U)\}$.

  It follows that there exists a permutation $\zeta$ of $U\cap X$ such
  that for any $y$ in $U\cap X$, $f(y)=\{\zeta(y)\}$. 

  The existence of a permutation $\chi$ of $X\setminus U$ such
  that $\beta$ is the disjoint union of $\chi$ and $\zeta$
  follows. Hence, for any $x$ in $X\setminus U$, $f(x)=\{\chi(x)\}$.

  Similarly, there must also be a permutation $\upsilon$ of $U\setminus
  X$ such that $\alpha$ is the disjoint union of $\upsilon$ and $\zeta$.
  Hence, for any $u$ in $U\setminus X$, $f(u)\cap U =\{\upsilon(u)\}$.
  Elements of $U\setminus X$ may however have some images in
  $X\setminus U$. So we get finally that for any $u$ in $U\setminus
  X$, there is some $\emptyset\subseteq X_u \subseteq X \setminus U$ such that
  $f(u)=\{\upsilon(u)\}\cup X_u$.
  This proves that $f$ is in the 3-permuted form and we are done.
\end{proof}
\begin{theorem}%[\textbf{characterisation of reduced DSM}]
  \label{theo:3PermutedForm}
  Let $\mathcal{M}$ be a reduced DSM.
  Every shop in $\mathcal{M}$ is in the 3-permuted form.
  Moreover, every $U$-$X$-shop in $\mathcal{M}$ follows 
  the additional requirement that the elements of $U\setminus X$ cover
  the set $X\setminus U$, more formally that
  \begin{equation*}
    f(U\setminus X)\cap X=\bigcup_{u \in U\setminus X} X_u = X\setminus U.
  \end{equation*}
\end{theorem}
\begin{proof}
  We can now deduce easily from Lemmata~\ref{fact:UXShopMapsNoProperXToProperU}
  and~\ref{lem:3permutedForm} that $U$-$X$-shops in
  $\mathcal{M}$ must take the $3$-permuted form. It remains to prove that an arbitrary shop
  $f$ in $\mathcal{M}$ is in the 3-permuted form.
  Let $h$ be the canonical shop of $\mathcal{M}$.
  It follows from Lemma~\ref{lem:compositionAndUXshops} that
  $f':= h\circ f \circ h$ is a $U$-$X$-shop. Hence,
  $f'$ is in the $3$-permuted form.
  Let $z$ in $X$ and $u$ in $U\setminus X$.
  If $f(z)\ni u$ then $f'(z)\ni u$ and $f'$ would not be in the
  3-permuted form. It follows that $f(X)\cap (U\setminus
  X)=\emptyset$ and appealing to Lemma~\ref{lem:3permutedForm} that $f$ is in
  the 3-permuted form. 
\end{proof}

We do not need the following result in order to prove our main
result. But surprisingly in a reduced DSM, $U$ and $X$ are unique.
This means that we may speak of \emph{the canonical shop of
  $\mathcal{M}$} instead of some canonical $U$-$X$-shop. It also means
that we can define the $U$-$X$-core of a structure $\mathcal{D}$
\emph{without explicitly referring to $U$ or $X$} as the minimal substructure
of $\mathcal{D}$ which satisfy the same \mylogic\ sentences.
\begin{theorem}
  \label{theo:uniquenessUXWhenFullDomain}
  Let $\mathcal{D}$ be a structure that is both a $U$-$X$-core and a
  $U'$-$X'$-core then it follows that $U=U'$ and $X=X'$.
\end{theorem}
\begin{proof}
  We do a proof by contradiction.
  Let $h$ and $h'$ be the canonical $U$-$X$-shop and $U'$-$X'$-shop,
  respectively.  
  Assume $U'\neq U$ and let $x$ in $U'\setminus U$.
  Note that since $D = U\cup X$, our notation is consistent as $x$
  does belong to $X\setminus U$. Thus, there exists some $u$ in
  $U\setminus X$ such that $h(u)\supseteq \{u,x\}$ (and necessarily
  $u\neq x$).

  By Theorem~\ref{theo:3PermutedForm}, $h$ has to be in the $3$-permuted
  form w.r.t. $U'$ and $X'$, which means that $h$ can send an element
  to at most one element of $U'$. Since $x$ belongs to $U'$, it
  follows that $u$ belongs to $D\setminus U'=X'\setminus U'$. But the
  three permuted form prohibits an element of $X'$ to reach an element
  of $U'$. A contradiction.
 
  It does not follow yet that $X'=X$ as the pairs of sets  may have
  shifting intersections. However, the dual argument to the above
  applies and yields $X=X'$.
\end{proof}

\begin{corollary}%[\textbf{The $U$-$X$-core with implicit $U$ and $X$}]
\label{cor:implicit}
  Let $\mathcal{D}$ be a finite structure.
  The $U$-$X$-core of $\mathcal{D}$ is unique up to isomorphism.
  It is a minimal induced substructure $\widetilde{\mathcal{D}}$ of
  $\mathcal{D}$,  that satisfies the same
  $\{\exists,\forall,\land,\lor\}$-$\FO$
  formulae with free-variables in $\widetilde{D}$. 
  Moreover, once $\widetilde{D}$ is fixed, there are two uniquely determined subsets $U$
  and $X$ such that $U\cup X = \widetilde{\mathcal{D}}\subset D$
  which are minimal within $D$ with respect to the following
  equivalent properties,
  \begin{itemize}
  \item  $\mathcal{D}$ has \emph{$\forall U$-$\exists X$-relativisation}
    w.r.t. $\{\exists,\forall,\land,\lor\}$-$\FO$; or,  %(\textbf{logic}) 
  \item $\mathcal{D}$ has a $U$-$X$-shop that may act as the identity
    over $U\cup X$. % (\textbf{algebraic})  
\end{itemize}
\end{corollary}
\begin{proof}
  The last point follows from our definition of a $U$-$X$-core and 
  from Proposition~\ref{prop:ShopWithIdentity}. It is equivalent to
  the \emph{$\forall U$-$\exists X$-relativisation} property by
  Theorem~\ref{relativization}. 
  It follows that $\mathcal{D}$ and $\widetilde{\mathcal{D}}$  satisfy the same
  $\{\exists,\forall,\land,\lor\}$-$\FO$
  formulae with free-variables in $\widetilde{D}$ (see
  Corollary~\ref{cor:UXretract}). 
  Conversely, if $\mathcal{D}$ and $\widetilde{\mathcal{D}}$  satisfy the same
  $\{\exists,\forall,\land,\lor\}$-$\FO$ formulae with free-variables
  in $\widetilde{D}$, then $\mathcal{D}$ has
  $\widetilde{D}$-$\widetilde{D}$-relativisation. The existence of a
  ``$\widetilde{D}$-$\widetilde{D}$-shop'' follows by
  Theorem~\ref{relativization}.
  Enforcing the minimality criteria, we get some $U$-$X$-shop with
  some $U,X\subseteq \widetilde{D}$ 
  (this is because, we may proceed by retraction, as explained in
  the beginning of Subsection~\ref{sec:u-x-core}). Moreover, by
  minimality of $\widetilde{\mathcal{D}}$, we must have $U\cup X =
  \widetilde{D}$. We have a $U$-$X$-core as in our original definition
  in terms of a $U$-$X$-shop satisfying minimality criteria.
  It follows from Theorem~\ref{theo:uniquenessUXWhenFullDomain} that
  $U$ and $X$ are unique (within $\widetilde{D}$).
\end{proof}
Recall that the \mylogic-core $\mathcal{D}'$ of $\mathcal{D}$ is
  the smallest (w.r.t. domain size) structure that is \mylogic-equivalent to
  $\mathcal{D}$.
\begin{proposition}
  The notion of a $U$-$X$-core and of a \mylogic-core coincide.
\end{proposition}
\begin{proof}
  Let $\mathcal{D}$ be a structure that is a $U$-$X$-core with
  (unique) subsets $U$ and $X$. Let $c$ be the canonical shop of
  $\mathcal{D}$. 

  Let $\mathcal{D}'$ be a \mylogic-core of $\mathcal{D}$, that is
  a smallest (w.r.t. domain size) structure that is \mylogic-equivalent to
  $\mathcal{D}$.
  Let $U'$ and $X'$ be subsets of $D'$ witnessing that
  $\mathcal{D}'$ is a $U'$-$X'$ core. Note that $U'\cup X'=D'$ by
  minimality of $\mathcal{D}'$ (and consequently, $U'$ and $X'$ are uniquely determined by Theorem~\ref{theo:uniquenessUXWhenFullDomain}). Let $c'$ be the canonical shop of
  $\mathcal{D}'$.

  By Proposition~\ref{prop:EqFreeFoContainment}, since $\mathcal{D}$
  and $\mathcal{D}'$ are \mylogic-equivalent, there exist two
  surjective hyper-morphisms $g$ from $\mathcal{D}$ to $\mathcal{D}'$
  and $f$ from $\mathcal{D}'$ to $\mathcal{D}$.
  
  Let $U''$ be a minimal subset of $(g)^{-1}(U')$  such that $g(U'')=U'$.
  Note that   $f\circ c'\circ g$ is a $U''$-surjective shop of $\mathcal{D}$. By minimality of $U$, it follows that $|U|\leq|U''|\leq|U'|$. A similar argument over $\mathcal{D'}$ gives $|U'|\leq|U|$, and consequently, $|U|=|U'|$.
  Moreover, since $c\circ (f\circ c'\circ g)$ is a $U''$-surjective $X$-total surjective hyperendomorphism of $\mathcal{D}$, By Theorem~\ref{theo:uniquenessUXWhenFullDomain}, it follows that $U=U''$.

This means that there is a bijection $\alpha'$ from $U'$ to $U$ such that, for any $u'$ in $U'$,
$g^{-1}(u')=\{\alpha'(u')\}$.

By duality we obtain similarly that $|X|$=$|X'|$ and that there is a bijection $\beta$ from $X$ to $X'$ such that, for any $x$ in $X$, $g(x)=\{\beta(x)\}$.

Thus, $g$ acts necessarily as a bijection from $U\cap X$ to $U'\cap X'$.

The map $\tilde{g}$ from $D$ to $D'$ defined for any $u$ in $U$ as $\tilde{g}(u):=\alpha'^{-1}(u)$
and $\tilde{g} (x):=\beta(x)$ is a homomorphism from $\mathcal{D}$ to $\mathcal{D'}$ that is both injective and surjective.

A symmetric argument yields a map $\tilde{f}$ that is a bijective homomorphism from $\mathcal{D}'$ to $\mathcal{D}$. Isomorphism of $\mathcal{D}'$ and $\mathcal{D}$ follows.
 
%A symmetric argument with $f$ yields that there is a bijection $\alpha$ from $U$ to $U'$ such that, for any $u$ in $U$,
%$f^{-1}(u)=\{\alpha(u)\}$ and a bijection $\beta'$ from $X'$ to $X$ such that, for any $x'$ in $X'$, $f(x')=\{\beta'(x')\}$.
%note that, $f$ acts necessarily as a bijection from $U'\cap X'$ to $U\cap X$.

%Since we deal with finite structures, a suitable iterate of $f\circ g$ %say $(f\circ g)^{r}$ will act as the identity over $X$ and the inverse of this iterate will act as the identity over $U$. Thus, starting over with $(f\circ g)^{r-1}\circ g$ in the role of $g$ and $f$ in the role of $f$, we may assume w.l.o.g. that $\alpha'$ is the inverse of $\alpha$ and that $\beta$ is the inverse of $\beta'$.
%Hence, the inverse $\tilde{f}$ of $\tilde{g}$ is a range restriction of $f$ and is therefore a homomorphism (we identify as before hyper-morphism to a singleton with homomorphisms). It follows that $\mathcal{D}$ and $\mathcal{D}'$ are isomorphic.
\end{proof}

\begin{remark}
  To simplify the presentation, we defined the $\mathcal{L}$-core 
  as a minimal structure w.r.t. domain size. Considering minimal structures
  w.r.t. inclusion, we would get the same notion for $\mylogic$.
  This is also the case for \CSP, but it is not the case in general. 
  For example, this is not the case for the logic \qcsplogic, which
  corresponds to \QCSP~\cite{DBLP:journals/corr/abs-1204-5981}.
\end{remark}

\begin{lemma}\label{lemma:mickeyMouse}
 Let $\mathcal{M}$ be a reduced DSM with associated sets
 $U$ and $X$.
 There are only three cases possible.
 \begin{enumerate}
 \item $U\cap X\neq \emptyset$, $U\setminus X\neq \emptyset$  and $U\setminus X\neq \emptyset$.
 \item $U=X$.
 \item $U\cap X=\emptyset$.
 \end{enumerate}
\end{lemma}
\begin{proof}
  We prove that $U\subsetneq X$ is not possible.
  Otherwise, let $x$ in $X\setminus U$ and $h$ be the canonical shop.
  There exists some $u$ in $U\subsetneq X$ such that $h(u)\ni x$ by
  $U$-surjectivity of $h$. Since $u$ does not occur in the image of
  any other element than $u$ under the canonical shop, this would mean
  that $h$ is $X\setminus \{u\}$-total, contradicting the minimality
  of $X$.
  
  By duality $X\subsetneq U$ is not possible either and the result follows.
\end{proof}

\subsubsection{The hard DSM above $\mathcal{M}$}
\label{sec:hard-dsm-above}
Define the completion $\widehat{\mathcal{M}}$ of $\mathcal{M}$ to be the DSM that contains \emph{all} shops in the
3-permuted form of $\mathcal{M}$. 
More precisely, the canonical shop of $\widehat{\mathcal{M}}$ is the shop $\hat{h}$
where every set $X_u$ is the whole set $X\setminus U$, and, for every
permutation $\zeta$ of $X\cap U$, $\chi$ of $X\setminus U$ and
$\upsilon$ of $U\setminus X$, any shop in the 
3-permuted form with these permutations is in $\widehat{\mathcal{M}}$. 
Note that by construction, $\mathcal{M}$ is a sub-DSM of
$\widehat{\mathcal{M}}$. 
Note also that the minimality of $U$ and $X$ still holds in
$\widehat{\mathcal{M}}$.  We will establish hardness for $\widehat{\mathcal{M}}$, whereupon hardness of $\mathcal{M}$ follows from Theorem~\ref{thm:galois-connectionMyLogic}.

\subsubsection{Cases II and III: NP-hardness and co-NP-hardness}
\label{sec:NP-and-coNP-hardness}

We begin with Case II. We note first that $U=\{u\}$ and $|X|\geq 2$ implies $U \cap
X=\emptyset$ by Lemma~\ref{lemma:mickeyMouse}. 
The structure $\mathcal{K}_{|X|}\uplus \mathcal{K}_1$, the disjoint
union of a clique of size $|X|$ with an isolated vertex $u$, has associated DSM
$\widehat{\mathcal{M}}$. 
The problem $\cspDisj(\mathcal{K}_{|X|}\uplus
  \mathcal{K}_1)$ is \NP-hard, since the core of $\mathcal{K}_{|X|}\uplus
\mathcal{K}_1$ is $\mathcal{K}_{|X|}$ by Proposition~\ref{prop:cspDisj:complexity}.

For Case III, we may assume similarly to above that $X=\{x\}$,
$|U|\geq 2$ and $U\cap X=\emptyset$ by Lemma~\ref{lemma:mickeyMouse}. We use the duality principle,
which corresponds to taking the inverse of shops. Since the inverse of an $\{x\}$-total $U$-surjective shop with
$U\geq 2$ is a $\{U\}$-total $\{x\}$-surjective shop, we may use the structure $\overline{\mathcal{K}_{|U|}\uplus \mathcal{K}_1}$ which is \CocspDisj-equivalent to $\overline{\mathcal{K_{|U|}}}$ (and $\CocspDisj(\overline{\mathcal{K_{|U|}}})$ is \coNP-hard).
%\begin{proposition}
%  Let $\mathcal{D}$ be a structure.
%  The set $\shE(\overline{\mathcal{D}})$ consists of exactly the
%  inverses of the shops in $\shE(\mathcal{D})$. 
%\end{proposition}
%\begin{proof}
%  Follows directly from the definitions.
%\end{proof}

\subsubsection{case IV: Pspace-hardness}
\label{sec:Pspace-hardness}
We assume that $|U|\geq 2$ and $|X|\geq 2$ and consider the tree
possible cases given by Lemma~\ref{lemma:mickeyMouse}.

\paragraph*{Case 1: when $U\cap X\neq \emptyset$, $U\setminus X\neq
  \emptyset$ and $X\setminus U\neq \emptyset$}
Recall that if $\mathcal{M}$ is a sub-DSM
of a hard DSM $\widehat{\mathcal{M}}$ then $\mathcal{M}$ is also hard (see
Theorem~\ref{thm:galois-connectionMyLogic}).

We write $U\Delta X$ as an abbreviation for $(X\setminus U) \cup (U\setminus X)$.
To build $\widehat{\mathcal{M}}$ from $\mathcal{M}$, we added all permutations,
and chose for each set $X_u=X\setminus U$. 
We carry on with this completion process and consider the super-DSM
$\mathcal{M}'$ which is 
generated by a single shop $g'$ defined as follows:
\begin{itemize}
\item for every $y$ in $X\cap U$, $g'(y):=X\Delta U$; and,
\item for every $z$ in $X\Delta U$, $g'(z):=X\cap U$, where
  $X\Delta U $ denotes $(X\setminus U) \cup (U\setminus X)$.
\end{itemize}

The complete bipartite graph $\mathcal{K}_{X\Delta U,X\cap U}$  has
$\mathcal{M}'$ for DSM. 
Observing that there is a full surjective homomorphism from
$\mathcal{K}_{X\Delta U,X\cap U}$ to $\mathcal{K}_2$, thus by
Proposition~\ref{prop:EqFreeFoContainment} the two structures agree on
all sentences of \EqFreeFo\ and so also on all sentences of \mylogic.
It suffices therefore to prove that $\mylogic(\mathcal{K}_2)$ is
\Pspace-hard, which we did by reduction from $\QCSP(\mathcal{B}_{\mbox{\textsc{nae}}})$ in the
proof of Theorem~\ref{theorem:dichotomy:mylogic:boolean}.

\paragraph*{Case 2: when $U=X$}
The clique $\mathcal{K}_{|U|}$ has DSM $\widehat{\mathcal{M}}$. 
The problem $\mylogic(\mathcal{K}_{|U|})$ is \Pspace-complete by
Theorem~\ref{theorem:dichotomy:mylogic:boolean} in the Boolean case;
and, beyond that, it is also \Pspace-hard as a generalisation of the
\Pspace-complete $\QCSP(\mathcal{K}_{|U|})$. The \Pspace-completeness of
$\mylogic(\mathcal{D})$ follows from Theorem~\ref{thm:galois-connectionMyLogic}.

\paragraph*{Case 3: when $U\cap X = \emptyset$}
We can no longer complete the monoid $\widehat{\mathcal{M}}$ into $\mathcal{M}'$, as we
would end up with a trivial monoid.
The remainder of this section is devoted to a generic hardness proof.  
Assume that $|U|=j\geq 2$ and $|X|=k\geq 2$ and w.l.o.g. let
$U=\{1,2,\ldots j\}$ and $X=\{j+1,j+2,\ldots j+k\}$. Recalling that the symmetric group is generated by a transposition and a cyclic permutation, let $\widehat{\mathcal{M}}$
be the DSM given by
\[
\langle \
\resizebox{!}{.12\textwidth}{
\ensuremath{
\begin{array}{c|c}
1 & 2, j+1, \ldots, j+k \\
\hline
2 & 1, j+1, \ldots, j+k\\
\hline
3 & 3, j+1, \ldots, j+k\\
\hline
\vdots & \vdots \\ 
\hline
j & j, j+1, \ldots, j+k\\
\hline
j+1 & j+1\\
\hline
j+2 & j+2\\
\hline
j+3 & j+3\\
\hline
\vdots & \vdots \\
\hline 
j+k & j+k
\end{array}
}
}
, \
\resizebox{!}{.12\textwidth}{
\ensuremath{
\begin{array}{c|c}
1 & 2, j+1, \ldots, j+k\\
\hline
2 & 3, j+1, \ldots, j+k\\
\hline
3 & 4, j+1, \ldots, j+k\\
\hline
\vdots & \vdots \\ 
\hline
j & 1, j+1, \ldots, j+k\\
\hline
j+1 & j+1\\
\hline
j+2 & j+2\\
\hline
j+3 & j+3\\
\hline
\vdots & \vdots \\ 
\hline
j+k & j+k
\end{array}
}
}
, \
\resizebox{!}{.12\textwidth}{
\ensuremath{
\begin{array}{c|c}
1 & 1, j+1, \ldots, j+k\\
\hline
2 & 2, j+1, \ldots, j+k\\
\hline
3 & 3, j+1, \ldots, j+k\\
\hline
\vdots & \vdots \\ 
\hline
j & j, j+1, \ldots, j+k\\
\hline
j+1 & j+2 \\
\hline
j+2 & j+1\\
\hline
j+3 & j+3\\
\hline
\vdots & \vdots \\ 
\hline
j+k & j+k
\end{array}
}
}
, \
\resizebox{!}{.12\textwidth}{
\ensuremath{
\begin{array}{c|c}
1 & 1, j+1, \ldots, j+k\\
\hline
2 & 2, j+1, \ldots, j+k\\
\hline
3 & 3, j+1, \ldots, j+k\\
\hline
\vdots & \vdots \\ 
\hline
j & j, j+1, \ldots, j+k\\
\hline
j+1 & j+2\\
\hline
j+2 & j+3\\
\hline
j+3 & j+4\\
\hline
\vdots & \vdots \\ 
\hline
j+k & j+1
\end{array}
}
}
\ \rangle
.
\]

We will give a structure $\widehat{\mathcal{D}}$ such that
$\shE(\widehat{\mathcal{D}})=\widehat{\mathcal{M}}$. 
Firstly, though, given some fixed $u$ in $U$ and $x$ in $X$, let
$\mathcal{G}^{|U|,|X|}_{u,x}$ be the symmetric graph with
self-loops with domain $D=U \cup X$ such that
\begin{itemize}
\item $u$ and $x$ are adjacent;
\item The graph induced by $X$ is a reflexive clique
  $\mathcal{K}_X^{\text{ref}}$; and,
\item $U\setminus \{u\}$ and $X\setminus\{x\}$
  are related via a complete bipartite graph
  $\mathcal{K}_{|X\setminus\{x\}|,|U\setminus \{u\}|}$. 
\end{itemize}
The structure $\mathcal{G}^{|U|,|X|}_{u,x}$ and the more specific
$\mathcal{G}^{4,5}_{1,5}$ are drawn in Figure~\ref{fig:PspaceHardnessGadget}. 
\begin{figure}
  \centering
   \subfloat[$\mathcal{G}^{2,2}_{1,3}$]{\label{fig:graphH}
     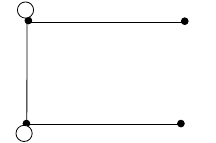}
  \quad \quad \quad
  \subfloat[$\mathcal{G}^{4,5}_{1,5}$]{\label{fig:specGadget}
    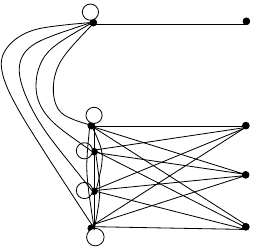}
  \quad \quad \quad
  \subfloat[$\mathcal{G}^{|U|,|X|}_{u,x}$]
  {\label{fig:GenGadget} 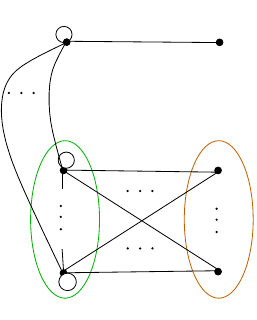}
  \caption{Main Gadget.}
  \label{fig:PspaceHardnessGadget}
\end{figure}
Denote by $E^{|U|,|X|}_{u,x}$ the binary relation of
$\mathcal{G}^{|U|,|X|}_{u,x}$ and let $\widehat{\mathcal{D}}$ be the structure
with a single $4$-ary relation $R^{\widehat{\mathcal{D}}}$ with domain $\widehat{D}=U\cup X$ specified as follows,
\begin{multline*}
R^{\widehat{\mathcal{D}}} := 
\bigcup_{u \in U} 
\Biggl(\,
\biggl(\,
\bigcup_{x\in X} (u,x)\times E^{|U|,|X|}_{u,x} \, 
\biggr)
\cup %\\
\,
\biggl(
\bigcup_{x_1,x_2,x_3\in X} (x_1,x_2)\times
E^{|U|,|X|}_{u,x_3}
\biggr)
\Biggr).
\end{multline*}
Essentially, when the first argument in a quadruple is from $U$, then
the rest of the structure allows for the unique recovery of some
$\mathcal{G}^{|U|,|X|}_{u,x}$; but if the first argument is from $X$
then all possibilities from $X$ for the remaining arguments are
allowed. In particular, we note from the last big cup that 
$(x_1,x_2,x_3,x_4)$ is a tuple of $R^{\widehat{\mathcal{D}}}$ for all quadruples
$x_1,x_2,x_3,x_4$ in $X$. 
\begin{lemma}
$\shE(\widehat{\mathcal{D}})=\widehat{\mathcal{M}}$.
\end{lemma}
\begin{proof}
Recall that, according to Theorem~\ref{theo:3PermutedForm} and our assumption on
$U$, $X$ and $\widehat{\mathcal{M}}$, a maximal (w.r.t. sub-shop inclusion) shop
$f'$ is of the following form,
\begin{itemize}
  \item for any $x$ in $X\setminus U=X$, $f(x)=\{\chi(x)\}$; and,
  \item for any $u$ in $U\setminus X=U$, $f(u)=\{\upsilon(u)\}\cup X$.
  \end{itemize}
where $\chi$ and $\upsilon$ are permutations of $X$ and $U$,
respectively.

(Backwards; $\widehat{\mathcal{M}} \subseteq \shE(\widehat{\mathcal{D}})$.) 
It suffices to check that a maximal shop $f'$ in $\widehat{\mathcal{M}}$
preserves $\widehat{\mathcal{D}}$. This holds by construction. 
We consider first tuples from $(x_1,x_2)\times E^{|U|,|X|}_{u,x_3}$. 
\begin{itemize}
\item A tuple with elements from $X$ only will map to a like tuple,
  which must occur, so we can ignore such tuples from now on.
\item 
  A tuple $(x_1,x_2,u,x_3)$ maps either to
  $\bigl(\chi(x_1),\chi(x_2),\upsilon(u),\chi(x_3)\bigr)$ which appears in
  $\bigl(\chi(x_1),\chi(x_2)\bigr)\times
  E^{|U|,|X|}_{\upsilon(u),\chi(x_3)}$, 
  or it maps to a tuple containing only elements from $X$. 
%\item A tuple $(x_1,x_2,x_4,u')$ where $x_4\neq x_3$ and $u'\neq u$
%  maps either to
%  $\bigl(\chi(x_1),\chi(x_2),\chi(x_4),\upsilon(u)\bigr)$, 
%  which appears in
%  $\bigl(\chi(x_1),\chi(x_2)\bigr)\times
%  E^{|U|,|X|}_{\upsilon(u),\chi(x_4)}$, or it maps
%  to a tuple containing only elements from $X$.
\item A tuple $(x_1,x_2,x_3,u)$  maps either to $\bigl(\chi(x_1),\chi(x_2),\chi(x_3),\upsilon(u)\bigr)$,  which appears in  $\bigl(\chi(x_1),\chi(x_2)\bigr)\times  E^{|U|,|X|}_{\upsilon(u),\chi(x_3)}$, or it maps to a tuple containing only elements from $X$.
\end{itemize}
We consider now tuples from $(u,x)\times E^{|U|,|X|}_{u,x}$. 
\begin{itemize}
\item  If the first coordinate $u$ is mapped to $\upsilon(u)$, then the
  tuple is mapped to different tuples from
  $\bigl(\upsilon(u),\chi(x)\bigr)\times
  E^{|U|,|X|}_{\upsilon(u),\chi(x)}$, depending whether the second $u$
  is mapped to an element from $X$ or to $\upsilon(u)$.
\item Otherwise, the first coordinate $u$ is mapped to an element
  $x_1$ from $X$, and some other element from $u'$ in $U$ occurs (or
  the tuple contains elements from $X$ only) and a tuple is mapped to a
  tuple of the form
  $(x_1,\chi(x),\upsilon(u'),x_3)$ which appears in
  $\bigl(x_1,\chi(x)\bigr)\times E^{|U|,|X|}_{\upsilon(u'),x_3}$.
\end{itemize}

(Forwards; $\shE(\widehat{\mathcal{D}}) \subseteq \widehat{\mathcal{M}}$.)
 We proceed by contraposition, demonstrating that $R^{\widehat{\mathcal{D}}}$ is
 violated by any $f \notin \widehat{\mathcal{M}}$. We consider the different
 ways that $f$ might not be in $\widehat{\mathcal{M}}$.  
 \begin{itemize}
 \item If $f$ is s.t. $u \in f(x)$ for $x\in X$ and $u\in U$ then we,
   e.g., take $(u,x,x,x) \in R^{\widehat{\mathcal{D}}}$ but $(z,u,u,u) \notin
   R^{\widehat{\mathcal{D}}}$ (for any $z \in f(u)$) and we are done. It follows
   that $f(X)=X$.
 \item 
   Assume now that $f$ is s.t. $\{x'_1,x'_2\} \subseteq f(x)$ for
   $x'_1 \neq x'_2$ and $x,x'_1,x'_2 \in X$. Let $u,u' \in U$ be
   s.t. $u' \in f(u)$. Take $(u,x,u,x) \in R^{\widehat{\mathcal{D}}}$;
   $(u',x'_1,u',x'_2) \notin R^{\widehat{\mathcal{D}}}$ and we are done. It
   follows that $f$ is a permutation $\chi$ on $X$.
 \item 
   Assume now that $f$ is s.t. $\{u'_1,u'_2\} \subseteq f(u)$ for
   $u'_1 \neq u'_2$ and $u,u'_1,u'_2 \in U$. Let $x,x' \in X$ be
   s.t. $x' \in f(x)$. Take $(u,x,u,x) \in R^{\widehat{\mathcal{D}}}$;
   $(u'_1,x',u'_2,x') \notin R^{\widehat{\mathcal{D}}}$ and we are done. It
   follows that $f$ restricted to $U$ is a permutation $\upsilon$ on
   $U$.
 \end{itemize}
Hence, $f$ is a sub-shop of a maximal shop $f'$ from the DSM $\widehat{\mathcal{M}}$,
and $f$ belongs to $\widehat{\mathcal{M}}$ (recall that a DSM is closed under
sub-shops). The result follows. 
\end{proof}

\begin{proposition}
$\mylogic(\widehat{\mathcal{D}})$ is Pspace-complete.
\end{proposition}
\begin{proof}
We begin with the observation that $\mylogic(\mathcal{G}^{|U|,|X|}_{u,x})$
is \Pspace-complete (for each $u\in U$ and $x\in X)$. This follows
straightforwardly from the Pspace-completeness of
$\mylogic(\mathcal{G}^{2,2}_{1,3})$, the simplest gadget which is
depicted on Figure~\ref{fig:graphH}. These gadgets
$\mathcal{G}^{|U|,|X|}_{u,x}$ agree on all equality-free
sentences -- even ones involving negation -- by
Proposition~\ref{prop:EqFreeFoContainment}, as there is a full surjective
homomorphism from $\mathcal{G}^{|U|,|X|}_{u,x}$ to
$\mathcal{G}^{2,2}_{1,3}$.

We will prove that $\mylogic(\mathcal{G}^{2,2}_{1,3})$ is
\Pspace-hard, by reduction from the \Pspace-complete problem
$\QCSP(\mathcal{B}_{\mbox{\textsc{nae}}})$. Recall that we may assume 
w.l.o.g. that universal variables are relativised to $U$ and that
existential variables are relativised to $X$, by
Theorem~\ref{relativization}. 
Let $\varphi$ be an instance of $\QCSP(\mathcal{B}_{\mbox{\textsc{nae}}})$. We reduce
$\varphi$ to a (relativised) instance $\psi$ of
$\mylogic(\mathcal{G}^{2,2}_{1,3})$. The reduction goes as follows:
\begin{itemize}
\item an existential variable $\exists x$ of $\varphi$ is replaced by
  an existential variable $\exists v_x \in X$ in $\psi$;
\item a universal variable $\forall u$ of $\varphi$ is replaced by
  $\forall u \in U\ \exists v_u \in X,\,\, E(u,v_u)$ in $\psi$; and,
\item every clause $C_i:= R(\alpha,\beta,\gamma)$ in $\varphi$ is
  replaced by the following formula in $\psi$,
  $$\forall c_i \in U,\,\, E(c_i,v_\alpha) \lor E(c_i,v_\beta)
  \lor E(c_i,v_\gamma).$$
\end{itemize}
The truth assignment is read from $\exists$ choices in $X$ for the variables
$v$: we arbitrarily see one value in $X$ as true and the other as
false. It is not relevant which one is which for the problem
not-all-equal satisfiability, we only need to ensure that no three
variables involved in a clause can get the same value. 
The $\forall c_i \in U$ acts as a conjunction, enforcing 
``one of $v_\alpha, v_\beta,v_\gamma$ is true'' and  
``one of $v_\alpha, v_\beta,v_\gamma$ is false''.  
This means that at least one in three has a different value.

Now, we can prove that $\mylogic(\widetilde{\mathcal{D}})$ is Pspace-complete by
substituting $R(u_0,x_0,u,v)$ for each instance of $E(u,v)$ in the
previous proof, and by quantifying the sentence so-produced with the
prefix $\forall u_0 \in U \ \exists x_0 \in X$, once $u_0$ and $x_0$ are
chosen, play proceeds as above but in the copy
$\mathcal{G}^{|U|,|X|}_{u_0,x_0}$, and the result
follows. 
\end{proof}

%%%%%%%%%%%%%%%%%%%%%%%%%%%%%%%%%%%%%%%%%%%%%%%%%%%%%%%%%%%%%%%%%
%\input{metaproblem}
\subsection{The Complexity of the Meta-Problem}
\label{sec:meta-problem}
The \mylogic($\sigma$) meta-problem takes as input a finite
$\sigma$-structure $\mathcal{D}$ and answers L, NP-complete,
co-NP-complete or Pspace-complete, 
according to the complexity of $\{\exists, \forall, \wedge,\vee \}
\mbox{-}\mathsf{\FO}(\mathcal{D})$. The principle result of this
section is that this problem is \NP-hard even for some fixed and
finite signature $\sigma_0$, which consists of two binary and
  three unary predicates (the unaries are for convenience, but it is not clear whether a single binary suffices).

Note that one may determine if a given shop $f$ is a surjective hyper-endomorphism of a
structure $\mathcal{D}$ in, say, quadratic time in $|D|$ 
Since we are not interested here in distinguishing levels
within P, we will henceforth consider such a test to be a basic
operation. We begin with the most straightforward case. 
\begin{proposition}
\label{prop:easy-L}
On input $\mathcal{D}$, the question ``\emph{is \mylogic-($\mathcal{D}$) in L?}'' is
in P. 
\end{proposition}
\begin{proof}
By Theorem~\ref{tetrachotomy}, we need to check whether there is both  an A-shop and an
E-shop in $\shE(\mathcal{D})$. In this special case, it suffices to test for each $u,x$ in $D$, if the
following $\{u\}$-$\{x\}$-shop $f$ preserves $\mathcal{D}$:
%(this shop was denoted by $\forall_u\exists_x$ in~\cite{LICS09,TOCL10}):
$f(u):= D$ and $f^{-1}(x):=D$.
\end{proof}
\begin{proposition}
\label{prop:hard}
For some fixed and finite signature $\sigma_0$, on input of a
$\sigma$-structure $\mathcal{D}$, the question ``is $\{\exists,
\forall, \wedge,\vee \} \mbox{-}\mathsf{\FO}(\mathcal{D})$ in NP
(respectively, NP-complete, in co-NP, co-NP-complete)?'' is
NP-complete. 
\end{proposition}
\begin{proof}
The four variants are each in NP. For the first, one guesses and
verifies that $\mathcal{D}$ has an A-shop, for the second, one further
checks that there is no $\{u\}$-$\{x\}$-shop (see the proof of
Proposition~\ref{prop:easy-L}). Similarly for the third, one guesses
and verifies that $\mathcal{D}$ has an E-shop; and, for the fourth, one
further checks that there is no $\{u\}$-$\{x\}$-shop. The result
then follows from Theorem~\ref{tetrachotomy}. 

For NP-hardness we will address the first problem only. The same
proof will work for the second (for the third and fourth, recall that
a structure $\mathcal{D}$ has an A-shop iff its complement
$\overline{\mathcal{D}}$ has an E-shop). We reduce from
\emph{graph $3$-colourability}. Let $\mathcal{G}$ be an undirected graph with
vertices $V:=\{v_1,v_2,\ldots,v_s\}$. We will build a structure $\mathcal{S}_\mathcal{G}$ over
the domain $D$ which consists of the disjoint union of ``three
colours'' $\{0,1,2\}$, $u$, and the ``vertices'' from $V$.

The key observation is that there is a structure $\mathcal{G}_{V}$
whose class of surjective hyper-endomorphisms $\shE(\mathcal{G}_{V})$ is generated by the
following A-shop: 
$$f_{V}:=
\resizebox{!}{1.5cm}{
\ensuremath{
\begin{array}{c|l}
0 & 0 \\
\hline
1 & 1 \\
\hline
2 & 2 \\
\hline
u & 0,1,2,u,v_1,\ldots,v_s \\
\hline
v_1 & 0,1,2 \\
\hline
v_2 & 0,1,2 \\
\hline
\vdots & \vdots \\
\hline
v_s & 0,1,2 \\
\end{array}
}
}$$
The existence of such a $\mathcal{G}_{V}$ is in fact guaranteed by the
Galois connection, fully given in~\cite{DBLP:conf/cp/Martin10}, but
that may require relations of unbounded arity, and we wish to
establish our result for a fixed signature. So we will appeal to
Lemma~\ref{lem:Cs}, below, for a $\sigma_V$-structure $\mathcal{G}_{V}$
with the desired class of surjective hyper-endomorphisms, where the signature $\sigma_V$ consists of
one binary relation and three monadic predicates.
The signature $\sigma_0$ will be $\sigma_V$ together with a binary
relational symbol $E$. 

The structure ${\mathcal{S}_\mathcal{G}}$ is defined as in
$\mathcal{G}_{V}$ for symbols in $\sigma_V$, and for the additional
binary symbol $E$, as the edge relation of the instance $\mathcal{G}$
of $3$-colourability together with a clique $\mathcal{K}_3$ for the
colours $\{0,1,2\}$. By construction, the following holds. 
\begin{itemize}
\item Any surjective hyper-endomorphism $g$ of ${\mathcal{S}_\mathcal{G}}$ will be a sub-shop of
  $f_V$.
\item 
  Restricting such a shop $g$ to $V$ provides a set of mutually consistent
  $3$-colourings: i.e. we may pick arbitrarily a colour from $g(v_i)$ to get
  a $3$-colouring $\tilde g$. 
  If there is an edge between $v_i$ and $v_j$ in
  $\mathcal{G}$, then $E(v_i,v_j)$ holds in
  ${\mathcal{S}_\mathcal{G}}$. Since $g$ is a shop, for any pair of
  colours $c_i,c_j$, where $c_i\in g(v_i)$ and $c_j\in g(v_j)$, we
  must have that $E(c_i,c_j)$ holds in ${\mathcal{S}_\mathcal{G}}$.
  The relation $E$ is defined as $\mathcal{K}_3$ over the colours.
  Hence $c_i \neq c_j$ and we are done.
\item Conversely, a $3$-colouring $\tilde g$ induces a sub-shop $g$ of $f_V$:
  set $g$ as $f_V$ over elements from $\{0,1,2,u\}$ and as $\tilde g$
  over $V$. The detailed argument is similar to the above.
\end{itemize}
This proves that graph $3$-colourability reduces to the meta-question 
``is $\{\exists,\forall, \wedge,\vee \}
\mbox{-}\mathsf{\FO}(\mathcal{D})$ in NP''.
\end{proof}
\noindent Note that it follows from the given proof that the
meta-problem itself is NP-hard. To see this, we take the structure
$\mathcal{S}_{\mathcal{G}}$ from the proof of Proposition~\ref{prop:hard}
and ask which of the four classes L, NP-complete,
co-NP-complete or Pspace-complete the corresponding problem belongs
to. If the answer is NP-complete then $\mathcal{G}$ was $3$-colourable;
otherwise the answer is Pspace-complete and $\mathcal{G}$ was not
$3$-colourable.
\begin{lemma}
\label{lem:Cs}
Let $\sigma_V$ be a signature involving one binary relations $E'$ and
three monadic predicates $\mathrm{Zero},\mathrm{One}$ and $\mathrm{Two}$. 
There is a $\sigma_V$-structure $\mathcal{G}_{V}$ such that
$\shE(\mathcal{G}_{V}) = \langle f_{V} \rangle$.  
\end{lemma}
\begin{proof}
  We begin with the graph $\mathcal{G}'$ on signature $\langle E' \rangle$, depicted
  on Figure~\ref{fig:WithSpecificAshopQuotient}.
  Note that 
  $$\shE(\mathcal{G}):=\langle 
  \resizebox{!}{.5cm}{
    \ensuremath{
      \begin{array}{c|c}
        c & c \\
        \hline
        u & c,u,v \\
        \hline
        v & c
      \end{array}
    }
  } \rangle.$$
  \begin{figure}[h]
  \centering
   \subfloat[$\mathcal{G}'$]{\label{fig:WithSpecificAshopQuotient} 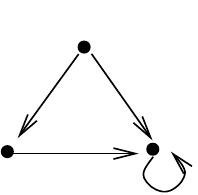}
  \quad \quad \quad
   \subfloat[$\mathcal{G}''$]{\label{fig:WithSpecificAshop}
     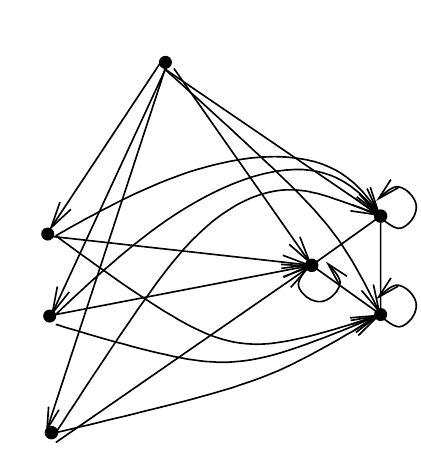
   }
  \caption{Building a structure with $\shE(\mathcal{G}_{V}) = \langle f_{V} \rangle$.}
  \label{fig:toto}
\end{figure}
We now replace $c$ by $\{0,1,2\}$ and $v$ by $V$ to obtain a graph $\mathcal{G}''$.
Formally, this graph is the unique graph with domain
$\{0,1,2,u\}\cup V$ such that the mapping  which maps $\{0,1,2\}$ to
$c$, fixes $u$ and maps $V$ to $v$, is a strong surjective
homomorphism. By construction, 
$$\shE(\mathcal{G}''):=\langle 
\resizebox{!}{1.5cm}{
\ensuremath{
\begin{array}{c|l}
0 & 0,1,2 \\
\hline
1 & 0,1,2 \\
\hline
2 & 0,1,2 \\
\hline
u & 0,1,2,u,v_1,\ldots,v_s \\
\hline
v_1 & 0,1,2 \\
\hline
\vdots & \vdots \\
\hline
v_s & 0,1,2 \\
\end{array}
}
}
\rangle.$$
We now build $\mathcal{G}_{V}$ as the structure with binary relation
$E'$ which is the edge relation from $\mathcal{G}''$ and by setting the unary predicates as follows: $\mathrm{Zero}$ holds only over $0$, $\mathrm{One}$ holds only over
$1$ and $\mathrm{Two}$ holds only over $2$. This effectively fixes surjective hyper-endomorphisms to
act as the identity over the colours $\{0,1,2\}$ as required.
\end{proof}

\begin{sidewaysfigure}
  \centering
  \label{fig:ModelChecking}
    \begin{tikzpicture}[node distance = 2cm, auto,scale=0.65, every
      node/.style={transform shape}]
    \tikzstyle{silly} = [rectangle, rounded corners,
    text centered, thick, rounded corners,
    draw=yellow!50!black!50] 
    \tikzstyle{stuff} = [rectangle, rounded corners,
    text centered, thick, rounded corners,
    draw=yellow!50!black!50, top color=white, bottom
    color=yellow!50!black!20] 
    \tikzstyle{booleanlike} = [rectangle,
    rounded corners,
    text centered, thick, rounded corners,
    draw=green!90!black!50, top color=white, bottom
    color=green!50!black!20]
    \tikzstyle{tetralike} = [rectangle,
    rounded corners,
    text centered, thick, rounded corners,
    draw=orange!90!black!50, top color=white, bottom
    color=orange!50!black!20]
    \tikzstyle{csplike} = [rectangle,
    rounded corners,
    text centered, thick, rounded corners,
    draw=blue!90!black!50, top color=white, bottom
    color=blue!50!black!20]
    \tikzstyle{qcsplike} = [rectangle,
    rounded corners,
    text centered, thick, rounded corners,
    draw=purple!90!black!50, top color=white, bottom
    color=purple!50!black!20]

    \tikzstyle{line} = [draw, very thin, loosely dashed] 

    % Place node
    \node[tetralike] (PosEqFree) {$\mylogic$};
    \node[qcsplike,below left of =PosEqFree, node distance = 5cm] (Qcsp) {$\qcsplogic$};
    \node[csplike,below left of =Qcsp, node distance = 8cm] (Csp) {$\csplogic$};
    \node[stuff,above left of =Csp, node distance = 4cm] 
    (CspDisj) {$\{\exists, \land, \lor \} \mbox{-}\mathrm{FO}$};
    \node[silly, left of = Csp, node distance = 10cm] 
    (SillyCase) {$\{\exists, \lor \} \mbox{-}\mathrm{FO}$};

    \node[stuff,above of =CspDisj, node distance=1.5cm] 
    (CspDisjEq) {$\{\exists, \land, \lor, = \} \mbox{-}\mathrm{FO}$};
    \node[stuff,left of =CspDisjEq, node distance = 3.5cm] 
    (CspDisjNeq) {$\{\exists, \land, \lor, \neq \} \mbox{-}\mathrm{FO}$};

    \node[silly,above of =SillyCase, node distance=1.5cm] 
    (SillyCaseEq) {$\{\exists, \lor, = \} \mbox{-}\mathrm{FO}$};
    % \coordinate(BelowSillyCaseEqNeq) at ($(SillyCaseEq)+ (-1.5cm,0cm)$);
    % \node[silly,above of =BelowSillyCaseEqNeq, node distance=1.5cm] 
    % (SillyCaseEqNeq) {\sillyEqNeq};
    \node[silly,left of =SillyCaseEq, node distance = 3.5cm] 
    (SillyCaseNeq)  {$\{\exists, \lor, \neq \} \mbox{-}\mathrm{FO}$};

    \node[csplike,above of =Csp, node distance = 1.5cm]
    (CspEq) {$\{\exists, \land, = \} \mbox{-}\mathrm{FO}$};
    %\node[booleanlike,above of =Csp, node distance=2.5cm] 
    %\node[booleanlike,right of =Csp, node distance=3.5cm,yshift=.8cm] 
    \node[booleanlike,left of =CspEq, node distance=6.4cm] 
    (CspNeq) {$\{\exists, \land, \neq \} \mbox{-}\mathrm{FO}$};

    \node[qcsplike,above of =Qcsp, node distance = 1.5cm] 
    (QcspEq) {$\{\forall, \exists, \land, = \} \mbox{-}\mathrm{FO}$};
    \node[booleanlike,left of =QcspEq, node distance = 3.5cm] 
    (QcspNeq) {$\{\forall, \exists, \land, \neq \} \mbox{-}\mathrm{FO}$};

    \node[stuff,left of =PosEqFree, node distance = 4cm,yshift=5cm] 
    (PosFoEq) {$\{\forall, \exists, \land, \lor, = \} \mbox{-}\mathrm{FO}$};
    \node[stuff,right of =PosEqFree, node distance = 4cm,yshift=5cm] 
    (PosFoNeq){$\{\exists, \forall,  \lor, \land, \neq \} \mbox{-}\mathrm{FO}$};

    \node[stuff,above of =PosEqFree, node distance = 2.5cm] 
    (EqFree){$\{\forall, \exists, \land, \lor, \lnot \} \mbox{-}\mathrm{FO}$};

    \node[stuff,above of =EqFree, node distance = 4cm] (FO)
    %{$\{\forall, \exists, \land, \lor, \lnot,= \} \mbox{-}\mathrm{FO}$};
    {$\mathrm{FO}$};

     \path [line] (SillyCase) -- (SillyCaseEq);
     % \path [line] (SillyCaseEq) -- (SillyCaseEqNeq);
     % \path [line] (SillyCaseNeq) -- (SillyCaseEqNeq);
     \path [line] (SillyCase) -- (SillyCaseNeq);
     % \path [line] (SillyCase) -- (CspDisj);
     \path [line] (SillyCase) edge [out=45,in=180,looseness=1.3] (CspDisj);
     %\path [line] (SillyCaseEq) edge [out=90,in=-90,looseness=0.6] (CspDisjEq);

     \path [line] (Csp) -- (Qcsp);
     \path [line] (Csp) -- (CspEq);
     \path [line] (Csp) -- (CspNeq);
     \path [line] (Csp) -- (CspDisj);

     \path [line] (CspDisj) -- (CspDisjEq);
     \path [line] (CspDisj) -- (CspDisjNeq);
     % \path [line] (CspDisj) edge [out=0,in=180,looseness=8] (PosEqFree);
     
     % \coordinate (toto) at ($(QcspNeq.west)+ (-1cm,1cm)$);
     % \coordinate (titi) at ($(CspDisjEq.east)+ (1cm,0cm)$);
     % \draw[line]  (CspDisj)-| (toto) -| (PosEqFree.west);
     % \draw[line]  (CspDisj.east) -- (titi) .. controls (toto).. (PosEqFree.west);

     %\path [line, loosely dotted] (CspDisjEq) edge [out=45,in=180,looseness=0.6] (PosFoEq);

     \path [line] (Qcsp) -- (QcspEq);
     \path [line] (Qcsp) -- (QcspNeq);
     \path [line] (Qcsp) -- (PosEqFree); 

     \path [line] (PosEqFree) -- (EqFree); 
     \path [line] (PosEqFree) -- (PosFoEq); 
     \path [line] (PosEqFree) -- (PosFoNeq); 

     \path [line] (EqFree) -- (FO); 
 
     %%% border ?
     % \draw  ([xshift=1cm,yshift=-2cm] SillyCase) -- ([xshift=1cm] CspDisj) -- ([xshift=1cm]PosEqFree) ;
     %\draw[transform canvas ={yshift=-1cm}]  (SillyCase) -- (CspDisj) -- (PosEqFree) ;

     \node[draw=black,inner
     sep=10pt,thick,rectangle,fit=(SillyCase)(SillyCaseEq)(SillyCaseNeq)]%(SillyCaseEqNeq)]
     (Sillycases){};
     \node at (Sillycases.south) [above, inner sep=1mm,xshift=-1.2cm]
     {\textcolor{black}{first class: \textbf{trivial}: \Logspace}};

     \node[draw=blue!50!black!80,inner sep=3pt,thick, rectangle,fit=(Csp) (CspEq)] (cspcases){};
     \node at (cspcases.south) [below, inner sep=1mm]%,xshift=-1.2cm]
     {\textcolor{blue!50!black!80}{\textbf{\CSP\ dichotomy?}}};

     \node[draw=purple!50!black!80,inner sep=3pt,thick,
     rectangle,fit=(Qcsp) (QcspEq)] (qcspcases){};
     \node at (qcspcases.south) [below, inner sep=1mm]%,xshift=-1.2cm]
     {\textcolor{purple!50!black!80}{\textbf{\QCSP\ trichotomy?}}};
    
     \node at (CspNeq.south) [below, inner sep=1mm]
     {\textcolor{green!50!black!80}{``\textbf{Boolean \CSP}''}};

     \node at (QcspNeq.south) [below, inner sep=1mm]
     {\textcolor{green!50!black!80}{``\textbf{Boolean \QCSP}''}};

     \node at (PosEqFree.south) [below, inner sep=1mm]
     {\textcolor{orange!50!black!80}{\textbf{tetrachotomy}}};

     \node[draw=black, loosely dashed, inner sep=18pt,thin,rectangle,
     fit=(Csp) (CspEq) (CspDisj) (CspDisjEq)] 
     (corecases){};
     \node at (corecases.north) [below, inner sep=1mm]
     {\textcolor{black}{\textbf{$\mathscr{L}$-core}: (classical) core}};

     \node[%draw=black, loosely dashed, 
     inner sep=18pt,thin,rectangle,
     fit=(FO) (PosFoEq) (PosFoNeq)] 
     (isocases){};
     \node at (isocases.north) [below, inner sep=1mm]
     {\textcolor{black}{\textbf{$\mathscr{L}$-equivalence}: isomorphism}};

     \coordinate (EqFreeWiderTopLeft) at ($(EqFree.north west)+(-2cm,0cm)$);
     \coordinate (EqFreeWiderBorRight) at ($(EqFree.south east)+ (+2cm,0cm)$);
     \node[draw=black, loosely dashed, inner sep=18pt,thin,rectangle,
     fit=(EqFreeWiderTopLeft) (EqFree)(EqFreeWiderBorRight)] 
     (EqFreeCase){};
     \node at (EqFreeCase.north) [below, inner sep=1mm]
     {\textcolor{black}{\textbf{$\mathscr{L}$-equivalence}: full
         surjective hyper-morphism}};

     \coordinate (PosEqFreeWiderTopLeft) at ($(PosEqFree.north west)+(-1cm,0cm)$);
     \coordinate (PosEqFreeWiderBorRight) at ($(PosEqFree.south east)+ (+1cm,0cm)$);
     \node[draw=black, loosely dashed, inner sep=18pt,thin,rectangle,
     fit=(PosEqFreeWiderTopLeft) (PosEqFree)(PosEqFreeWiderBorRight)] 
     (PosEqFreeCase){};
     \node at (PosEqFreeCase.north) [below, inner sep=1mm]
     {\textcolor{black}{\textbf{$\mathscr{L}$-core}: $U$-$X$-core}};

     %%% A little legend to explain where to find the four classes.

     \coordinate (Legend) at ($(FO)+(-22.5cm,0cm)$);
     \node[below of = Legend, node distance=0cm] (LegendTitle) {\textbf{Legend}};

     \node[silly, below of = LegendTitle, node distance = 1cm] (classOne){first class};
     
     \node[stuff, below of = classOne, node distance=1.5cm]
     (classTwo){second class};

     \node[booleanlike, below of = classTwo, node distance=1.5cm]
     (classThreeBool){third class};
     \node[csplike, below of = classThreeBool, node distance=.8cm]
     (classThreeCsp){third class};
     \node[qcsplike, below of = classThreeCsp, node distance=.8cm]
     (classThreeQcsp){third class};
     
     \node[tetralike, below of = classThreeQcsp, node distance=1.5cm]
     (classFour){fourth class};

     \coordinate (classOnePos) at ($(classOne.east)+(+6cm,0cm)$);
     %\node at (classOne.east) [red] {$+$};
     \node at (classOnePos) [text width=10cm]
     (classOneText){Alway trivial (in \Logspace).};

     \node at ($(classOneText.west)+(5.1cm,-1.5cm)$) [text width=10cm]
     (classTwoText)
     {Trivial complexity delineation: \\
       trivial if the $\mathscr{L}$-core has a single element, hard otherwise.};

     \node at ($(classTwoText.west)+(+5.1cm,-1.5cm)$) 
     [text width=10cm]
     (classThreeBoolText)
      {Non-trivial complexity delineation in the Boolean case,\\ hard with three or more elements.};

     \node at ($(classThreeBoolText.west)+(+5.1cm,-0.9cm)$) 
     [text width=10cm]
     (classThreeCspText)
     {Dichotomy between \Ptime{} and \NP-complete? Does not depend
       on core size.};

     \node at ($(classThreeCspText.west)+(+5.1cm,-0.9cm)$) 
     [text width=10cm]
     (classThreeQcspText)
     {Trichotomy between \Ptime{}, \NP-complete and \Pspace-complete?
       Does not depend on $\qcsplogic$-core size.};

     \node at ($(classThreeQcspText.west)+(+5.1cm,-1.5cm)$) 
     [text width=10cm]
     (classFourText)
     {The complexity follows a tetrachotomy according to the
       $U$-$X$-core and whether one or both of $U$ and $X$ has a
       single element or not.}; 

     \node[draw=black,inner sep=10pt,thick,rectangle,fit=(Legend)(classOne)(classTwo)(classThreeBool)(classThreeCsp)(classThreeQcsp)(classFour)(classOneText)(classTwoText)(classThreeBoolText)(classThreeCspText)(classThreeQcspText)(classFourText)](BoxLegend){};

    \end{tikzpicture}
    \caption{Classification of the complexity of the model-checking
      problem}
\end{sidewaysfigure}

\section{Conclusion}
We have classified the complexity of the model checking problem for
all fragments of \FO\ but those corresponding to the \CSP\ and the \QCSP.
Our results are summarised as Figure~\ref{fig:ModelChecking}. 
%%% TO DO explain picture
The inclusion of fragments is denoted by dashed edges, a larger
fragment being above. Each fragment is classified in two
fashions. Firstly, we have indicated on the figure the notion of
core used to classify fragments, by regrouping them in the same box. 
Secondly, we have organised the fragments in four classes according to the
nature of the complexity classification they follow. The first class
is trivial. Tractability for a fragment $\mathscr{L}$ of the second class corresponds
precisely to having a one element $\mathscr{L}$-core. The third class regroups
fragments which have a non trivial classification viz complexity, in
the sense that it does not always depend on the size of the $\mathscr{L}$-core, and
include the two open cases of \CSP\ and \QCSP\ which we discuss in
some detail below. The fourth class contains the fragment \mylogic\
which exhibits a behaviour intermediate between the third class and
the fourth class: its complexity is fully explained in terms of the
$U$-$X$-core, yet as this notion involves two sets, the fragment exhibits richness in its ensuing tetrachotomy.

For the \CSP, the dichotomy conjecture has been proved in the Boolean
case by Schaefer (see Theorem~\ref{theorem:SchaeferCSP}) and in the case of
undirected graphs.
\begin{theorem}[\cite{HellNesetril}]
  Let $\mathcal{G}$ be an undirected graph.
  If $\mathcal{G}$ is bipartite then $\CSP(\mathcal{G})$ is in
  \Logspace, otherwise $\CSP(\mathcal{G})$ is \NP-complete.\footnotemark{}
\end{theorem}
\footnotetext{In the bipartite case, assuming that the graph
  $\mathcal{G}$ has at least one edge, then the core of $\mathcal{G}$
  is $\mathcal{K}_2$. The problem $\CSP(\mathcal{K}_2)$ is
  2-colourability which is in the complexity class \emph{symmetric
    logspace} now known to be equal to \Logspace~\cite{DBLP:journals/jacm/Reingold08}.} 

For \CSP\ in general, it would suffice to settle the dichotomy
conjecture for (certain) directed
graphs~\cite{DBLP:journals/siamcomp/FederV98}. 
The dichotomy conjecture has been settled for smooth digraphs (graphs
with no sources and no sinks)~\cite{DBLP:journals/siamcomp/BartoKN09}.
According to the algebraic reformulation of the dichotomy conjecture, it would
suffice to prove that every structure that has a \emph{Sigger's term} has a
tractable \CSP\ (see~\cite{DBLP:conf/dagstuhl/BulatovV08,DBLP:conf/csr/Bulatov11} for
recent surveys on the algebraic approach to the dichotomy conjecture).

For the \QCSP, much less is known. We have already seen that a
dichotomy between \Ptime\ and \Pspace-complete holds in the Boolean
case (Theorem~\ref{theorem:SchaeferQCSP}). However, the complexity is not even known
for undirected graphs. It is fully classified for graphs with at most
one cycle. 
\begin{theorem}%[\textbf{partial trichotomy for undirected graphs}~\cite{DBLP:conf/cie/MartinM06}]\
  [\cite{DBLP:conf/cie/MartinM06}]
  Let $\mathcal{G}$ be an undirected graph.
  \begin{itemize}
  \item If $\mathcal{G}$ is bipartite then $\QCSP(\mathcal{G})$ is in \Logspace;
  \item if $\mathcal{G}$ is not bipartite and not connected then
    $\QCSP(\mathcal{G})$ is \NP-complete; and,
  \item if $\mathcal{G}$ not bipartite, connected and contains at most
    one cycle then $\QCSP(\mathcal{G})$ is \Pspace-complete.
  \end{itemize}
\end{theorem}

The algebraic approach to \QCSP\ uses \emph{surjective polymorphisms}
and has led to a trichotomy in the case where all graphs of permutations are
available. Recall first the definition of some special surjective
operations. A $k$-ary \emph{near-unanimity operation} $f$ satisfies 
$$f(x_1,\ldots,x_k)=
  \begin{cases}
    x & \text{if $\{x_1,\ldots,x_k\}= \{x\}$; and,}\\
    x & \text{if all but one of $x_1,\ldots,x_k$ is equal to $x$.}
  \end{cases}
$$
When $k=3$, we speak of a \emph{majority operation}.
The \emph{$k$-ary near projection operation} is defined as
$$
l_k(x_1, \ldots, x_k )=
\begin{cases}
  x_1 & \text{when $|\{x_1,\ldots,x_k\}|=k$; and,}\\
  x_k & \text{otherwise.}
\end{cases}
$$
The \textsl{ternary switching operation} is defined as
$$
s(x,y,z)=
\begin{cases}
  x & \text{if $y=z$,}\\
  y & \text{if $x=z$,}\\
  z & \text{otherwise.}
\end{cases}
$$
The \emph{dual discriminator operation} is defined as
$$d(x,y,z)=
  \begin{cases}
    y & \text{if $y = z$; and,} \\
    x & \text{otherwise.}
  \end{cases}
$$
When $f(x,y,z)= x - y + z$ w.r.t. some Abelian group structure,
we say that $f$ is an \emph{affine operation}.

\begin{theorem}% [\textbf{Trichotomy for structures with all
    % permutations}~\cite{DBLP:journals/iandc/BornerBCJK09}]\
  [\cite{DBLP:journals/iandc/BornerBCJK09}]
  Let $\mathcal{D}$ be a structure such that there is an extensional
  binary symbol for each graph of a permutation of $D$. Then the complexity of $\QCSP(\mathcal{D})$
  follows the following trichotomy.
  \begin{itemize}
  \item If $\mathcal{D}$ has a surjective polymorphism which is the
    dual discriminator, the switching operation or an affine operation
    then $\QCSP(\mathcal{D})$ is in \Ptime.
  \item Else, if $l_{|D|}$ is a surjective
    polymorphism of $\mathcal{D}$ then $\QCSP(\mathcal{D})$ is \NP-complete.
  \item Otherwise, $\QCSP(\mathcal{D})$ is \Pspace-complete.
  \end{itemize}
\end{theorem}

In general, it is known that if a structure $\mathcal{D}$ is preserved
by a \emph{near-unanimity operation} then $\QCSP(\mathcal{D})$ is in
\Ptime, because it implies a property of \emph{collapsibility}. This
property means that an instance holds if, and only, if all sentences
induced by keeping only a bounded number of universal quantifiers -- the
so-called \emph{collapsings} --
hold~\cite{DBLP:journals/siamcomp/Chen08}.

For undirected partially reflexive graphs (\textsl{i.e.} with possible self-loops),
we have the following partial classification (reformulated
algebraically).
\begin{theorem}[\cite{DBLP:conf/cp/Martin11}]
  Let $\mathcal{T}$ be a partially reflexive forest.
  \begin{itemize}
  \item If $\mathcal{T}$ is $\qcsplogic$-equivalent to a structure
    that is preserved by a majority operation then
    $\QCSP(\mathcal{T})$ is in \Ptime; and,
  \item otherwise, $\QCSP(\mathcal{T})$ is \NP-hard.
  \end{itemize}
\end{theorem}

In the case of structures with all constants,
Hubie Chen has ventured some conjecture regarding the
\NP/\Pspace-hard border: he suggests that the \emph{polynomially generated power
  property (PGP)} -- a property which generalises collapsibility -- explains a drop in complexity to \NP
(see~\cite{DBLP:journals/corr/abs-1201-6306} for details).

%%%%%%%%%%%%%%%%%%%%%%%%%%%%%%%%%%%%%%%%%%%%%%%%%%%%%%%%%%%%%%%%%%%%%%%%%%%%%%%%%%%%%%%%%%%%%%%%%%%%%%

\subsection*{Acknowledgment}
The authors thank Jos Martin for his enthusiasm with this project and
his technical help in providing a computer assisted proof in the four
element case~\cite{DBLP:conf/csl/MartinM10}, which was instrumental in deriving the tetrachotomy for
\mylogic\ in the general case.

% Bibliography
%\bibliographystyle{plainnat}
\bibliographystyle{alpha}
\bibliography{MMtetra}

% History dates
%\received{July 2012}{}{}

\end{document}

%% file: MMmacro.tex
\newcommand{\Nesetril}[0]{Ne\-\v{s}e\-t\v{r}il}

\newcommand{\CSP}[0]{\ensuremath{\textrm{CSP}}}

\newcommand{\QCSP}[0]{\ensuremath{\textrm{QCSP}}}

\newcommand{\Pol}[0]{\ensuremath{\mathsf{Pol}}}
\newcommand{\End}[0]{\ensuremath{\mathsf{End}}}

\newcommand{\Inv}[0]{\ensuremath{\mathsf{Inv}}}
\newcommand{\shE}[0]{\ensuremath{\mathsf{shE}}}
\newcommand{\hE}[0]{\ensuremath{\mathsf{hE}}}

\renewcommand{\phi}{\varphi}

\newcommand{\NP}[0]{\ensuremath{\mathrm{NP}}}

\newcommand{\coNP}[0]{\ensuremath{\mathrm{co\mbox{-}NP}}}

\newcommand{\Logspace}[0]{\ensuremath{\mathrm{L}}}

\newcommand{\Ptime}[0]{\ensuremath{\mathrm{P}}}
\newcommand{\Pspace}[0]{\ensuremath{\mathrm{Pspace}}}

\newcommand{\FO}[0]{\ensuremath{\mathrm{FO}}}

\newcommand{\csplogic}{\ensuremath{\{\exists, \land \} \mbox{-}\mathrm{FO}}}
\newcommand{\csplogicEq}{\ensuremath{\{\exists, \land, = \} \mbox{-}\mathrm{FO}}}
\newcommand{\csplogicNeq}{\ensuremath{\{\exists, \land, \neq \} \mbox{-}\mathrm{FO}}}
\newcommand{\CocsplogicNeq}{\ensuremath{\{\forall, \lor, = \} \mbox{-}\mathrm{FO}}}
\newcommand{\cspDisj}{\ensuremath{\{\exists, \land, \lor \} \mbox{-}\mathrm{FO}}}
\newcommand{\CocspDisj}{\ensuremath{\{\forall, \lor, \land \} \mbox{-}\mathrm{FO}}}
\newcommand{\cspDisjEq}{\ensuremath{\{\exists, \land, \lor, = \} \mbox{-}\mathrm{FO}}}
\newcommand{\cspDisjNeq}{\ensuremath{\{\exists, \land, \lor, \neq \} \mbox{-}\mathrm{FO}}}
\newcommand{\qcsplogic}{\ensuremath{\{\exists, \forall, \land \} \mbox{-}\mathrm{FO}}}
\newcommand{\qcsplogicEq}{\ensuremath{\{\exists, \forall, \land,= \} \mbox{-}\mathrm{FO}}}
\newcommand{\qcsplogicNeq}{\ensuremath{\{\exists, \forall, \land, \neq \} \mbox{-}\mathrm{FO}}}
\newcommand{\qcspDisj}{\ensuremath{\{\exists, \forall, \land, \lor \} \mbox{-}\mathrm{FO}}}

\newcommand{\mylogic}{\ensuremath{\{\exists, \forall, \land,\lor \} \mbox{-}\mathrm{FO}}}
\newcommand{\posFoEq}{\ensuremath{\{\exists, \forall, \land,\lor,= \} \mbox{-}\mathrm{FO}}}
\newcommand{\posFoNeq}{\ensuremath{\{\exists, \forall, \land,\lor,\neq \} \mbox{-}\mathrm{FO}}}

\newcommand{\EqFreeFo}{\ensuremath{\{\exists, \forall, \land,\lor,\lnot \} \mbox{-}\mathrm{FO}}}

\newcommand{\tuple}[1]{\ensuremath{\mathbf{#1}}}

\newcommand{\shee}[2]{
\resizebox{!}{.4cm}{
\ensuremath{
\begin{array}{c|c}
0 & #1 \\
\hline
1 & #2 \\
\end{array}
}
}
}

%% file: LicsSpecGadgetGraphH.pdf_t
\begin{picture}(0,0)%
\includegraphics{LicsSpecGadgetGraphH.pdf}%
\end{picture}%
%
%  Created by WinFIG version 4.0 beta 1 
%  METADATA <version>1.0</version> 
%
\setlength{\unitlength}{2072sp}%
\begingroup\makeatletter\ifx\SetFigFont\undefined%
\gdef\SetFigFont#1#2#3#4#5{%
  \reset@font\fontsize{#1}{#2pt}%
  \fontfamily{#3}\fontseries{#4}\fontshape{#5}%
  \selectfont}%
\fi\endgroup%
\begin{picture}(1968,1290)(2599,424)
%  METADATA <id>55</id> 
\put(4445,1520){\makebox(0,0)[lb]{\smash{{\SetFigFont{6}{7.2}{\rmdefault}{\mddefault}{\updefault}{\color[rgb]{.75,.38,0}1}%
}}}}
%  METADATA <id>56</id> 
\put(4430,606){\makebox(0,0)[lb]{\smash{{\SetFigFont{6}{7.2}{\rmdefault}{\mddefault}{\updefault}{\color[rgb]{.75,.38,0}2}%
}}}}
%  METADATA <id>89</id> 
\put(2614,608){\makebox(0,0)[lb]{\smash{{\SetFigFont{6}{7.2}{\rmdefault}{\mddefault}{\updefault}{\color[rgb]{0,.69,0}4}%
}}}}
%  METADATA <id>88</id> 
\put(2614,1531){\makebox(0,0)[lb]{\smash{{\SetFigFont{6}{7.2}{\rmdefault}{\mddefault}{\updefault}{\color[rgb]{0,.69,0}3}%
}}}}
\end{picture}%

%% file: LicsSpecGadget.pdf_t
\begin{picture}(0,0)%
\includegraphics{LicsSpecGadget.pdf}%
\end{picture}%
%
%  Created by WinFIG version 4.0 beta 1 
%  METADATA <version>1.0</version> 
%
\setlength{\unitlength}{2072sp}%
\begingroup\makeatletter\ifx\SetFigFont\undefined%
\gdef\SetFigFont#1#2#3#4#5{%
  \reset@font\fontsize{#1}{#2pt}%
  \fontfamily{#3}\fontseries{#4}\fontshape{#5}%
  \selectfont}%
\fi\endgroup%
\begin{picture}(2494,2246)(2004,-511)
%  METADATA <id>91</id> 
\put(2626,659){\makebox(0,0)[lb]{\smash{{\SetFigFont{6}{7.2}{\rmdefault}{\mddefault}{\updefault}{\color[rgb]{0,.69,0}6}%
}}}}
%  METADATA <id>92</id> 
\put(2626,346){\makebox(0,0)[lb]{\smash{{\SetFigFont{6}{7.2}{\rmdefault}{\mddefault}{\updefault}{\color[rgb]{0,.69,0}7}%
}}}}
%  METADATA <id>93</id> 
\put(2626,-15){\makebox(0,0)[lb]{\smash{{\SetFigFont{6}{7.2}{\rmdefault}{\mddefault}{\updefault}{\color[rgb]{0,.69,0}8}%
}}}}
%  METADATA <id>94</id> 
\put(2626,-348){\makebox(0,0)[lb]{\smash{{\SetFigFont{6}{7.2}{\rmdefault}{\mddefault}{\updefault}{\color[rgb]{0,.69,0}9}%
}}}}
%  METADATA <id>55</id> 
\put(4132,1585){\makebox(0,0)[lb]{\smash{{\SetFigFont{6}{7.2}{\rmdefault}{\mddefault}{\updefault}{\color[rgb]{.75,.38,0}1}%
}}}}
%  METADATA <id>56</id> 
\put(4376,539){\makebox(0,0)[lb]{\smash{{\SetFigFont{6}{7.2}{\rmdefault}{\mddefault}{\updefault}{\color[rgb]{.75,.38,0}2}%
}}}}
%  METADATA <id>88</id> 
\put(4376,114){\makebox(0,0)[lb]{\smash{{\SetFigFont{6}{7.2}{\rmdefault}{\mddefault}{\updefault}{\color[rgb]{.75,.38,0}3}%
}}}}
%  METADATA <id>89</id> 
\put(4376,-374){\makebox(0,0)[lb]{\smash{{\SetFigFont{6}{7.2}{\rmdefault}{\mddefault}{\updefault}{\color[rgb]{.75,.38,0}4}%
}}}}
%  METADATA <id>90</id> 
\put(2709,1544){\makebox(0,0)[b]{\smash{{\SetFigFont{6}{7.2}{\rmdefault}{\mddefault}{\updefault}{\color[rgb]{0,.69,0}5}%
}}}}
\end{picture}%

%% file: LicsGenGadget.pdf_t
\begin{picture}(0,0)%
\includegraphics{LicsGenGadget.pdf}%
\end{picture}%
%
%  Created by WinFIG version 4.0 beta 1 
%  METADATA <version>1.0</version> 
%
\setlength{\unitlength}{2072sp}%
\begingroup\makeatletter\ifx\SetFigFont\undefined%
\gdef\SetFigFont#1#2#3#4#5{%
  \reset@font\fontsize{#1}{#2pt}%
  \fontfamily{#3}\fontseries{#4}\fontshape{#5}%
  \selectfont}%
\fi\endgroup%
\begin{picture}(2322,2906)(2259,-767)
\put(4186,-698){\makebox(0,0)[b]{\smash{{\SetFigFont{6}{7.2}{\rmdefault}{\mddefault}{\updefault}{\color[rgb]{.75,.38,0}$U\setminus\{u\}$}%
}}}}
%  METADATA <id>62</id> 
\put(2870,1988){\makebox(0,0)[b]{\smash{{\SetFigFont{6}{7.2}{\rmdefault}{\mddefault}{\updefault}{\color[rgb]{0,.69,0}$x$}%
}}}}
%  METADATA <id>55</id> 
\put(4247,1992){\makebox(0,0)[b]{\smash{{\SetFigFont{6}{7.2}{\rmdefault}{\mddefault}{\updefault}{\color[rgb]{.75,.38,0}$u$}%
}}}}
\put(2836,-691){\makebox(0,0)[b]{\smash{{\SetFigFont{6}{7.2}{\rmdefault}{\mddefault}{\updefault}{\color[rgb]{0,.69,0}$X\setminus\{x\}$}%
}}}}
\end{picture}%

%% file: digraph1.pdf_t
\begin{picture}(0,0)%
\includegraphics{digraph1.pdf}%
\end{picture}%
%
%  Created by WinFIG version 3.03 
%  METADATA <version>1.0</version> 
%  Created by WinFIG version 4.0 beta 1 
%
\setlength{\unitlength}{4144sp}%
\begingroup\makeatletter\ifx\SetFigFont\undefined%
\gdef\SetFigFont#1#2#3#4#5{%
  \reset@font\fontsize{#1}{#2pt}%
  \fontfamily{#3}\fontseries{#4}\fontshape{#5}%
  \selectfont}%
\fi\endgroup%
\begin{picture}(890,879)(353,-208)
%  METADATA <id>10</id> 
\put(721,524){\makebox(0,0)[b]{\smash{{\SetFigFont{11}{13.2}{\rmdefault}{\mddefault}{\updefault}{\color[rgb]{0,0,0}$u$}%
}}}}
%  METADATA <id>12</id> 
\put(406, 29){\makebox(0,0)[rb]{\smash{{\SetFigFont{11}{13.2}{\rmdefault}{\mddefault}{\updefault}{\color[rgb]{0,0,0}$v$}%
}}}}
%  METADATA <id>13</id> 
\put(1081, 29){\makebox(0,0)[lb]{\smash{{\SetFigFont{11}{13.2}{\rmdefault}{\mddefault}{\updefault}{\color[rgb]{0,0,0}$c$}%
}}}}
\end{picture}%

%% file: digraph2.pdf_t
\begin{picture}(0,0)%
\includegraphics{digraph2.pdf}%
\end{picture}%
%
%  Created by WinFIG version 3.03 
%  METADATA <version>1.0</version> 
%  Created by WinFIG version 4.0 beta 1 
%
\setlength{\unitlength}{4144sp}%
\begingroup\makeatletter\ifx\SetFigFont\undefined%
\gdef\SetFigFont#1#2#3#4#5{%
  \reset@font\fontsize{#1}{#2pt}%
  \fontfamily{#3}\fontseries{#4}\fontshape{#5}%
  \selectfont}%
\fi\endgroup%
\begin{picture}(1915,2080)(1411,-1619)
%  METADATA <id>23</id> 
\put(1426,-599){\makebox(0,0)[lb]{\smash{{\SetFigFont{11}{13.2}{\rmdefault}{\mddefault}{\updefault}{\color[rgb]{0,0,0}$v_1$}%
}}}}
%  METADATA <id>24</id> 
\put(1426,-986){\makebox(0,0)[lb]{\smash{{\SetFigFont{11}{13.2}{\rmdefault}{\mddefault}{\updefault}{\color[rgb]{0,0,0}$v_2$}%
}}}}
%  METADATA <id>25</id> 
\put(1441,-1546){\makebox(0,0)[lb]{\smash{{\SetFigFont{11}{13.2}{\rmdefault}{\mddefault}{\updefault}{\color[rgb]{0,0,0}$v_s$}%
}}}}
%  METADATA <id>33</id> 
\put(2836,-646){\makebox(0,0)[lb]{\smash{{\SetFigFont{11}{13.2}{\rmdefault}{\mddefault}{\updefault}{\color[rgb]{0,0,0}$0$}%
}}}}
%  METADATA <id>33</id> 
\put(3151,-331){\makebox(0,0)[lb]{\smash{{\SetFigFont{11}{13.2}{\rmdefault}{\mddefault}{\updefault}{\color[rgb]{0,0,0}$2$}%
}}}}
%  METADATA <id>26</id> 
\put(3151,-1141){\makebox(0,0)[lb]{\smash{{\SetFigFont{11}{13.2}{\rmdefault}{\mddefault}{\updefault}{\color[rgb]{0,0,0}$1$}%
}}}}
%  METADATA <id>28</id> 
\put(2126,314){\makebox(0,0)[lb]{\smash{{\SetFigFont{11}{13.2}{\rmdefault}{\mddefault}{\updefault}{\color[rgb]{0,0,0}$u$}%
}}}}
\end{picture}%